\documentclass[a4paper,11pt]{article}
\usepackage{amsmath,a4wide}
\usepackage{amssymb}
\usepackage{amscd}
\usepackage{epsfig}
\usepackage{graphics}

\usepackage{graphicx}

\usepackage{multirow}
\usepackage{subfigure}
\usepackage{epstopdf}
\usepackage{booktabs}
\usepackage{ifpdf}
\usepackage{cite}

\usepackage{amsthm}
\usepackage{color}
\usepackage[round]{natbib}
\usepackage{pdflscape}

\newtheorem{pro}{Proposition}[section]
\newtheorem{dfn}{Definition}[section]
\newtheorem{thm}{Theorem}[section]

\newtheorem{as}{Assumption}
\newtheorem{rem}{Remark}[section]

\newtheorem{pr}{Problem}

\def\t{\theta}

\begin{document}

\title{\LARGE {\bf Duesenberry's Theory of Consumption:  Habit, Learning, and Ratcheting\footnote{We thank participants of the JAFEE-Columbia-NUS Conference, Tokyo, Japan, the Asian Finance Quantitative Finance Conference, Guangzhou, China and the Quantitative Methods in Finance Conference 2018 in Sydney for helpful and encouraging comments. We thank Philip Dybvig for his inspiring work and discussions which provided motivation for this work.}
 } }

\author{
    Kyoung Jin Choi\footnote{E-mail: {\tt kjchoi@ucalgary.ca}, Haskayne School of Business, University of Calgary, Canada. }
    \and
    Junkee Jeon\footnote{E-mail: {\tt junkeejeon@khu.ac.kr}, Department of Applied Mathematics, Kyung Hee University, Korea. }
    \and
    Hyeng Keun Koo\footnote{E-mail: {\tt hkoo@ajou.ac.kr}, Department of Financial Engineering, Ajou University, Korea.    }
}

\date{\today}

\maketitle \pagestyle{plain} \pagenumbering{arabic}

\abstract{This paper investigates the consumption and risk taking decision  of an economic agent with partial irreversibility of consumption decision by formalizing the theory proposed by \cite{Duesenberry}. The optimal policies exhibit a type of the (s, S) policy: there are two wealth thresholds within which consumption stays constant. Consumption increases or decreases at the thresholds and after the adjustment new thresholds are set. The share of  risky investment in the agent's total investment is inversely U-shaped within the (s, S) band, which generates time-varying risk aversion that can fluctuate widely over time. This property can explain  puzzles and questions on asset pricing and  households' portfolio choices, e.g., why aggregate consumption is so smooth whereas the high equity premium is high and the equity return has high volatility,  why the risky share is so low whereas the estimated risk aversion by the micro-level data is small, and whether and when an increase in wealth has an impact on the risky share. Also, the partial irreversibility model can explain both the excess sensitivity and the excess smoothness of consumption.
}

\begin{flushleft}

\medskip


\end{flushleft}


\newpage
\section{Introduction}

\begin{quote}
This critique is based on a demonstration that two fundamental assumptions of aggregate demand theory are invalid. These assumptions are (1) that every individual's consumption behaviour is independent of that of every other individual, and (2) that consumption relations are reversible in time \citep*{Duesenberry}.
\end{quote}

This paper analyzes a model of consumption and portfolio choice decision. Building on the second part of the quote by \cite{Duesenberry} as a critique of the Keynesian consumption function, we aim to model the \textit{irreversibility} of consumption decision. However, we are also aware that each consumption decision is not fully irreversible, but \textit{partially irreversible}. Thus, we model the partial irreversibility by introducing  proportional costs for each consumption adjustment. The partial irreversibility makes consumption a non-smooth function of wealth and  change infrequently  over time, demonstrating the excess sensitivity and the excess smoothness for moderate shocks. The risky share of the household is U-shaped, which can reconcile several conflicting views in the literature, regarding time-varying risk aversion and the impact of a wealth change on  risky investments.

In our model, conditional on consumption never changing, the household has the von-Neuman Morgenstern preference such as the constant relative risk aversion (CRRA) preference. However, there is a utility cost proportional to the current level of marginal utility whenever the household increases or decreases the consumption level. The former is a learning cost in increasing the consumption level\footnote{For example, consider a household who lives in Chicago and goes for a vacation within the U.S. once per year. Suppose the family has a permanent increase in income and decides to visit Europe for a regular vacation from this year. Our assumption implies that the household incurs search or learning costs when preparing to visit a new place. We do not consider a monetary cost of search that should follow every consumption decision.  Rather, our model assumes that there is one time utility cost of searching or learning the new consumption pattern when they increase the consumption level.} and the latter represents consumption ratcheting. Note that the existence of the utility cost from decreasing consumption is important in our model since the model can generate all the properties derived in our model only with the cost of downward adjustment.\footnote{In addition, the cost of upward adjustment is usually very small in our calibration exercise. We can set it to be zero and generate the same resultf by slightly changing values of other variables.}

The optimal polices exhibit a (s, S) type of policy as follows. Suppose the current consumption level is $c$. Then, there are two wealth thresholds $c\underline{x}$ and $c\bar{x}$, where $\underline{x}$ and $\bar{x}$ are  constants, i.e., the current (s, S) band is interval $(c\underline{x}, c\bar{x})$. In this case, consumption increases and decreases if and only if the wealth level hits $c\underline{x}$ and $c\bar{x}$, respectively. Otherwise, consumption stays constant inside the band. Once a boundary  is reached, the new consumption level, $c^{\mbox{new}}$ is set and the next new (s, S) band is updated as $(c^{\mbox{new}} \underline{x}, c^{\mbox{new}}\bar{x})$.   The adjustment is made whenever there is a change in consumption. We describe how each new consumption level is determined in more detail in the main body of the paper.

The optimal risky share is U-shaped in wealth for each (s, S) band. The risky asset holdings consist of two components: the myopic component and  the hedging component. Noticing the minus sign in this decomposition, the hedging demand takes zero at the boundaries of each (s, S) interval and the maximum value inside the interval. The reason is to avoid  a high utility cost that would be incurred if the household frequently adjusted its consumption. Therefore, the risky share is U-shaped with the maximum at the one of the boundaries, the minimum somewhere inside the interval. On the other hand, naturally there is heterogeneity among households in utility costs as well as risk aversion.  Naturally arises the question of how to infer risk aversion of a certain household. We define RCRRA (revealed coefficient of relative risk aversion) by risk aversion inferred by the outsider observing the risky asset holdings of a household over time. Since the risky share is time-varying,  the RCRRA is time-varying. Moreover, the  RCRRA is inverse U-shaped inside the (s, S) boundary since the risky share is U-shaped.

Having the above properties in mind, our model can explain a number of interesting implications for the household  consumption and risky investment decisions. First, our model can fill a gap between different views on  risk aversion in  the literature on  decision theory, structural estimation, behavioral, and asset pricing. Note that most asset pricing models use the relative risk aversion coefficient of around 10 or higher for  calibration exercises to match asset pricing moments (e.g., \cite{BKY}, \cite{CGM}, and references therein\footnote{There are many other asset pricing models  to suggest even much higher value than 10. For example, \cite{KS1991} argues $\gamma = 29$.}). Along the similar lines, households usually hold $ 6 - 20\%$ in equity (conditional on participation, up to 40 \%), which implies relative  risk aversion is at least 10 or larger following calibration according to the standard models  with widely accepted market parameter values. On the contrary, the estimated individual relative risk aversion takes values between 0.7 and 2 in decision theory, structural estimation, or behavioral literature such as \cite{BT2012}, \cite{Campo}, \cite{C2006}, \cite{GH2013, GH2015}, \cite{HS},  \cite{LMN2008}, and \cite{Szpiro}, among which more recent ones tend to claim that risk aversion is less than 1 or around 1.

We show that time-varying risk aversion, i.e., RCRRA in our model is the same as actual risk aversion only when the wealth process hits either one of (s, S) boundaries. The set of these events, however, has measure zero for any stock price sample path. The RCRRA is greater than actual risk aversion for most of times. We note that the maximum value of RCRRA increases with the consumption adjustment costs (RCRRA is always equal to risk aversion if there is no such cost). We show that the RCRRA takes fairly high values closer to its maximum than actual risk aversion during the times when the wealth process stays in the middle range of the (s, S) band and does not have high fluctuation: These are times when the market is neither bullish nor bearish, rather has small volatility for a while. This feature is in sharp contrast with that from traditional habit models since risk aversion in the habit models becomes higher only in downturns. For example, \cite{ORW} show that the size of equity premium in the traditional habit model is determined by a relatively insignificant amount of high-frequency volatility in US aggregate consumption.  Moderate shocks have occurred in the world financial markets over substantial time periods when the RCRRA takes high values in our model. Therefore, our model can generate the average value of RCRRA consistent with that often used in the asset pricing literature, while the coefficient of relative risk aversion is close to that suggested by the behavioral or experimental literature.

Second, but more important is that the U-shaped risky share resulted from our model has an interesting implication for the effect of the change of wealth on the risky share on which the empirical literature looks inclusive. For example, \cite{CCS2009} and \cite{CS2014} favor habit, commitment, or DARA (decreasing relative risk aversion) models predicting that the impact is positive, i.e., the risky share increases with an increase in financial wealth. \cite{BN2008} and \cite{CP2011}, however, show no relationship (or slightly positive relationship, if any) between the financial (liquid) wealth change and the risky share. We argue that both can happen depending on which time-series is chosen. The U-shape implies that the risky share is decreasing in wealth in the left side of (s, S) band and increasing in wealth in the right side of (s, S) band (see Figure  \ref{RCRRA_region}). The optimal wealth process tends to stay longer in the increasing (decreasing) regions of each (s, S) interval over time if the market is such that good shocks are more frequent than bad shocks (see Figure  \ref{wealth-consumption-Ss} and its description). Thus, if good shocks occurred more frequently in the stock market during the data period, the relationship tends to be positive. By using this intuition, we simulate four types of sample paths: (a) bullish, (b) intermediate, (c) bearish, and (d) highly fluctuating. Then, we regress the change of risky share on the change of financial wealth within each sample. For cases (a) and (b), we find a significant positive impact of the wealth increase on the risky share. Moreover, we find no or slightly negative relationship for cases (c) and (d). Thus, our model can reconcile the discrepancy in the empirical literature.

Third, the optimal consumption process in our model features two well-documented empirical regularities: the excess smoothness \citep*{Deaton} and the excess sensitivity \citep*{Flavin}. First, consumption tends not to respond to a permanent shock as long as the shock does not push up or down the wealth level to one of the  (s, S) boundaries, which implies the excess smoothness. Second, suppose there is a good shock in the permanent income. While this shock may not increase the current consumption level, it increases the probability that consumption will increase in the future, by which the excess sensitivity of consumption appears in our model. By the same reason, however, both the excess smoothness and the excess sensitivity vanish for a large shock that immediately pushes the wealth process up or down to either of the (s, S) boundaries. In other words, the partial irreversibility of consumption in our model implies that consumption responds to a large shock consistent with the permanent income hypothesis (\cite{JP}).

Fourth, we explore asset pricing implications and find that our partial reversibility model has a potential to well the U.S. data well. We follow \cite{Const} and \cite{MP}, simulate optimal consumption of individuals, and obtain a simulated series of monthly aggregate consumption. We compute several moments such as the consumption growth rate, the standard deviation of  marginal rate of intertemporal substitution, and the autocorrelation of consumption growth rate. We find that our model matches better the US data than the traditional habit models.

Finally, our paper also makes a theoretical contribution to the literature of dynamic optimization. We first transform the dynamic consumption and portfolio selection problem into a static one, and then derive the dual Lagrangian problem. The advantage of solving the dual Lagrangian problem is that we do not need to consider the portfolio choice and we are left to analyze two sided singular control problem (of consumption). After each adjustment of consumption, the problem is to decide whether to increase or to decrease consumption and if so, how much to change. Then, the original value function is obtained from the convex duality relation. We characterize the full analytic solution. The optimal portfolio can be derived from the dual value function by the convex duality and It\'o's lemma. The detail for each step is presented in the appendix.

\subsection{Literature Review}

Our paper is motivated from \cite{Duesenberry}. On one hand, the first part of his critique quoted in the beginning of our introduction is closely related to external habit formation and thus has significantly contributed to developing the modern habit models such as  \cite{Abel}, \cite{Const}, and \cite{CC1999}.  Some of our results resemble those from habit model. For example, our model generates time-varying risk aversion. However, ours are fundamentally different from habit models in that the agent in habit models becomes more conservative when the current consumption level gets closer to the habit stock (e.g., in downturns) while RCRRA in our model becomes higher in times when the stock market is neither bullish nor bearis, but rather flat and has low volatility. On the other hand, there is a handful of previous literature that grew out from the second part of the critique such as \cite{Dyb} and \cite{JKS}.\footnote{Note that the second part of critique is well described as follows:
\begin{quote}
At any moment a consumer already has a well-established set of consumption habits... Suppose a man suffers a 50 per cent reduction in his income and expects this reduction to be permanent. Immediately after
the change he will tend to act in the same way as before... In retrospect
he will regret some of his expenditures. In the ensuing periods the same
stimuli as before will arise, but eventually he will learn to reject some
expenditures and respond by buying cheap substitutes for the goods
formerly purchased (Dusenberry 1949, p. 24).
\end{quote}}
However, \cite{Dyb} and \cite{JKS} assume that consumption decision is fully irreversible (other than allowing a predetermined depreciation). Therefore, they are an extreme special case of our model.\footnote{More precisely, \cite{Dyb} and \cite{JKS} assume that consumption is not allowed to decrease. The model converges to those of \cite{Dyb} and \cite{JKS} if the cost of decreasing consumption goes to infinite and the cost of increasing consumption is equal to zero in our model.} We better formulate the ratchet effect with partial irreversibility by allowing costly downward adjustment, which generates the U-shape risky share and its novel implications that \cite{Dyb} and \cite{JKS} do not.

Our model is also related to dynamic consumption and investment models with durable consumption or consumption commitment such as \cite{GL}, \citet{HH92,HH93}, \citet{HHZ97}, \cite{FN}, and \cite{CS2007, CS2016}. We view our model as complement to the consumption commitment literature. For example, \cite{CS2016} show in the commitment model that the excess smoothness and excess sensitivity arise for moderate shocks and they vanish for large shocks. The same result  holds from a different channel in our model, namely through learning or adjustment costs. In our model there exist two different adjustment costs, downward and upward adjustment costs. Thus magnitude of large shocks which make immediate adjustment and individuals aggressive in risk taking is in general different for good shock and for bad shocks. This paves a way for empirical test whether the effects of large shock on consumption and risk taking are different for these two different shocks. Furthermore, our model generalizes the models of 
consumption and portfolio selection with durability and local substitution by \citet{HH92} and \citet{HHZ97} via the isomorphism discovered by \citet{SS}.

The rest of the paper continues as follows. Section  \ref{section:Model} describes the model. \ref{sec:sol-tech} present the analysis for the explicit solution. The implications of risky investment  and consumption are provided in Sections \ref{sec:imp-investment} and \ref{sec:imp-consumption}, respectively. Section \ref{sec:conclusion} concludes.

%

\section{Model \label{section:Model}}
We consider a simple and standard continuous-time financial market.
The financial market consists of two assets: a riskless asset and a
risky asset. We assume that the risk-free rate, the rate of return
on the riskless asset, is constant and equal to $r$. The price $S_t$
of the risky asset or the market Index evolves as follows:
$$
dS_t/S_t=\mu dt+\sigma dB_t,
$$
where $\mu,\sigma$ are constants, $\mu>r$, and $B_t$ is a Brownian
motion on a standard probability space
$(\Omega,\mathcal{F},\mathbb{P})$\footnote{See \citet{KS} for
details of mathematics and the probability theory.} endowed with an
augmented filtration $\{\mathcal{F}_t\}_{t\geq 0}$ generated by the
Brownian motion $B_t$.

The agent's wealth process $(X_t)_{t=0}^{\infty}$ evolves according
to the following dynamics:
\begin{eqnarray}\begin{split}\label{eq:wealth}
dX_t = [rX_t +\pi_t (\mu-r) - c_t] dt +\sigma\pi_t dB_t, \ \ \
X_0=X>0,
\end{split}\end{eqnarray}
where   $c_{t}\geq 0$ and $\pi_t$ are the consumption rate and the
dollar amount invested in the risky asset, respectively, at time
$t$.

For non-negative constants $\alpha$ and $\beta$, the agent's utility function is given by
\begin{eqnarray}\label{eq:utility}
U \equiv \mathbb{E}\left[\int_{0}^{\infty}e^{-\delta t}\left(u(c_t)dt-\alpha d(u(c_t))^{+}  -\beta d(u(c_t))^{-}\right)\right],
\end{eqnarray}
where $\delta>0$ is the subjective discount rate and $u(\cdot)$ is a
twice-continuously differentiable, strictly concave, and strictly
increasing function. We decompose
$$u(c_t) = u(c_0) + u(c_t)^+ -
u(c_t)^-,$$ where $u(c_t)^+$ and $u(c_t)^-$ are non-decreasing
processes. This decomposition is well-defined if the consumption
process $c_t$ is admissible (See Definition \ref{def-admissible} and
equation \eqref{con-decomposition}). \eqref{eq:utility} implies that
the agent instantaneously loses the $\alpha$ or $\beta$ unit of
utility whenever there is one unit increase or decrease in the
marginal utility. The example of the first kind is the utility cost
of learning how to spend  and that of the second is the point made
by models with ``catching up with Joneses" such as \citet{Abel},
\citet{Const}, \citet{Gali} and \citet{CC1999}.

In this paper, we assume the constant relative risk
aversion(CRRA) utility function:
\begin{eqnarray}\label{eq:CRRA}
u(c)=\dfrac{c^{1-\gamma}}{1-\gamma},\qquad \gamma>0,\;\gamma\neq1,
\end{eqnarray}
where $\gamma$ is the agent's risk coefficient of relative risk
aversion.

To define the strategy set, we denote by $\Pi$ the family of all
c\'{a}gl\'{a}d, $\mathcal{F}_t$-adapted, non-decreasing process with starting at $0$ and assume that there exist $c^+,c^- \in
{\Pi}$ such that the agent's consumption $c_t$ can be expressed by
\begin{eqnarray} \label{con-decomposition}
c_t = c+ c_t^{+} - c_t^{-},
\end{eqnarray}
where $c$ is the agent's initial consumption rate. Then, the agent's
objective \eqref{eq:utility} is rewritten as
\begin{eqnarray}\label{eq:utility2}
U \equiv \mathbb{E}\left[\int_{0}^{\infty}e^{-\delta t}\left(u(c_t)dt-\alpha u'(c_t) dc_{t}^{+}  -\beta u'(c_t) dc_t^{-}\right)\right].
\end{eqnarray}
We define the set of consumption and risky investment strategies as
follows.
\begin{dfn} \label{def-admissible}
    We call a consumption-portfolio plans $(c^+,c^-,\pi)$ admissible if
    \begin{itemize}
        \item[(a)] A consumption strategy $(c^+,c^-)$ satisfies
        \begin{equation}\label{eq:well-posed}
        \mathbb{E}\left[\int_{0}^{\infty}e^{-\delta t}\left(|u(c_t)|dt +\alpha u'(c_t) dc_{t}^{+}  + \beta u'(c_t) dc_t^{-}\right) \right] <\infty.
        \end{equation}
        We denote by $\Pi(c)$ the class of all consumption strategies $(c^+,c^-)$ satisfying the condition \eqref{eq:well-posed}.
        \item[(b)] For all $t\ge 0$, $\pi_t$ is measurable process with repsect to $\mathcal{F}_t$ satisfying
        \begin{eqnarray}
        \int_{0}^{t}\pi_s^2 ds <+\infty, a.e.
        \end{eqnarray}
        \item[(c)] For all $t\ge 0$, the wealth process is non-negative, i.e.,
        \begin{eqnarray}
        X_t \ge 0.
        \end{eqnarray}
    \end{itemize}
\end{dfn}
The following assumption should be satisfied in order to guarantee
the existence of the special case of $\alpha = \beta = 0$, i.e., the
classical Merton problem.
\begin{as}
   $$
   K \equiv r + \dfrac{\delta -r}{\gamma} + \dfrac{\gamma-1}{2\gamma^2}\t^2 > 0.
   $$
\end{as}
If the agent instantaneously increases her consumption by a small
amount, from $c$ to $ c + \Delta c$ between $t$ and $t + \Delta t$,
the intertemporal utility gain from the additional consumption is $
\dfrac{1}{\delta}\left(u(c+\Delta c)-u(c)\right) \approx
\dfrac{1}{\delta}u'(c)\Delta c.$ On the other hand, the utility loss
from the consumption increase is $\alpha u'(c)\Delta c$. The gain
should be greater than the loss. Otherwise, the agent will never
increase consumption in our model. This observation leads to the
following assumption.
\begin{as} \label{assumption-alpha}
$$ 0\le \delta \alpha < 1.$$
\end{as}
There is no parameter restriction for $\beta \in [0,\infty)$. The
problem studied by \citet{Dyb} is an extreme case of ours.
\citet{Dyb} consider the case when  $\alpha = 0$ and $\beta
\rightarrow \infty$, which means there is no utility cost of
increasing consumption and infinity utility loss from decreasing
consumption (i.e., ratcheting of consumption).

Now we state the problem as follows.
\begin{pr}[Primal Problem (Dynamic Version)]~\label{pr:dynamic_primal_problem}\\
    Given $c_0 = c > 0$ and $X_0 = X>0$, we consider the following utility maximization problem:
    \begin{eqnarray}
    V(X,c)=\sup \;\mathbb{E}\left[\int_{0}^{\infty}e^{-\delta t}\left(u(c_t)dt-\alpha u'(c_t)dc_t^{+}-\beta u'(c_t)dc_t^-\right)\right]
    \end{eqnarray}
    where the supremum is taken over all admissible consumption/portfolio plans $(c^+,c^-,\pi)$ subject to the wealth process \eqref{eq:wealth}.
\end{pr}
Note that Problem  \ref{pr:dynamic_primal_problem} is subject to the
dynamics budget constraint \eqref{eq:wealth}. In Section
\ref{sec:problem-reformulation}, we first transform Problem
\ref{pr:dynamic_primal_problem} into a static problem (Problem
\ref{pr:primal_problem}) by the well-known method of linearizing the
budget constraint suggested by \citet{KLS} and \citet{CoxH}. Finally
we transform that static problem to a singular control problem
(Problem \ref{pr:dual_problem}). Theorem \ref{thm:duality} shows
that the solution to Problem \ref{pr:dynamic_primal_problem} is
recovered from the solution to Problem \ref{pr:dual_problem} by the
duality relationship. We will obtain the solution to Problem
\ref{pr:dual_problem} and characterize optimal policies by using it
in later sections. Note that the advantages of dealing with the
singular control problem in Problem \ref{pr:dual_problem} are
described right below the problem statement.

\section{Solution Analysis \label{sec:sol-tech}}

\subsection{The Problem Reformulation \label{sec:problem-reformulation}}
In order to reformulate Problem \ref{pr:dynamic_primal_problem},
first we transform the wealth process satisfying~\eqref{eq:wealth}
into a static budget constraint. For this purpose we define, for
$t\ge 0$,
\begin{eqnarray*}\begin{split}
        \theta\equiv \dfrac{\mu-r}{\sigma},\;\xi_t\equiv e^{-r t} Z_t,\;\mbox{and}\;\; Z_t\equiv e^{-\frac{1}{2}\theta^2 t-\theta B_t}.
\end{split}\end{eqnarray*}
Let us define an equivalent measure $\mathbb{Q}$ by setting
\begin{eqnarray}
\dfrac{d\mathbb{Q}}{d\mathbb{P}}=Z_T,
\end{eqnarray}
so that the process $B^{\mathbb{Q}}_t =B_t +\t t $ is a standard
Brownian motion under the measure $\mathbb{Q}$. Then, the wealth
process \eqref{eq:wealth} is changed by
\begin{eqnarray}\begin{split}\label{eq:wealth2}
dX_t = [rX_t - c_t] dt +\sigma\pi_t dB_t^{\mathbb{Q}}.
\end{split}\end{eqnarray}
Applying Fatou's lemma and Bayes' rule to $e^{-rt}X_t$, we get the
following static budget constraint:
\begin{eqnarray}\begin{split}\label{eq:static_budget}
\mathbb{E}\left[\int_{0}^{\infty}\xi_t c_{t} dt \right] \le X.
\end{split}\end{eqnarray}
where $X$ is the initial wealth level, i.e., $X_0 = X$. Then, we
restate Problem \ref{pr:dynamic_primal_problem} as the following problem.

\begin{pr}[Primal Problem (Static Version)]~\label{pr:primal_problem}\\
    Given $c_0 = c > 0$ and $X_0 = X >0$, we consider the following utility maximization problem:
    \begin{eqnarray}
    V(X,c)=\sup \;\mathbb{E}\left[\int_{0}^{\infty}e^{-\delta t}\left(u(c_t)dt-\alpha u'(c_t)dc_t^{+}-\beta u'(c_t)dc_t^-\right)\right]
    \end{eqnarray}
    where the supremum is taken over all admissible consumption/portfolio plans $(c^+,c^-,\pi)$ subject to the static budget constraint \eqref{eq:static_budget}.
\end{pr}
By following \citet{CoxH} and \citet{KLS}, it is easy to see that
the solution to Problem \ref{pr:dynamic_primal_problem} is the same
as that to Problem \ref{pr:primal_problem}. Thus, henceforth the
both problems are called the primal problem in this paper.

Using the static budget constraint \eqref{eq:static_budget}, we
consider the following Lagrangian:
\begin{eqnarray}\begin{split}\label{eq:Lagrangian}
{\bf L} =& \mathbb{E}\left[\int_{0}^{\infty}e^{-\delta t}\left(u(c_t)dt-\alpha u'(c_t)dc_t^{+}-\beta u'(c_t)dc_t^-\right)\right]+y\left(X-\mathbb{E}\left[\int_{0}^{\infty}\xi_t c_t dt \right]\right)\\
=& \mathbb{E}\left[\int_{0}^{\infty}e^{-\delta t}\left((u(c_t)-y e^{\delta t}\xi_t c_t)dt-\alpha u'(c_t)dc_t^{+}-\beta u'(c_t)dc_t^-\right)\right]+y X,
\end{split}\end{eqnarray}
where $y>0$ is the Lagrange multiplier for the budget constraint. We define the process
$$
y_t = y e^{\delta t}\xi_t,\;\;\; t \ge 0
$$
which plays the role of the Lagrange multiplier for the budget constraint at time $t$, and thus $(y_t)_{t=0}^\infty$ represents the marginal utility(shadow price) of wealth process. We will describe how the marginal utility of wealth process is related to the agent's optimal consumption policy.

Now we introduce the dual problem of Problem
\ref{pr:dynamic_primal_problem} or Problem \ref{pr:primal_problem}:
\begin{pr}[Dual problem]~\label{pr:dual_problem}
   \begin{eqnarray}\begin{split}
    J(y,c)=&\sup_{(c^+,c^-)\in \Pi(c)}\mathbb{E}\left[\int_{0}^{\infty}e^{-\delta t}\left((u(c_t)-y_t c_t)dt-\alpha u'(c_t)dc_t^{+}-\beta u'(c_t)dc_t^-\right)\right]\\
    =&\sup_{(c^+,c^-)\in \Pi(c)}\mathbb{E}\left[\int_{0}^{\infty}e^{-\delta t}\left(h(y_t,c_t)dt-\alpha u'(c_t)dc_t^{+}-\beta u'(c_t)dc_t^-\right)\right],
   \end{split}\end{eqnarray}
where
$$
h(y,c)=u(c)-yc.
$$
and $\Pi(c)$ is the class of all consumption strategies $(c^+,c^-)$
satisfying the condition \eqref{eq:well-posed}.
\end{pr}
Problem \ref{pr:dual_problem} is the optimization problem with
singular controls over $\Pi(c)$. There are two advantages in dealing
with the dual problem. The first advantage is that we do not need to
consider the portfolio choice. This property is, in fact, inherited
from the formulation of Problem \ref{pr:primal_problem}. The second
advantage is that now the agent's problem becomes a singular control
problem of deciding either to increase or decrease the level of
consumption given the current consumption $c$. Therefore, we can
apply a  standard method of singular control problem developed by
\citet{DN} or \citet{FS}.

\subsection{Solution: Dual Value Function}

The dual value function $J(y,c)$ satisfies the following
Hamilton-Jacobi-Bellman(HJB) equation:
\begin{eqnarray}\label{eq:HJB_dual_value}
\max\{\mathcal{L}J(y,c)+u(c)-yc,\;J_c(y,c)-\alpha u'(c),-J_c-\beta u'(c)\}=0,\;\;\;(y,c)\in \mathcal{R}
\end{eqnarray}
where $\mathcal{R}\equiv\mathbb{R}_{+}\times \mathbb{R}_{+}$ and the differential operator $\mathcal{L}$ is given by
$$
\mathcal{L}=\dfrac{\theta^2}{2}y^2\dfrac{\partial^2}{\partial y^2}+(\delta-r)y\dfrac{\partial}{\partial y} - \delta.
$$
To solve the HJB equation \eqref{eq:HJB_dual_value}, we define the
\textit{increasing region} {\bf IR}, the \textit{decreasing region}
{\bf DR} and the \textit{non-adjustment region} {\bf NR} as follows:
\begin{eqnarray}
\begin{split}
{\bf IR}&=\{(y,c)\in\mathcal{R} \mid J_c(y,c)=\alpha u'(c)\},\\
{\bf NR}&=\{(y,c)\in\mathcal{R} \mid -\beta u'(c)<J_c(y,c)<\alpha
u'(c)\}, \\
{\bf DR}&=\{(y,c)\in\mathcal{R} \mid J_c(y,c)=-\beta u'(c)\}.
\end{split}
\end{eqnarray}
In what follows in this subsection we describe the explicit form of
the dual value function in each region. First, as shown in Appendix
\ref{sec:Appendix:A}, the regions {\bf IR}, {\bf NR} and {\bf DR}
are rewritten by
\begin{eqnarray*}
    \begin{split}
        {\bf IR}&=\{(y,c)\in\mathcal{R} \mid y \le u'(c)b_{\alpha}\},\\
        {\bf NR}&=\{(y,c)\in\mathcal{R}\mid u'(c)b_{\alpha} < y < u'(c)b_{\beta} \},\\
        {\bf DR}&=\{(y,c)\in\mathcal{R} \mid  u'(c)b_{\beta}\le y\},
    \end{split}
\end{eqnarray*}
respectively. See Figure \ref{figure-NR-DR-IR} for the graphical
representation of each region. It is important to characterize {\bf
IR}-, {\bf NR}- and {\bf DR}-regions in order to understand the
optimal strategies, which will be investigated in great detail in
Section \ref{sec:optimal-strategies1}. Here, $b_{\alpha}$ and
$b_{\beta}$ are given by
$$
b_{\alpha}=(1-\delta
\alpha)\dfrac{m_1-1}{m_1}\dfrac{\frac{1}{\kappa}w^{m_1}-1}{w^{m_1-1}-1}>0 \qquad \mbox{and} \qquad b_{\beta}=(1+\delta
\beta)\dfrac{m_1-1}{m_1}\dfrac{w^{m_1}-\kappa}{w^{m_1}-w} >0
$$
with $\kappa=\dfrac{1-\delta\alpha}{1+\delta\beta}$. $m_1$ and $m_2$
are positive and negative roots of following quadratic equation:
$$
\dfrac{\theta^2}{2}m^2 + (\delta-r-\dfrac{\theta^2}{2})m- \delta =0.
$$
Moreover, $w$ is a unique solution to the equation $f(w) =0$ in $(0,1)$, where
\begin{equation}\label{eq:f}
f(w)=(m_1-1)m_2(1-w^{1-m_2})(w^{m_1}-\kappa) - m_1 (m_2-1)(w^{m_1} -w)(1-\kappa w^{-m_2}).
\end{equation}
In the following proposition, we provide a solution to Problem
\ref{pr:dual_problem}.
\begin{pro}\label{pro:solution_dual} The dual value function $J(y,c)$ of Problem \ref{pr:dual_problem} is given by
        \begin{footnotesize}
            \begin{eqnarray}
            \begin{split}
            J(y,c)=
            \begin{cases}
            &\dfrac{D_1 yc}{(1-\gamma+\gamma m_1)b_{\alpha}}\left(\dfrac{y}{c^{-\gamma}b_{\alpha}}\right)^{m_1-1} + \dfrac{D_2 yc}{(1-\gamma+\gamma m_2)b_{\alpha}}\left(\dfrac{y}{c^{-\gamma}b_{\alpha}}\right)^{m_2-1}\\ \\ &+\dfrac{1}{\delta}\dfrac{c^{1-\gamma}}{1-\gamma}-\dfrac{yc}{r},\qquad\qquad\qquad\qquad\qquad\qquad\mbox{for}\;\;(y,c)\in {\bf NR},\\ \\
            \vspace{2mm}
            &J\left(y,I(\dfrac{y}{b_{\alpha}})\right)+\alpha\left(u(c)-u(I(\dfrac{y}{b_{\alpha}}))\right),\qquad\;\mbox{for}\;\;(y,c)\in{\bf IR},\\ \\
            \vspace{2mm}
            &J\left(y,I(\dfrac{y}{b_{\beta}})\right)-\beta\left(u(c)-u(I(\dfrac{y}{b_{\beta}}))\right),\qquad\;\mbox{for}\;\;(y,c)\in{\bf DR},
            \end{cases}
            \end{split}
            \end{eqnarray}
        \end{footnotesize}
    where
    $$
    D_1 = \dfrac{(\alpha-\frac{1}{\delta})m_2 + (m_2-1)\frac{b_{\alpha}}{r}}{m_2 - m_1},\;\;D_2 = \dfrac{(\alpha-\frac{1}{\delta})m_1 + (m_1-1)\frac{b_{\alpha}}{r}}{m_1 - m_2}.
    $$
\end{pro}
\begin{proof}
    The proof is given in Appendix \ref{sec:Appendix:A}.
\end{proof}
Finally we summarize the duality relationship between the value
function of the primal problem and the dual value function of
Problem \ref{pr:dual_problem} in the following theorem.
\begin{thm}\label{thm:duality}
    For the value function $V(X,c)$ of Problem \ref{pr:dynamic_primal_problem} and dual value function $J(y,c)$ of Problem \ref{pr:dual_problem}, the following duality relationship is established:
    \begin{eqnarray}\label{eq:dualityrelationship}
    V(X,c)=\min_{y>0}\left(J(y,c)+yX\right).
    \end{eqnarray}
In addition, there exists a unique solution $y^*$ for the minimization problem \eqref{eq:dualityrelationship}.
\end{thm}
\begin{proof}
    The proof is given in Appendix \ref{sec:Append:B}.
\end{proof}

%

\subsection{Optimal strategies \label{sec:optimal-strategies1}}
  \begin{figure}[h]
  	\begin{center}
  		\setlength{\unitlength}{1cm}
  		\begin{picture}(5,8)
  		\put(-0.4,-0.4){$0$}
  		\put(0,-0.5){\vector(0,1){7}}
  		\put(-1,6.8){\footnotesize dual variable $y$}
  		\put(-0.5,0){\vector(1,0){7}}
  		\put(5,-0.5){\footnotesize marginal utility $u'(c)$}
  		\put(0, 0){\line(2, 3){4}}
        \put(4,6){\footnotesize \;\;$y=u'(c)b_\beta$}    	
  		\put(0, 0){\line(2, 1){6.5}}
  		\put(6.4,3.2){\footnotesize \;\;$y=u'(c)b_\alpha$}
  		\put(0.3,4){{\footnotesize (${\bf DR}$-region)}}
  		\put(3,0.5){{\footnotesize (${\bf IR}$-region)}}
  		\put(4,4){{\footnotesize (${\bf NR}$-region)}}
  		\put(2.3,2){\vector(0,1){1.4}}
  		\put(2.3,2){\vector(0,-1){0.8}}
  		\put(2.5,2.3){\footnotesize $(u'(c_0),y_0)$}
  		\put(2.3,3.45){\line(1,0){0.2}}
  		\put(2.5,3.45){\line(0,1){0.3}}
  		\put(2.5,3.75){\line(1,0){0.2}}
  		\put(2.7,3.75){\line(0,1){0.3}}
  		\put(2.7,4.05){\line(1,0){0.2}}
  		\put(2.9,4.05){\line(0,1){0.3}}
  		\put(2.9,4.35){\line(1,0){0.2}}
  		\put(3.1,4.35){\vector(0,-1){0.4}}
  		\put(2.3,1.17){\line(-1,0){0.3}}
  		\put(2.0,1.17){\line(0,-1){0.15}}
  		\put(2.0,1.02){\line(-1,0){0.3}}
  		\put(1.7,1.02){\line(0,-1){0.15}}
  		\put(1.7,0.87){\line(-1,0){0.3}}
  		\put(1.4,0.87){\line(0,-1){0.15}}
  		\put(1.4,0.72){\line(-1,0){0.3}}
  		\put(1.1,0.72){\vector(0,1){0.3}}
  		\end{picture}
  	\end{center}
  	    \caption{{\bf DR}-region, {\bf NR}-region, and {\bf IR}-region \label{figure-NR-DR-IR}}
  \end{figure}
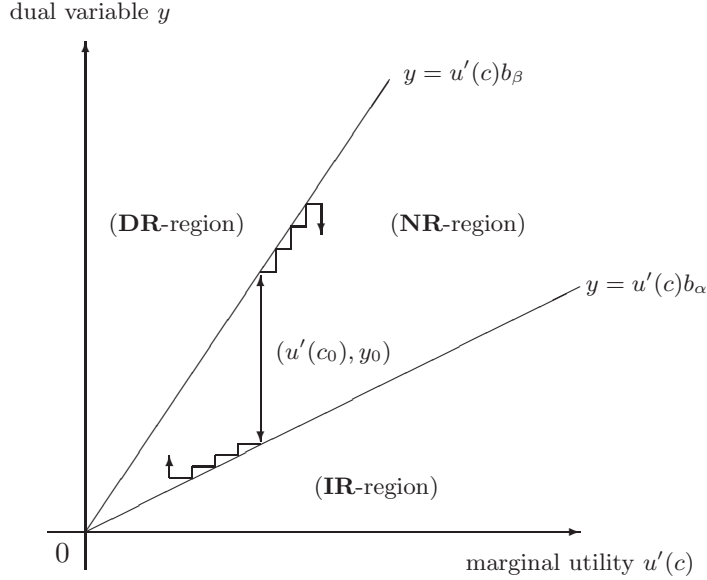

Refer Figure \ref{figure-NR-DR-IR} for the graphical representation
of the optimal consumption behavior. If the initial consumption
level $c_0$ is such that $y_0/u'(c_0)$ lies in the increasing region
{\bf IR} or the decreasing region {\bf DR}, it jumps immediately to
the non-adjustment region {\bf NR}. Suppose the level of consumption
is such that $y_0/u'(c_0)$ lies inside the {\bf NR}-region. The
level of consumption stays constant during the time $y_t$-process
lies inside the {\bf NR}-region. Consumption jumps down if and only
if $y_t$ process hits $u'(c_0) b_{\beta}$ and it jumps up if and
only if $y_t$ hits $u'(c_0) b_{\alpha}$. In this case, the question
is how much the consumption level jumps up or down. Proposition
\ref{pro:consumption} explicitly characterizes the amount of the
jump at each time when there is a consumption adjustment.

\begin{pro}\label{pro:consumption}
The optimal consumption $c_t^*$ for $t\ge 0$ is given by
$$
c_t^* = c_0 + c_t^{*,+}-c_t^{*,-},
$$
where $y_t^*=y^* e^{\delta t} \xi_t$ and $y^*$ is the unique solution to the minimization problem \eqref{eq:dualityrelationship} and
\begin{footnotesize}
    \begin{eqnarray}
        \begin{split}\label{eq:optimal_consumption}
            c_t^{*,+}&=\max\left\{0, -c_0 + \sup_{s\in[0,t)}\left(c_s^{*,-}+I(\frac{y_s^*}{b_{\alpha}})\right) \right\},\\
            c_t^{*,-}&=\max\left\{0, \;\;\;c_0 +\sup_{s\in[0,t)}\left(c_s^{*,+}-I(\frac{y_s^*}{b_{\beta}})\right) \right\}.
        \end{split}
    \end{eqnarray}
\end{footnotesize}
\end{pro}
\begin{proof}
    The proof is given in Appendix \ref{sec:Append:C}.
\end{proof}

\begin{figure}[h]
\centering
\subfigure[$\dfrac{y_t}{u'(c_t^*)}$]{\label{fig00b}\includegraphics[scale=0.4]{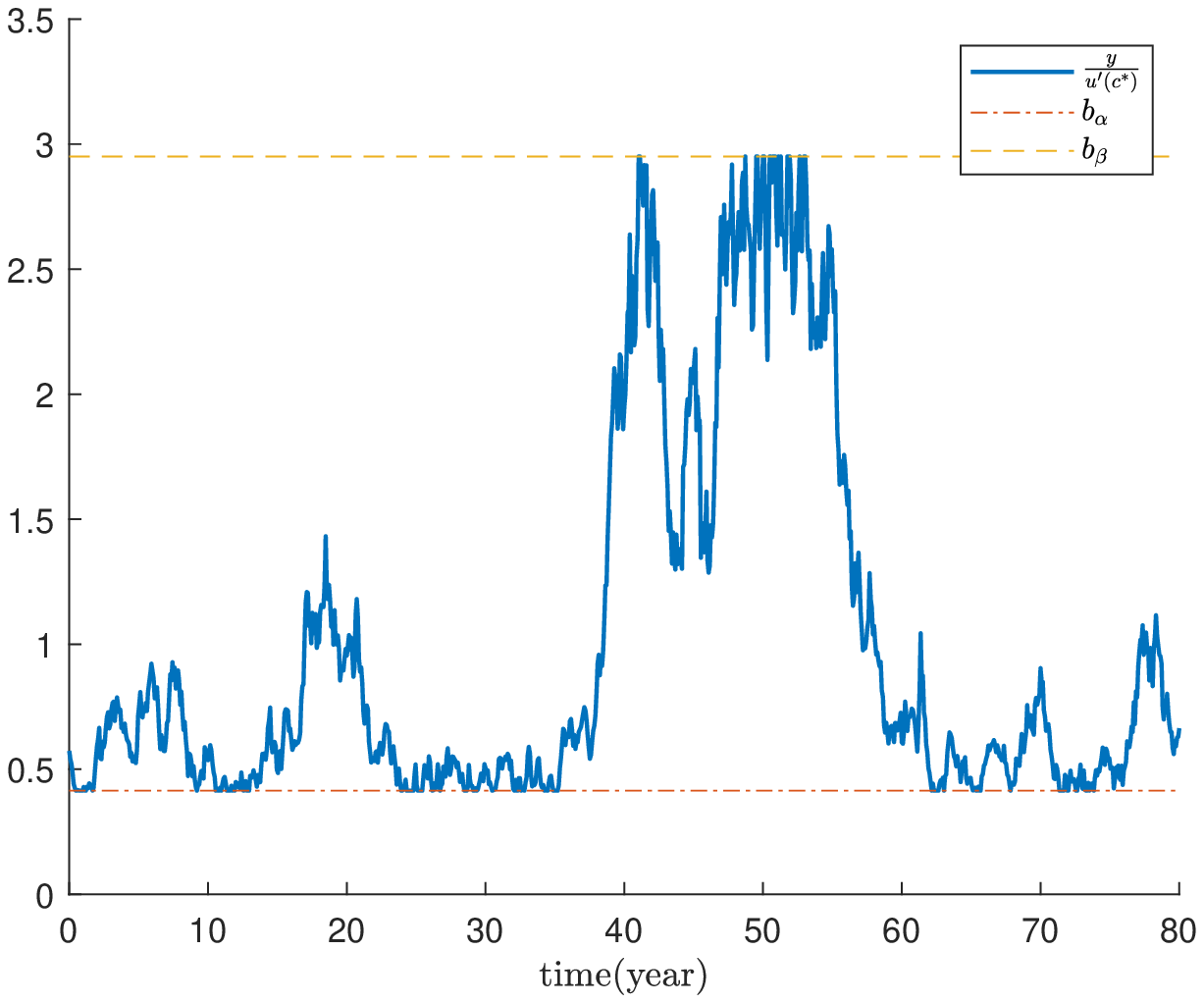}}
\subfigure[$c_t^*$]{\label{fig00a}\includegraphics[scale=0.4]{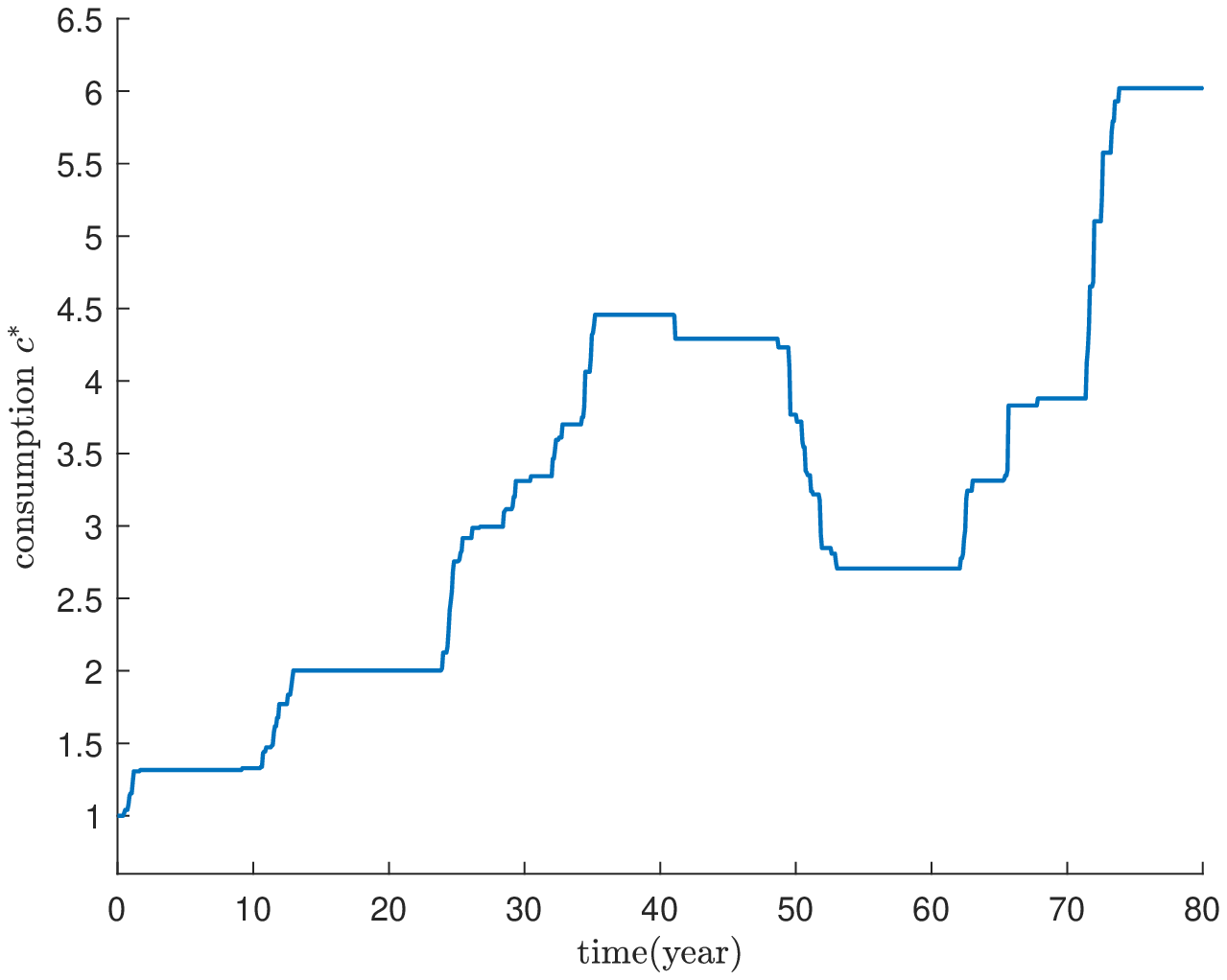}}
\subfigure[$c^{*,+}$]{\label{fig00c}\includegraphics[scale=0.4]{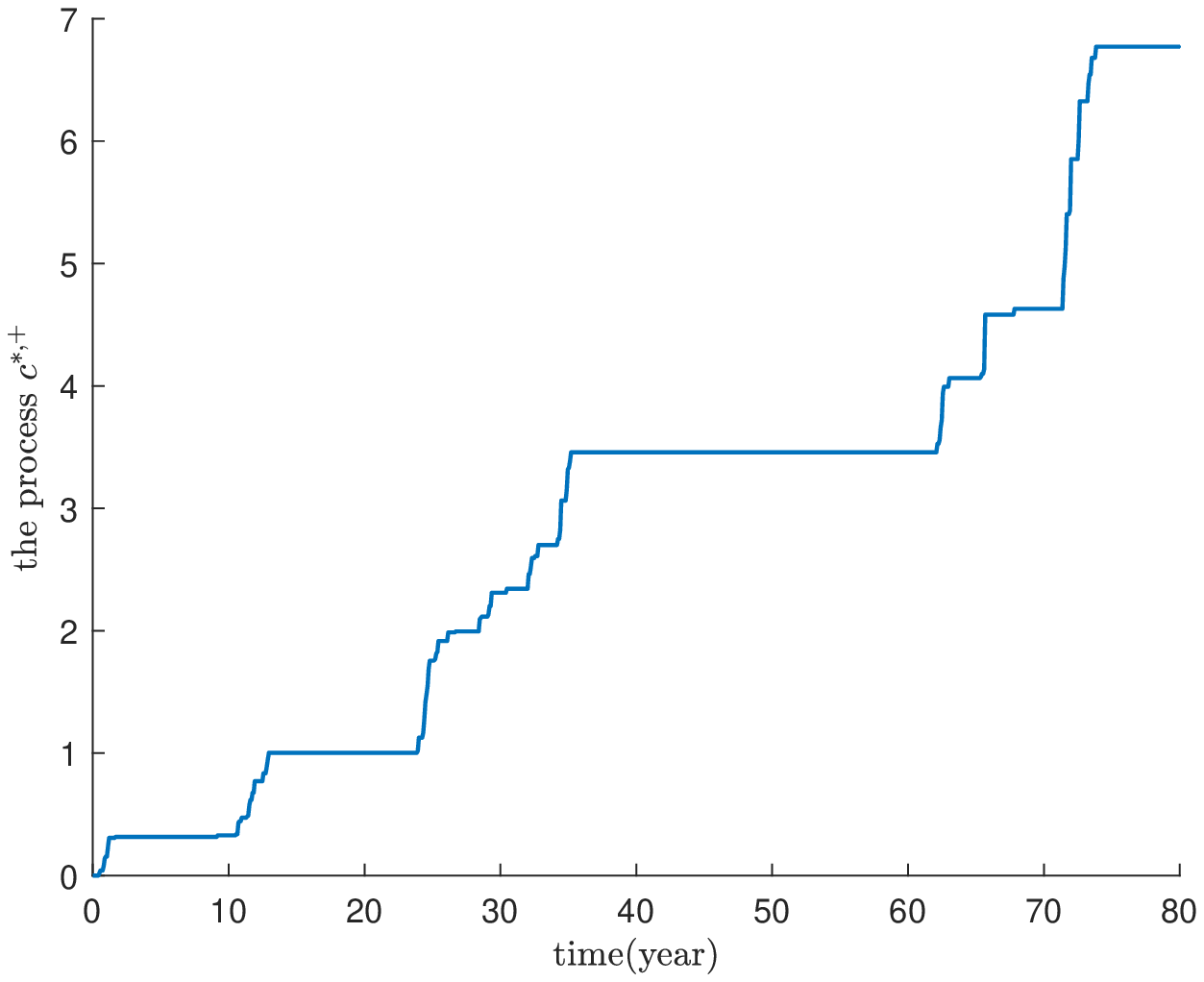}}
\subfigure[$c^{*,-}$]{\label{fig00d}\includegraphics[scale=0.4]{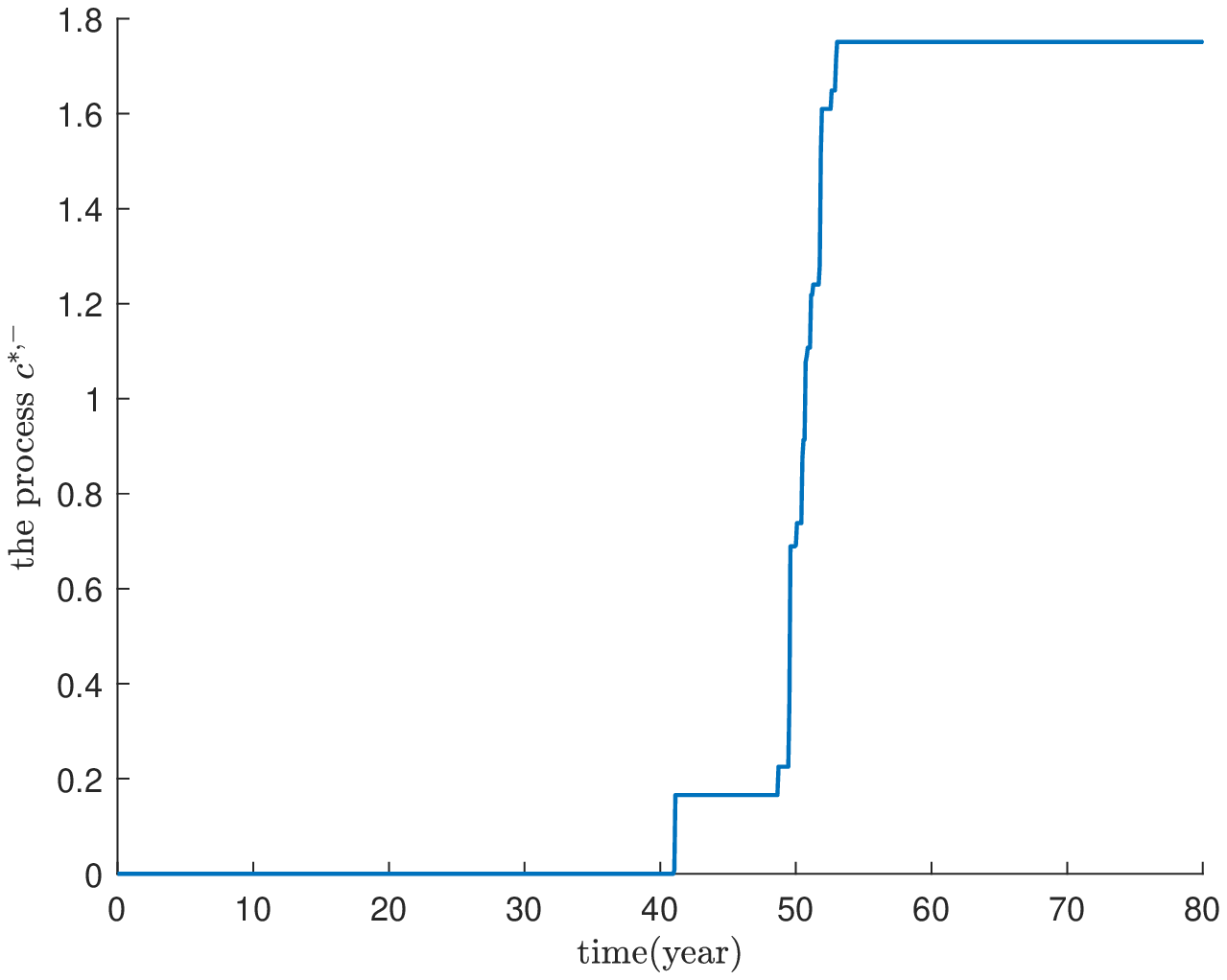}}
\caption{Simulation of optimal consumption, ${y}/{u'(c^*)}$,
$c^{*,+}$ and $c^{*,-}$ with $\alpha=5$, $\beta=10$. The other
parameter values are as follows:
$\rho=0.02,\;\mu=0.085,\;r=0.015,\;\sigma=0.25,\;\gamma=2,\;X=50$,
$c=1$ and $T=80$. \label{figure-consumption-path1}}
\end{figure}

Processes $c_t^{*,+}$ and $c_t^{*,-}$ are non-decreasing regulators such that
$y_t^*/u'(c_t^*)$ lies inside the ${\bf NR}$-region. $c_t^{*,+}$ and $c_t^{*,-}$
stay mostly constant.  Consumption increases whenever $c_t^{*,+}$ increases, but
consumption decreases whenever $c_t^{*,-}$ increases. More precisely, as seen in
Figure \ref{figure-NR-DR-IR}, the process $c_t^{*,+}$ stays constant within the
non-adjustment region {\bf NR} and increases if and only if
${y_t^*}/u'(c_t)$ hits the free boundary $b_{\alpha}$. Similarly,
the process $c_t^{*,-}$ also stays constant within the non-adjustment
region {\bf NR}.  $c_t^{*,-}$ increases if and only if ${y_t^*}$ hits the free
boundary $b_{\beta}$. Figure \ref{figure-consumption-path1} plots
sample paths of $c_t^{*,+}$ and $c_t^{*,-}$ together with those of
$\dfrac{y_t}{u'(c_t^*)}$ and $c_t^*$. See Figure  \ref{figure-consumption-path1}
for simulated paths of  $\frac{y_t}{u'(c_t^*)}$, $c_t^{*,+}$, $c_t^{*,-}$, and $c_t^*$.

Proposition \ref{pro:consumption} describes the consumption path
according to the ratio of the shadow price of wealth to the marginal
utility process. Then, how is the consumption process
related to the agent's wealth process? The following theorem provides
the answer to this question.

\begin{thm}~\label{thm:wealth}
Pick an arbitrary time $s \geq 0$ and let $c_s^*=c$ be the optimal
consumption at $s$ for some constant $c>0$.  For $t \geq s$, the optimal consumption is fixed as $c_t^* = c$
during the time in which $y_t$ is inside the $\bf{NR}$ region (by Proposition \ref{pro:consumption}).
In this case, the optimal wealth process $X_t^*$ follows
\begin{equation} \label{optimal-wealth-NR}
X_{t}^{*}= \dfrac{c}{r}-c\left(\dfrac{D_1
m_1}{(1-\gamma+\gamma m_1)b_{\alpha}}\left(\dfrac{y_t^*}{
{c}^{-\gamma} b_{\alpha}}\right)^{m_1-1}+\dfrac{D_2
m_2}{(1-\gamma+\gamma m_2)b_{\alpha}}\left(\dfrac{y_t^*}{
{c}^{-\gamma} b_{\alpha}}\right)^{m_2-1}\right).
\end{equation}
In addition, there exist two positive numbers
$\underline{x}$ and $\bar{x}$
such that $c_t^* = c$ for $t \geq s$  if and only if
\begin{equation} \label{S-s-threshold}
\underline{x} < \frac{X_t^*}{c} < \bar{x} \qquad \mbox{or} \qquad c \underline{x} < X_t^* < c\bar{x} .
\end{equation}
The explicit forms of $\underline{x}$ and
$\bar{x}$ are given in the proof.
\end{thm}
\begin{proof}
The proof is given in Appendix \ref{sec:Append:D}.
\end{proof}

Suppose $c_s^*=c$ is the optimal consumption level at a certain
time $s$ as in Theorem \ref{thm:wealth}. Equation
\eqref{optimal-wealth-NR} explicitly describes the optimal wealth
process during the times when the shadow value of wealth,
$y_t$-process stays inside the $\bf{NR}$-region for $t \geq s$.

The consumption level increases according to the first equation in
\eqref{eq:optimal_consumption} if and only if $y_t$ hits the lower
threshold $u'(c_s) b_{\alpha}$ or touches the $\bf{IR}$-region. The
latter part of Theorem \ref{thm:wealth} tells that this condition is
equivalent to the case when $X_t^*$ hits the upper threshold
$c \bar{X}$. Conversely, the consumption level decreases according to the second equation in
\eqref{eq:optimal_consumption} if and only if $y_t$ hits the upper
threshold $u'(c_s) b_{\beta}$ or touches the $\bf{DR}$-region. This
condition is equivalent to the case when $X_t^*$ hits the lower
threshold $c \underline{x}$.

	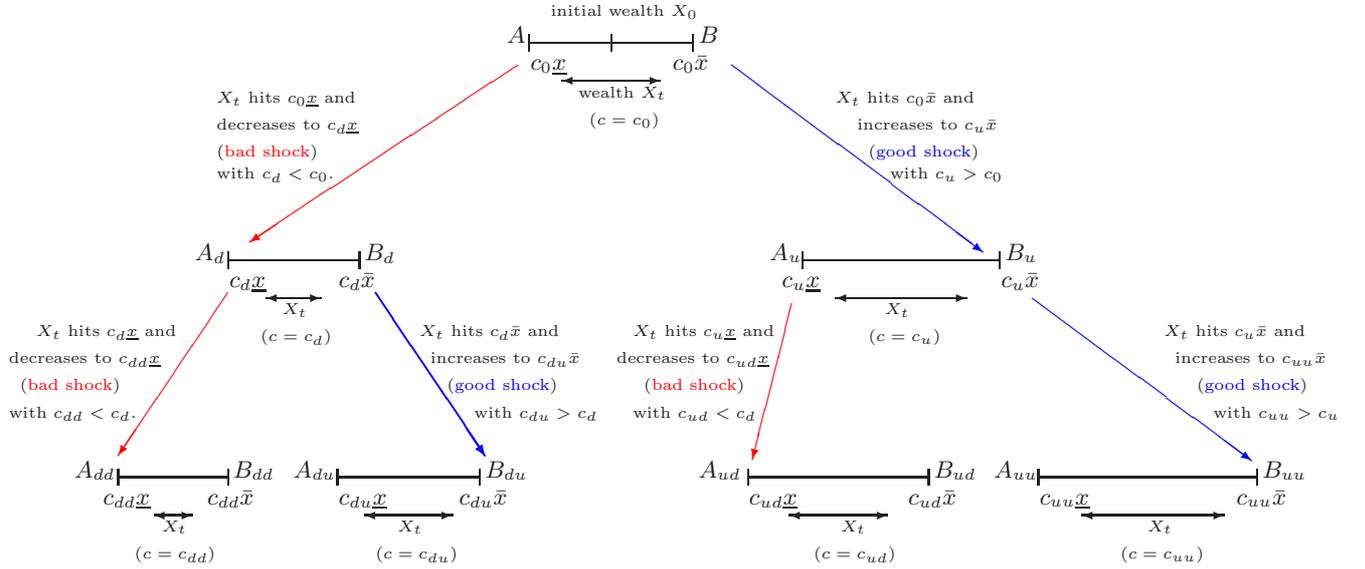
\begin{figure}[ht]
	\setlength{\unitlength}{0.72cm}
	\begin{picture}(5,11)
	\put(9, 10){\line(1,0){3}}
	\put(8.6, 10){\footnotesize $A$}
	\put(12.1, 10){\footnotesize $B$}	
	\put(9, 9.85){\line(0,1){0.3}}
	\put(10.5, 9.85){\line(0,1){0.3}}
	\put(12, 9.85){\line(0,1){0.3}}
	\put(9.4,10.5){\tiny initial wealth $X_0$}
    \put(10.1,8.5){\tiny $(c=c_0)$}
	\put(9.9,9){\tiny wealth $X_t$}    	
	\put(10.5,9.3){\vector(-1,0){0.9}}
	\put(10.5,9.3){\vector(1,0){0.9}}
	\put(8.9,9.5){\footnotesize \;$c_0\underline{x}$}    	
	\put(11.5,9.5){\footnotesize \;$c_0\bar{x}$}
	
	\put(3.5, 6){\line(1,0){2.4}}
	\put(3.5, 5.85){\line(0,1){0.3}}
	\put(5.9,5.85){\line(0,1){0.3}}
	\put(4.5,5){\tiny $X_t$}
	\put(4.1,4.5){\tiny $(c=c_d)$}  	
	\put(4.7,5.3){\vector(-1,0){0.5}}
	\put(4.7,5.3){\vector(1,0){0.5}}
	\put(3.4,5.5){\footnotesize \;$c_d\underline{x}$}
	\put(2.9,6){\footnotesize $A_d$}
	\put(6,6){\footnotesize $B_d$}	   			   	
	\put(5.4,5.5){\footnotesize \;$c_d\bar{x}$} 	
	
	\put(14., 6){\line(1,0){3.6}}
	\put(14., 5.85){\line(0,1){0.3}}
	\put(17.6, 5.85){\line(0,1){0.3}}
	\put(15.5,5){\tiny $X_t$}    	
    \put(15.2,4.5){\tiny $(c=c_u)$}    	
	\put(15.8,5.3){\vector(-1,0){1.2}}
	\put(15.8,5.3){\vector(1,0){1.2}}
	\put(13.5,5.5){\footnotesize \;$c_u\underline{x}$}    	
	\put(17.5,5.5){\footnotesize \;$c_u\bar{x}$}
	\put(13.4, 6){\footnotesize $A_{u}$}
	\put(17.7, 6){\footnotesize $B_{u}$}
	
	\put(1.5, 2){\line(1,0){2}}
	\put(1.5, 1.85){\line(0,1){0.3}}
	\put(3.5,1.85){\line(0,1){0.3}}
	\put(2.3,1){\tiny $X_t$}
	\put(1.8,0.5){\tiny $(c=c_{dd})$}    		
	\put(2.5,1.3){\vector(-1,0){0.35}}
	\put(2.5,1.3){\vector(1,0){0.35}}
	\put(1.1,1.5){\footnotesize \;$c_{dd}\underline{x}$}
	\put(0.7,2){\footnotesize $A_{dd}$}	
	\put(3.6,2){\footnotesize $B_{dd}$}	
	\put(3,1.5){\footnotesize \;$c_{dd}\bar{x}$}

	\put(5.5, 2){\line(1,0){2.6}}
	\put(5.5, 1.85){\line(0,1){0.3}}
	\put(8.1,1.85){\line(0,1){0.3}}
	\put(6.6,1){\tiny $X_t$}
	\put(6.2,0.5){\tiny $(c=c_{du})$} 	
	\put(6.8,1.3){\vector(-1,0){0.8}}
	\put(6.8,1.3){\vector(1,0){0.8}}
	\put(5.4,1.5){\footnotesize \;$c_{du}\underline{x}$}
    \put(4.7,2){\footnotesize $A_{du}$}	
    \put(8.2,2){\footnotesize $B_{du}$}			
	\put(7.6,1.5){\footnotesize \;$c_{du}\bar{x}$}
	
	\put(13., 2){\line(1,0){3.3}}
	\put(13., 1.85){\line(0,1){0.3}}
	\put(16.3,1.85){\line(0,1){0.3}}
	\put(14.7,1){\tiny $X_t$}
	\put(14.2,0.5){\tiny $(c=c_{ud})$} 		
	\put(14.65,1.3){\vector(-1,0){0.9}}
	\put(14.65,1.3){\vector(1,0){0.9}}
	\put(12.9,1.5){\footnotesize \;$c_{ud}\underline{x}$}    	
	\put(15.8,1.5){\footnotesize \;$c_{ud}\bar{x}$}
	\put(12.1,2){\footnotesize $A_{ud}$}	
	\put(16.4,2){\footnotesize $B_{ud}$}	
	
	\put(18.3, 2){\line(1,0){4}}
	\put(18.3, 1.85){\line(0,1){0.3}}
	\put(22.3,1.85){\line(0,1){0.3}}
	\put(20.3,1){\tiny $X_t$}    	
	\put(19.8,0.5){\tiny $(c=c_{uu})$} 	
	\put(20.4,1.3){\vector(-1,0){1.3}}
	\put(20.4,1.3){\vector(1,0){1.3}}
	\put(18.2,1.5){\footnotesize \;$c_{uu}\underline{x}$}    	
	\put(21.8,1.5){\footnotesize \;$c_{uu}\bar{x}$}
	
	\put(17.5,2){\footnotesize $A_{uu}$}	
	\put(22.4,2){\footnotesize $B_{uu}$}
	
	\textcolor{red}{\put(8.8, 9.6){\vector(-3,-2){4.9}}}
	\put(3.3,8.9){\tiny $X_t$ hits $c_0\underline{x} $ and}
	\put(3.3,8.4){\tiny decreases to $c_d \underline{x}$ }
	\put(3.3,7.9){\tiny (\textcolor{red}{bad shock})}
	\put(3.3,7.5){\tiny with $ c_d < c_0$.}    	
	
	\textcolor{red}{\put(3.5,5.4){\vector(-2,-3){2}}}
	\put(-0.0, 4.6){\tiny $X_t$ hits $c_d\underline{x}$ and}
	\put(-0.5, 4.1){\tiny decreases to $c_{dd}\underline{x}$ }
	\put(-0.3, 3.6){\tiny  (\textcolor{red}{bad shock})}
	\textcolor{blue}{\put(6.2,5.4){\vector(2, -3){2}}}
	\put(-0.5, 3.1){\tiny with $c_{dd} < c_d$.}
	
	\put(7.0, 4.6){\tiny $X_t$ hits $c_d\bar{x}$ and}
	\put(7.2, 4.1){\tiny increases to $c_{du}\bar{x}$ }
	\put(7.5, 3.6){\tiny (\textcolor{blue}{good shock})}
	\textcolor{blue}{\put(6.2,5.4){\vector(2, -3){2}}}
	\put(8.0, 3.1){\tiny with $c_{du} > c_d$}
	
	\textcolor{red}{\put(13.8,5.2){\vector(-1,-4){0.72}}}
	\put(10.9, 4.6){\tiny $X_t$ hits $c_u\underline{x}$ and}
	\put(10.6, 4.1){\tiny decreases to $c_{ud}\underline{x}$}
	\put(11.1, 3.6){\tiny (\textcolor{red}{bad shock})}
	\textcolor{blue}{\put(6.2,5.4){\vector(2, -3){2}}}
	\put(10.9, 3.1){\tiny with $c_{ud} < c_d$}
	
	\textcolor{blue}{\put(18.2,5.3){\vector(4,-3){4.}}}
	\put(20.6, 4.6){\tiny $X_t$ hits $c_u\bar{x}$ and}
	\put(20.8, 4.1){\tiny increases to $c_{uu}\bar{x}$}
	\put(21.2, 3.6){\tiny (\textcolor{blue}{good shock})}
	\textcolor{blue}{\put(6.2,5.4){\vector(2, -3){2}}}
	\put(21.5, 3.1){\tiny with $c_{uu} > c_{u}$}
	
	\textcolor{blue}{\put(12.7,9.6){\vector(4,-3){4.6}}}
	\put(14.6,8.9){\tiny $X_t$ hits $c_0\bar{x}$ and}
	\put(15,8.4){\tiny increases to $c_u\bar{x}$}
	\put(15.2,7.9){\tiny (\textcolor{blue}{good shock})}
	\put(15.6,7.5){\tiny with $c_u> c_0$}    	
	
	\end{picture}
	\caption{A discrete example of wealth and consumption. \label{wealth-consumption-Ss}}
\end{figure}

The optimal consumption policy looks like a (s, S) policy over
wealth. Note that the two threshold levels $\underline{x}$ and $\bar{x}$ depend on
the market parameter values and $(\alpha, \beta$).  Equation \eqref{S-s-threshold} implies that
for a given consumption $c$ at a certain point of time, the thresholds are $c \underline{x}$ and $c \bar{x}$.
Once the agent's wealth reaches either one of two thresholds of wealth, the new consumption level $c'$ is determined by \eqref{eq:optimal_consumption}. Then, the two new thresholds
levels are set as $c' \underline{x}$ and $c' \bar{x}$.

Figure \ref{wealth-consumption-Ss} describes a discrete-time  version of the wealth and consumption movement. Assume that current wealth is $X_0$ and consumption is $c_0$. Initially the wealth process lies in interval $(A, B)$. Suppose $X_t$ hits $c_0 \bar{x}$ (point $B$) after several good shocks. Then, a new consumption level is set as $c_u$ and the new (s, S) band is interval $(A_u, B_u)$. At this instant, the wealth level is at point $B_u$. If a positive shock arrives again at point $B_u$, a new consumption level is immediately set as $c_{uu}$ and the new (s, S) band is interval  $(A_{uu}, B_{uu})$. If a negative shock arrives at point $B_u$, however, $X_t$ moves inside  $(A_u, B_u)$ and the consumption level stays at $c_u$ until $X_t$ hits either $c_u \underline{x}$ (point $A_u$) or $c_u \bar{x}$ (point $B_u$).

On the contrary, let us go back to the time when the wealth process lies in interval $(A, B)$. Suppose $X_t$ hits $c_0 \underline{x}$ (point $A$) after several bad shocks.  Then, a new consumption level is set as $c_d$ and the new (s, S) band is interval $(A_d, B_d)$. At this instant, the wealth level is at point $A_d$. If a negative shock arrives again at point $A_d$, a new consumption level is immediately set as $c_{dd}$ and the new (s, S) band is interval  $(A_{dd}, B_{dd})$.  If a positive shock arrives at at point $A_d$, however, $X_t$ moves inside  $(A_d, B_d)$ and the consumption level stays at $c_u$ until $X_t$ hits either $c_d \underline{x}$ (point $A_d$) or $c_d \bar{x}$ (point $B_d$).

\begin{figure}[h]
	\centering
	\subfigure[$X_t^*$]{\label{fig00b}\includegraphics[scale=0.33]{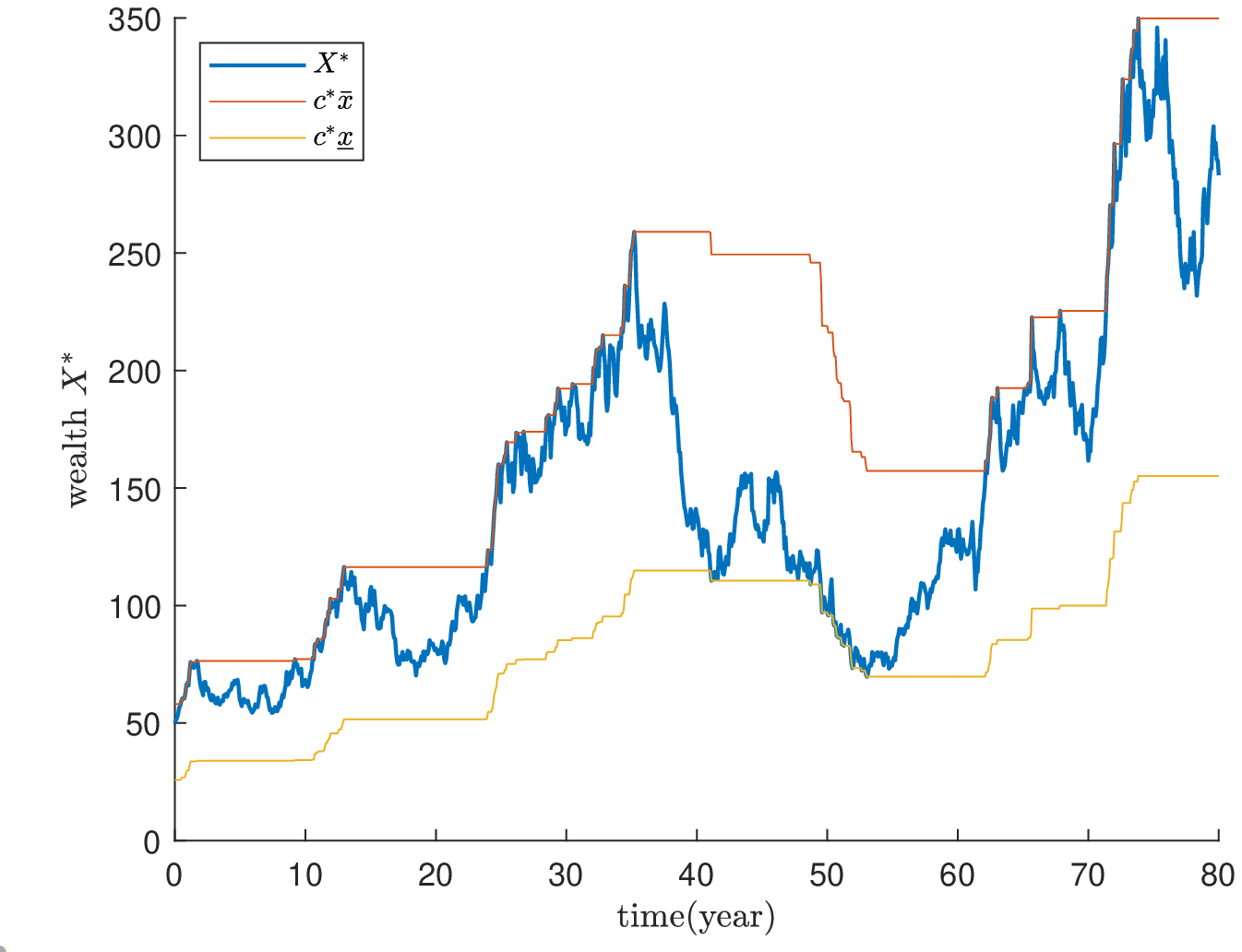}}
	\subfigure[$c_t^*$]{\label{fig00a}\includegraphics[scale=0.33]{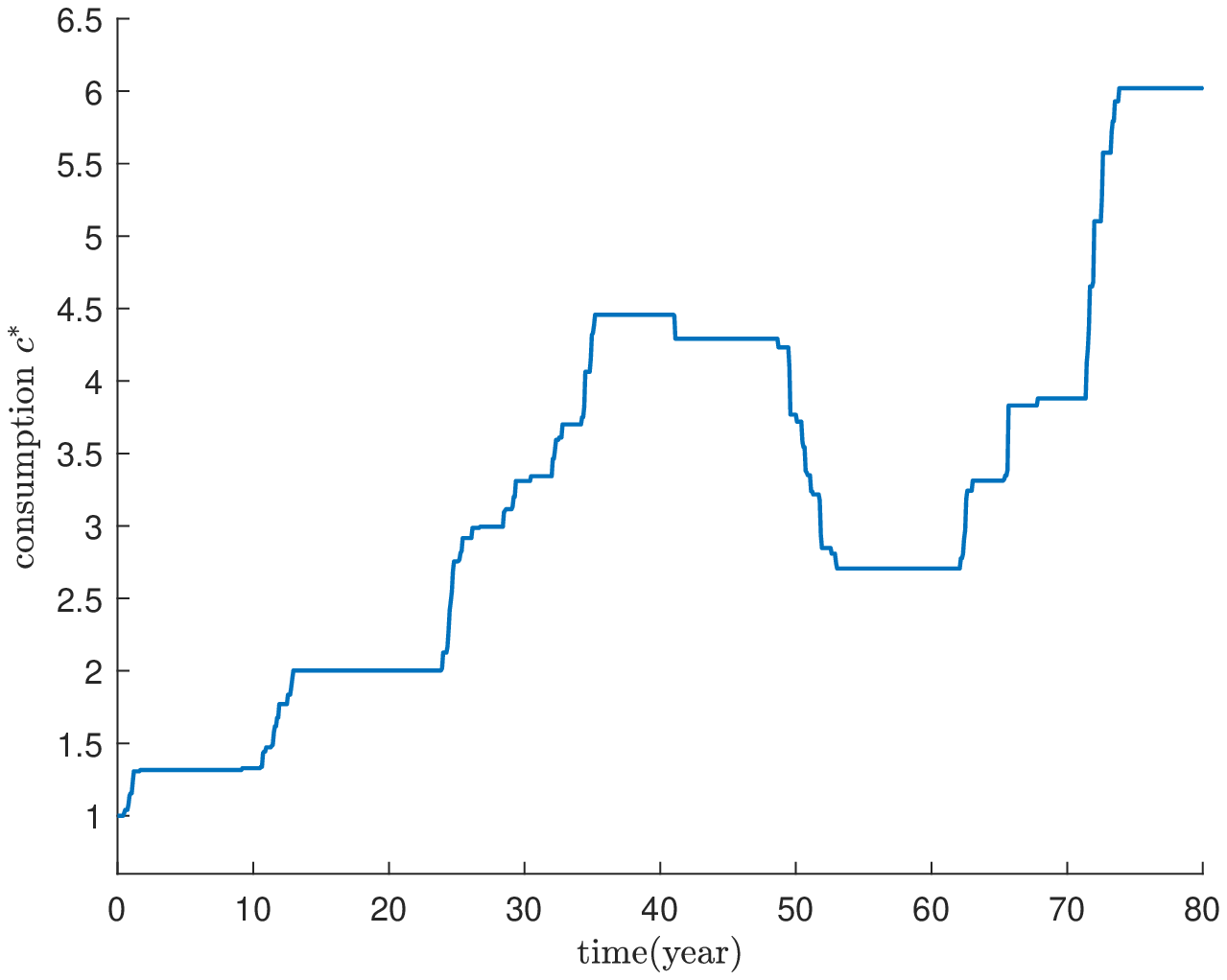}}
	\subfigure[$\pi_t^*$]{\label{fig00c}\includegraphics[scale=0.33]{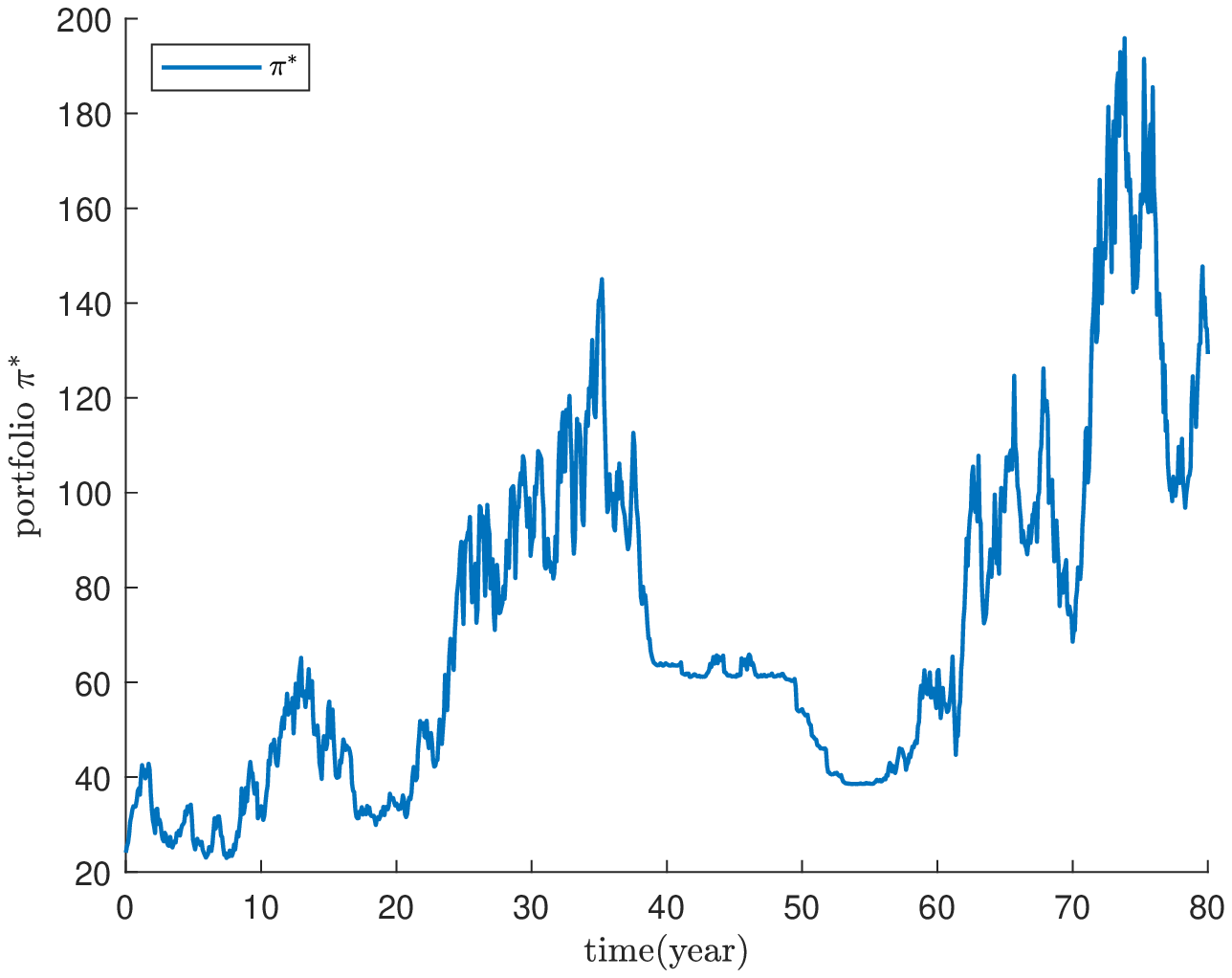}}
	\caption{\label{cp11} Simulation of optimal consumption, wealth and portfolio with $\alpha=5$, $\beta=10$. The other parameter values are as follows: $\rho=0.02$, $\mu=0.085$, $r=0.015$, $\sigma=0.25$,\;$\gamma=2$, $X_0=50$, $C_0=1$ and $T=80$. \label{fig-sim-path1}}
\end{figure}

Now let us turn to the optimal risky investment, which is given in Proposition \ref{pro:portfolio} below.
The explicit form of the optimal portfolio in is easily obtained by applying Ito's lemma to the optimal wealth process in \eqref{optimal-wealth-NR} and drawing out the corresponding term. Figure \ref{cp11}
shows a simulated path of $X_t^*$ and plots its corresponding consumption and risky portfolio.
Notice from the figure that the number of portfolio rebalances is far more than the number of consumption adjustment over time. In other words, the household keeps rebalancing risky stock holdings within the (s, S) band while consumption is adjusted only at the boundary.

\begin{pro}\label{pro:portfolio}
Let $c_s^* = c$ be the optimal consumption level at time $s$ for some constant $c$. The optimal portfolio $\pi_{t}^{*}$ is given by
\begin{equation} \label{opt-portfolio}
\pi_t^*=\dfrac{\t }{\sigma}c\left(\dfrac{D_1
m_1(m_1-1)}{(1-\gamma+\gamma m_1)b_{\alpha}}\left(\dfrac{y_t^*}{
{c}^{-\gamma} b_{\alpha}}\right)^{m_1-1}+\dfrac{D_2
m_2(m_2-1)}{(1-\gamma+\gamma m_2)b_{\alpha}}\left(\dfrac{y_t^*}{
{c}^{-\gamma} b_{\alpha}}\right)^{m_2-1}\right)
\end{equation}
for $t \geq s$ before the next consumption adjustment happens.

\end{pro}
\begin{proof}
The proof is given in Appendix \ref{sec:Append:E}.
\end{proof}

Based on the classical portfolio selection results, we can easily see that the optimal risky portfolio in Equation \eqref{opt-portfolio} in Proposition \ref{pro:portfolio} (a) will be decomposed as follows:
\begin{equation} \label{portfolio-decompose}
\pi_t^* = \frac{\mu -r}{\gamma \sigma^2}X_t^* - \left( \frac{\mu -r}{\gamma \sigma^2}X_t^* - \pi_t^*\right)
\qquad \mbox{or} \qquad \frac{\pi_t^*}{X_t^*} = \frac{\mu -r}{\gamma \sigma^2} - \left( \frac{\mu -r}{\gamma \sigma^2} - \frac{\pi_t^*}{X_t^*}\right) .
\end{equation}
The first term is a myopic term and the remaining term is a hedging term. Note the minus sign in front of the second term. Intuitively, the second term in  \eqref{portfolio-decompose} is the demand that crows out the myopic demand, in order to maintain the current consumption level since frequent changes of the consumption level incur high utility cost. In this sense, the second term itself is positive. As seen in Panel (b) of Figure \ref{fig-saving-portfolio-RCRRA}, the risky share is U-shaped and the ratio of the hedging demand takes the largest value at the certain point inside the (s, S) band at which the total risky share takes the minimum value, which generates interesting implications. We will discuss the implications of the U-shaped risky share over time in Section  \ref{sec:imp-investment} (e).

\begin{figure}[h]
	\centering
	\subfigure[saving \& portfolio ]{\includegraphics[scale=0.3]{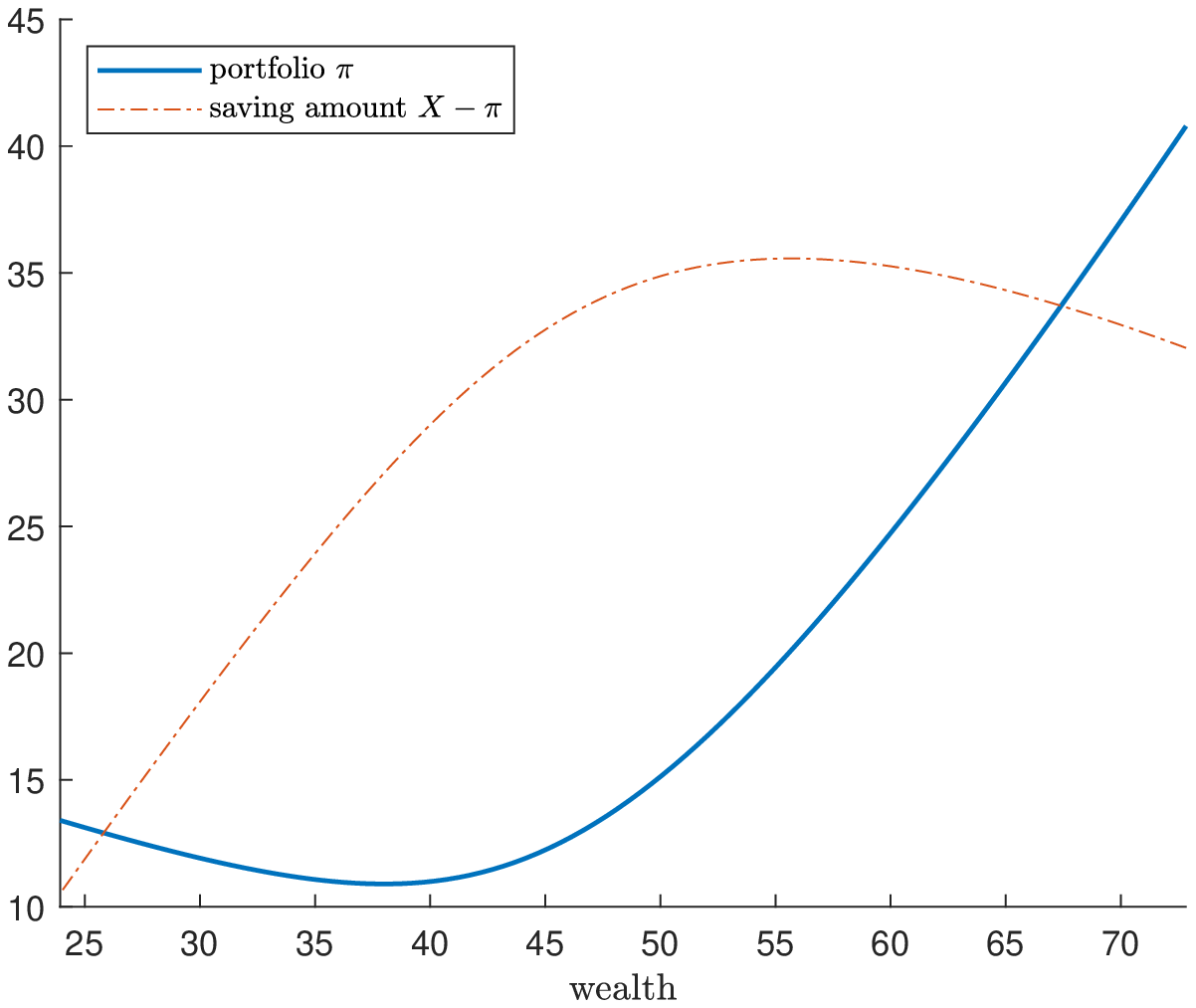}}
	\subfigure[$\frac{\mbox{saving}}{\mbox{wealth}}$ \& $\frac{\mbox{portolio}}{\mbox{wealth}}$]{\includegraphics[scale=0.3]{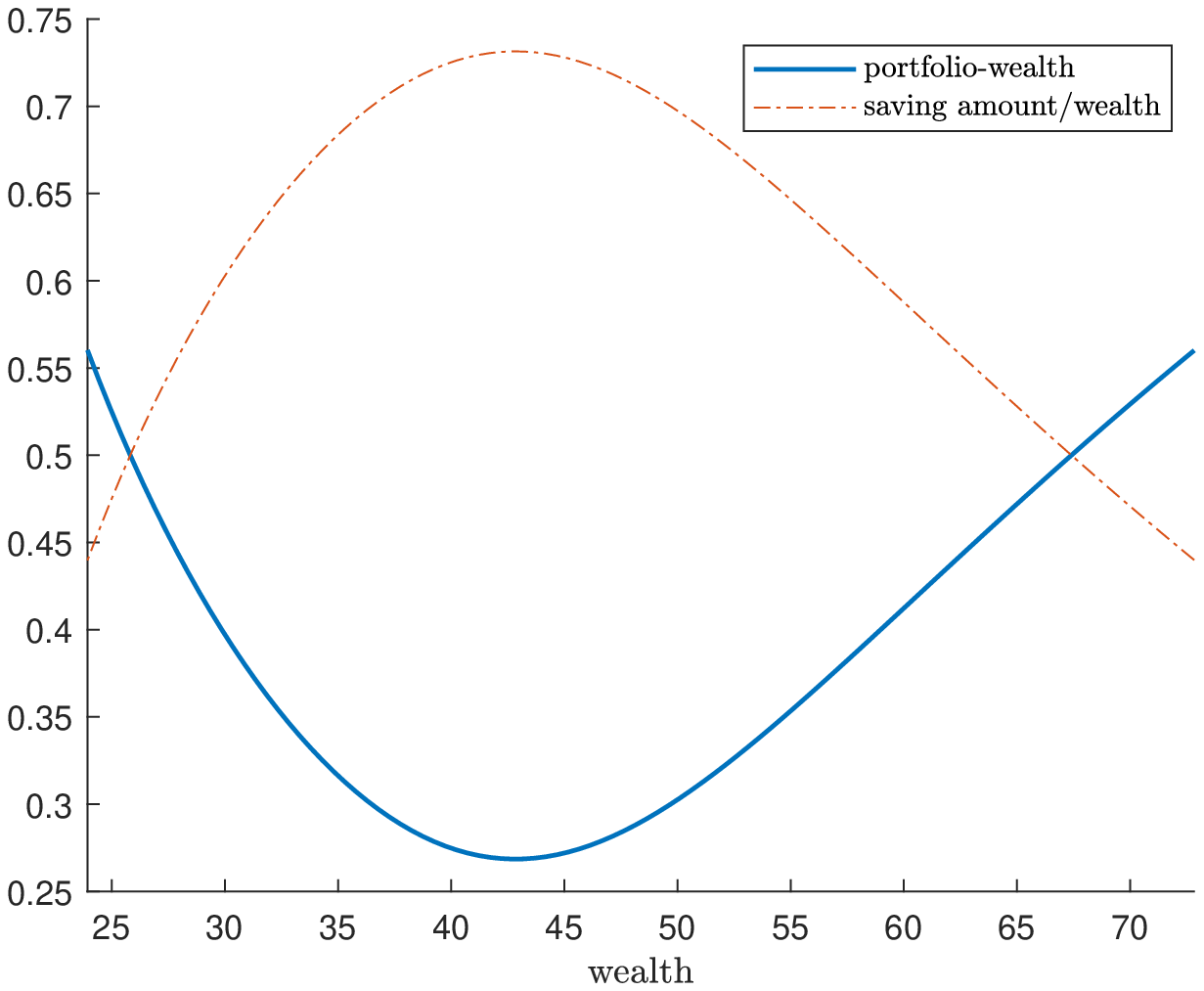}}
	\subfigure[RCRRA]{\includegraphics[scale=0.3]{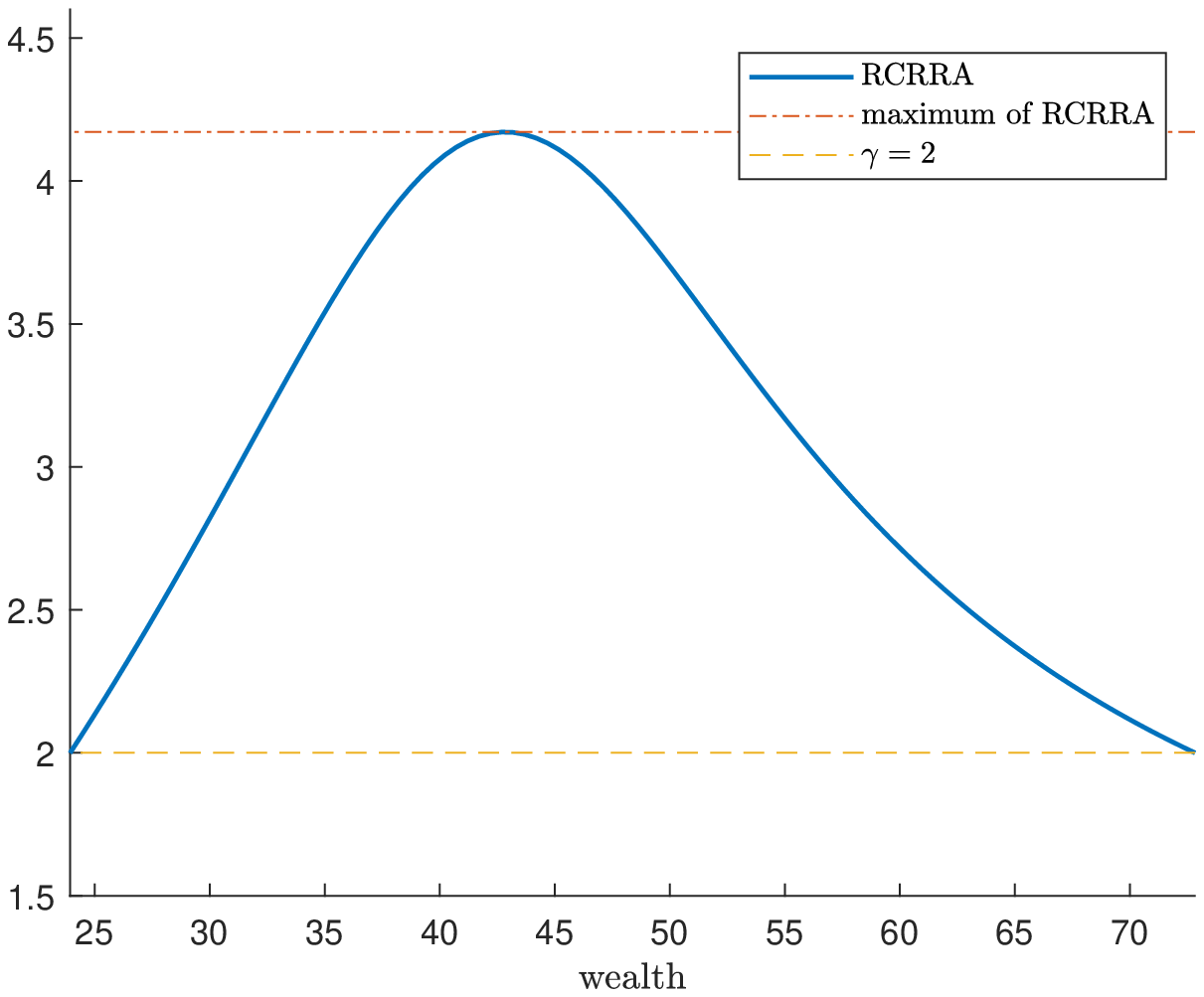}}
	\caption{Parameter value are as follows : $\delta=0.02,r=0.015,\mu=0.085,\sigma=0.25,\gamma=2,  X=50, c=1,\alpha=5, \beta=100$. In this case, $\underline{x}=23.93$ and $\bar{x}=72.82$. In particular, $\hat{X}=42.70$ in panel (c). The maximum RCRRA value is 4.1712.
	\label{fig-saving-portfolio-RCRRA}
	}
\end{figure}

While decomposition \eqref{portfolio-decompose} is intuitive, putting  \eqref{opt-portfolio} into  \eqref{portfolio-decompose} is complicated and thus it is not much informative. In order to better understand risky investment pattern, we define the {\it revealed coefficient of relative risk aversion (RCRRA)} as follows:
\begin{equation}\label{RCRRA00}
    \mbox{RCRRA}(t) \equiv \dfrac{\mu -r}{\sigma^2} \dfrac{X_t^*}{\pi^{*}_{t}}.
\end{equation}
The RCRRA is the level of relative risk aversion inferred from the
agent's portfolio allocation at time $t$ by outsiders who does not
know the agent's actual risk aversion. If the agent is unconstrained
($\alpha =\beta = 0$, i.e.,  the Merton case), the RCRRA is always the same as the
agent's true relative risk aversion $\gamma$. However, in general
$\mbox{RCRRA}$ is time-varying. In fact, it has maximum and
minimum values as shown in Theorem \ref{thm:RCRRA}.
The next theorem provides the properties of RCRRA and confirms the above intuition regarding the hedging demand.

\begin{thm}\label{thm:RCRRA}\begin{itemize}  \item[(a)] $  \mbox{RCRRA}(t)  \geq \gamma$ for all $t \geq 0$.
\item[(b)] Pick any $s \geq 0$. Consider the wealth process  $X_t^* \in (c\underline{x}, c \bar{x})$ for $t \geq s$, where $c_s^* = c$ is the optimal consumption level at time $s$ for some constant $c$.  Then, RCRRA$(t)$ approaches $\gamma$ as $X_t^*$ approaches either $c\underline{x}$ or $c\bar{x}$. Moreover, there exist $\hat{X} \in  (c\underline{x},  c\bar{x})$ such that RCRRA attains the maximum at $X_t = \hat{X}$. RCRRA decreases in wealth for $X_t^* \in (c\underline{x}, \hat{X})$ and increases in wealth for $X_t^* \in (\hat{X}, c\bar{x})$.
\end{itemize}

\end{thm}
\begin{proof}
The proof is given in Appendix \ref{sec:Appen_F}.
\end{proof}

Panle (c) of Figure \ref{fig-saving-portfolio-RCRRA} plots a typical RCRRA as a function of wealth. RCRRA is hump-shaped, which is opposite to the U-shaped risky share in Panel (b). It attains the minimum value $\gamma$ at the two ends of interval  $[c\underline{x}, c \bar{x}]$ and the maximum value at $\hat{X} \in (c\underline{x}, c \bar{x})$. As described in Theorem  \ref{thm:RCRRA}(b), RCRRA increases in wealth when the wealth level is smaller than $\hat{X}$ and it decreases in wealth  when the wealth level is greater than $\hat{X}$.

\section{Risk Attitude and Risky Share \label{sec:imp-investment}}

This section investigate the properties of the optimal risky investment policy obtained in Section  \ref{sec:optimal-strategies1} in detail.

\subsection{The consumption partial irreversibility model can explain the puzzle why the risky share of households is low \label{sec:why-risky-share-low}}

A puzzle in the classical dynamic portfolio selection literature is that while households usually hold $ 6 - 20\%$ in equity (conditional on participation, up to 40\%), standard models predicts much higher values. For example,  see the first row of Table \ref{ratio-comparison} that is the Merton case with CRRA preference with about $7\%$ risk premium and $25\%$ volatility. The risky share is  fairly high such as 124\% and 75\% when $\gamma = 0.9$ and $1.5$, respectively, while these values fall into the reasonable range of risk aversion estimated by the recent literature (see the first paragraph of Section  \ref{sec:reconcile-gap}).
   \begin{table}[h]
    \centering
    \label{tab:15}
    \begin{footnotesize}
    \begin{tabular}{c|ccccc}
        \hline\hline
        & $\gamma = 0.9$ &$\gamma = 1.5$ & $\gamma =3.5$    & $\gamma = 6$ &$\gamma = 10$ \\
        \hline \hline
      $(\alpha, \beta) =(0,0)$     & 124\%  & 75\%  & 32\%   &  19\%   & 11\%    \\
        \hline
      $(\alpha, \beta) =(5,10)$    & 100\% -- 124\%  & 55\% -- 75\%  & 24\% -- 32\% & 11\% -- 19\% & 9\% -- 11\%          \\
         \hline\hline
      $(\alpha, \beta) =(25,100)$  & 48\% -- 124\%   & 27\% -- 75\% & 14\% -- 32\% & 9\% -- 19\% & 6\% -- 11\%          \\
         \hline\hline
      $(\alpha, \beta) =(49,100)$  & 15\% -- 124\%   & 9\% -- 75\%  & 6\% -- 32\%  &  4\% -- 19\% & 2\% -- 11\%          \\
         \hline\hline
    \end{tabular}
\end{footnotesize}
    \caption{The risky share for different values of $(\alpha, \beta)'s$: the first row corresponds to the standard Merton case.  The other parameters values are $\delta=0.02,r=0.015,\mu=0.085$, and $\sigma=0.25$. \label{ratio-comparison}}
\end{table}

The optimal risky share in our case is not constant, but is U-shaped (or V-shaped) (see Panel (b) in Figure \ref{fig-saving-portfolio-RCRRA} or Panels (b) and (d) in Figure \ref{RCCRA_riskyshare}). The maximum value of the ratio of risky asset holdings is the same as that of the Merton case only when the wealth process hits the $(s, S)$  boundary. But, it is a measure zero event in a sample path. In most of times, the risky share is strictly smaller as seen in Figure \ref{fig-saving-portfolio-RCRRA}(b). As we will discuss in Section \ref{sec:RCRRA-large-when}, the ratio of risky asset holdings tends to be close to the minimum levels in times when there are moderate shocks and the wealth fluctuation is also moderate (equivalently, the RCRRA tends to be high for those times). For example, a household with $\gamma = 1.5$ and $(\alpha, \beta ) = (49,100)$ may keep the risky share close to 9\% for a long time if only moderate good and bad shocks alternates in the market for the same period time. The ratio gradually increases as the wealth process gradually increases since the equity premium is positive. However, as mentioned before, 75\% ratio happens fairly infrequently and the ratio should be much lower than 75\% for most of times. This means that the utility cost model can generate the household portfolio holdings consistent with the data from a plausible level of risk aversion (on which will be discussed in Section  \ref{sec:reconcile-gap} below).

\subsection{Time-varying risk attitude in our model can reconcile the gap among the risk aversion parameter values widely used in different literature \label{sec:reconcile-gap}}

There is a vast literature on measuring risk aversion. The most widely accepted values are between 0.7 and 2. More recent literature tends to argue that risk aversion is close to 1 or even less than 1 (For example, see \cite{C2006}, \cite{LMN2008}, \cite{BT2012}, and \cite{GH2015}).

\begin{table}[h]
    \centering
    \begin{tabular}{c|ccccc}
        \hline\hline
       Actual Risk aversion & $\gamma = 0.9$ &$\gamma = 1.05$ & $\gamma = 1.15$    & $\gamma = 1.3$ &$\gamma = 1.5$ \\
        \hline \hline
        ($\alpha$, $\beta$)     &  (40,10000) & (40,6000) & (40,4500) & (40,3200) & (40, 2200) \\
        \hline
        ($\alpha$, $\beta$)     &  (45,5000) & (45,3000) & (45,2200) & (45,1500) & (45, 1100) \\
         \hline
        ($\alpha$, $\beta$)     &  (49,1000) & (49,500) & (49,400) & (49,260) & (49, 180) \\
         \hline\hline
    \end{tabular}
    \caption{Actual risk aversion and the calibrated $(\alpha, \beta)'s$ to generate RCRRA with maximum 13: The other parameters values are $\delta=0.02,r=0.015,\mu=0.085,\sigma=0.25,  X=50$, and $c=1$. \label{calibration-RCRRA}}
\end{table}

However, most of asset pricing literature takes the level of risk aversion as about 10 or higher for the calibration analysis (e.g., \citet{BKY}, and \cite{CGM}). If we consider the investment aspect of these asset pricing models, it is related to the point made in Section \ref{sec:why-risky-share-low} in the sense that $\gamma =10$ results in the risky share consistent with the empirical observation for households' stock holdings (see the column for $\gamma =10$ when $(\alpha, \beta) = (0, 0)$ in Table \ref{ratio-comparison}).

In summary, the estimated results directly tell that risk aversion is small around unity while the calibration exercise in asset pricing indicates that the plausible level of risk aversion should be much higher. Our model can reconcile the gap in these two lines of literature. Table \ref{calibration-RCRRA} shows our numerical exercises. It calibrates $(\alpha, \beta)$ values that generate 13 as the maximum RCRRA for each risk aversion around unity with other widely accepted market parameter values given in the table.

Note that the conventional habit model can also generate time-varying risk aversion. However, it requires to define a internal or external habit stock process, which are usually ad hoc. Our model can perform a precise calibration analysis that can fit the maximum and minimum values of risk aversion.

\begin{figure}[h]
	\centering
    \subfigure[$\alpha=5$]{\includegraphics[scale=0.4]{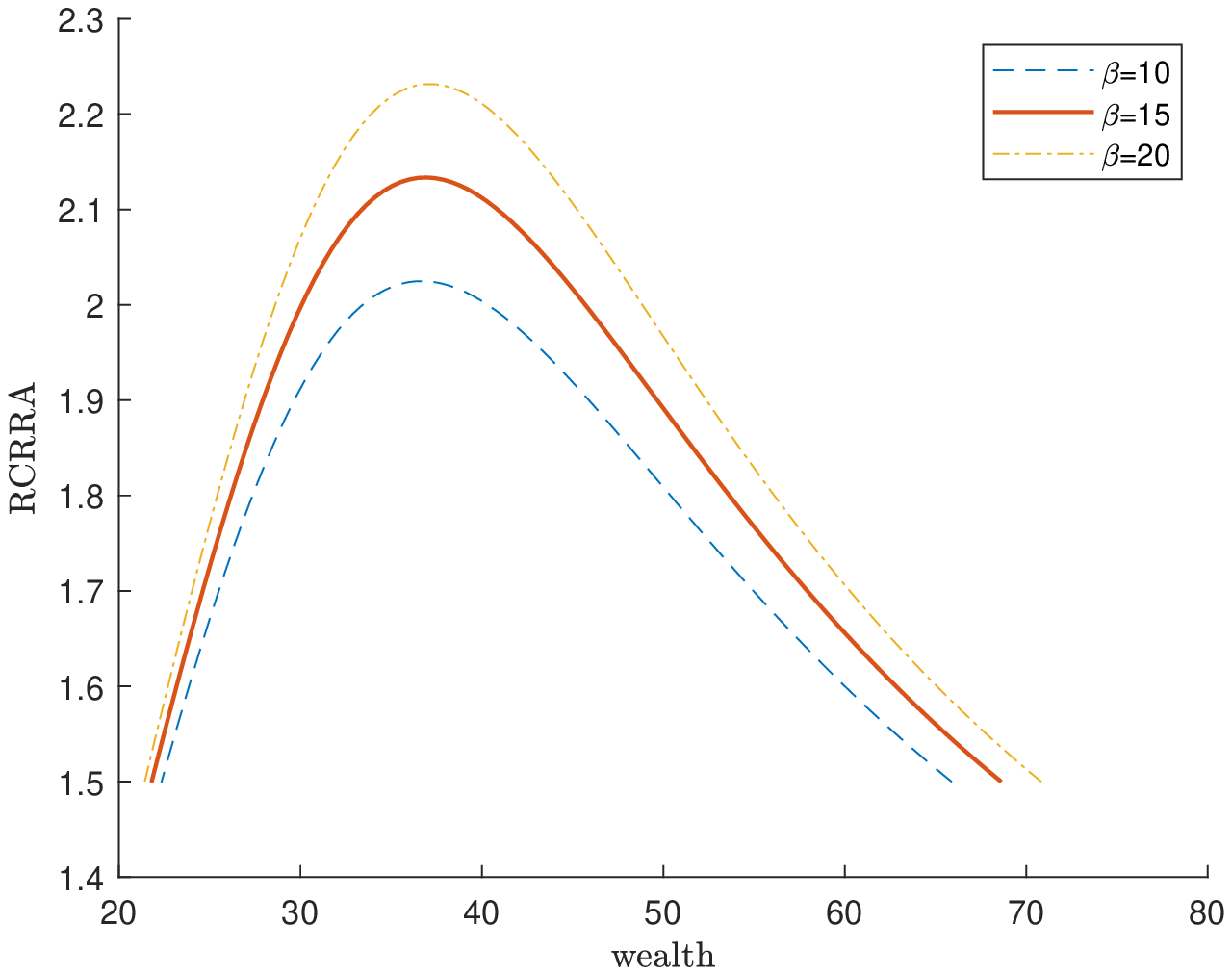}}
	\subfigure[$\alpha=5$]{\includegraphics[scale=0.4]{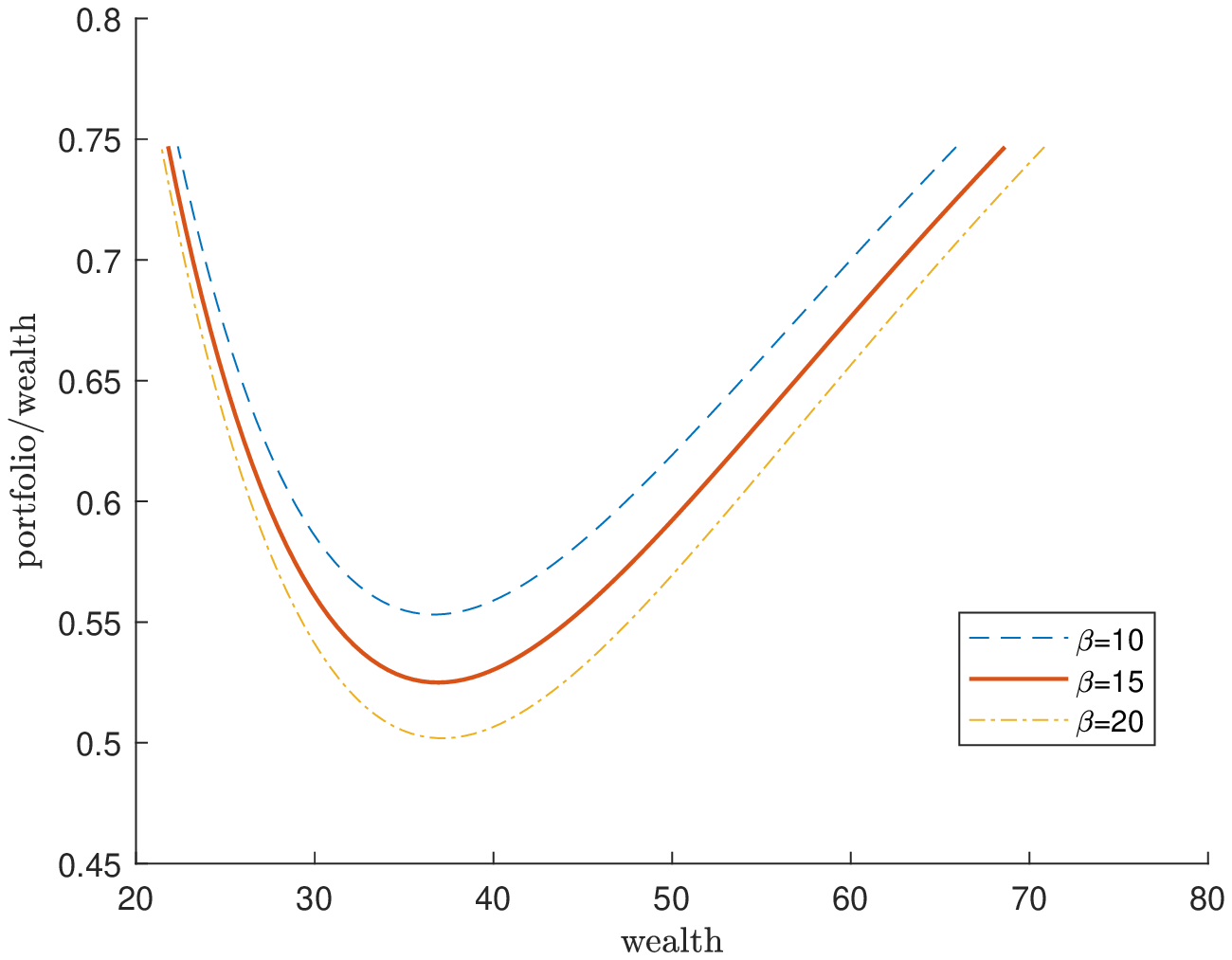}}
	\subfigure[$\alpha=49$]{\includegraphics[scale=0.4]{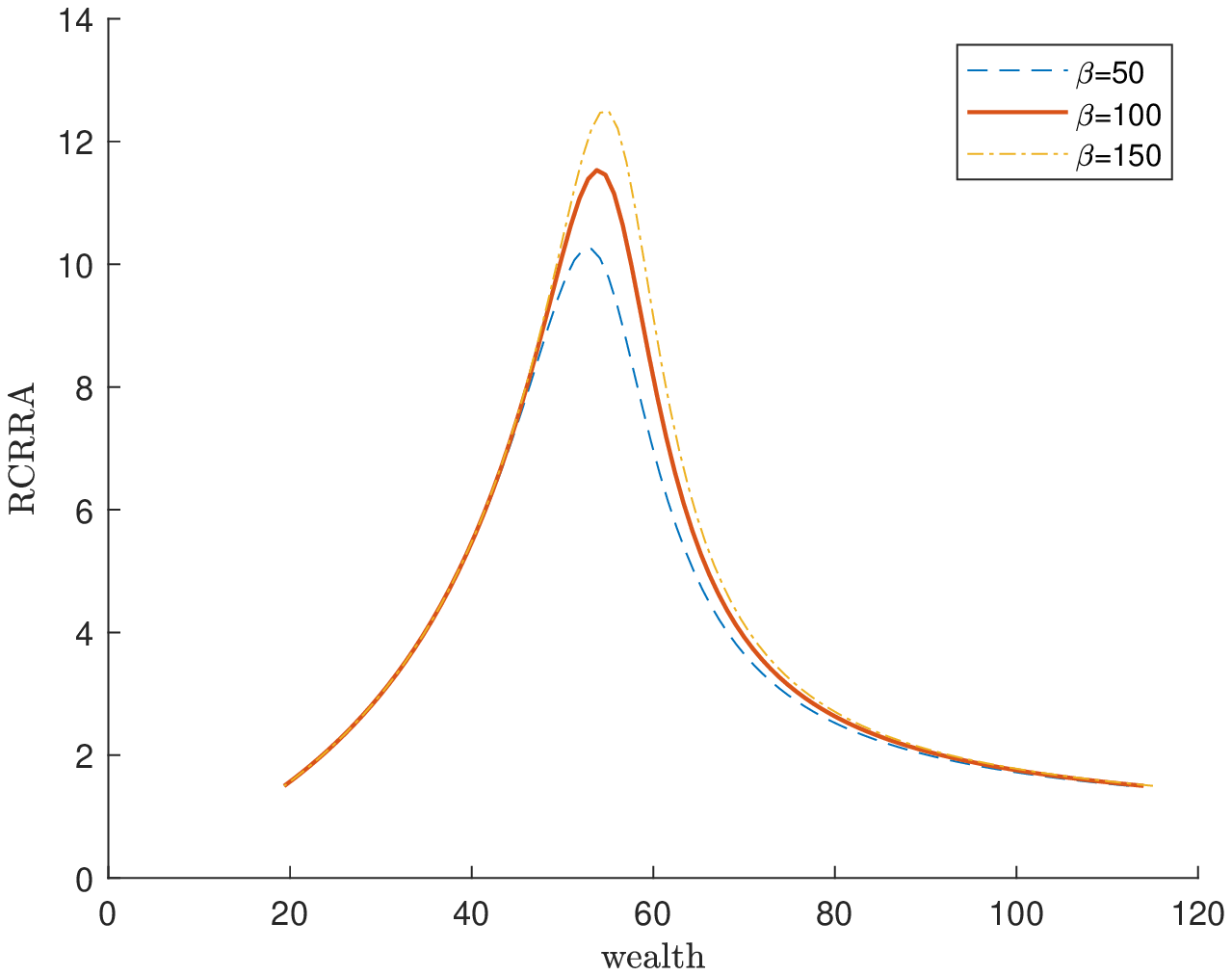}}
	\subfigure[$\alpha=49$]{\includegraphics[scale=0.4]{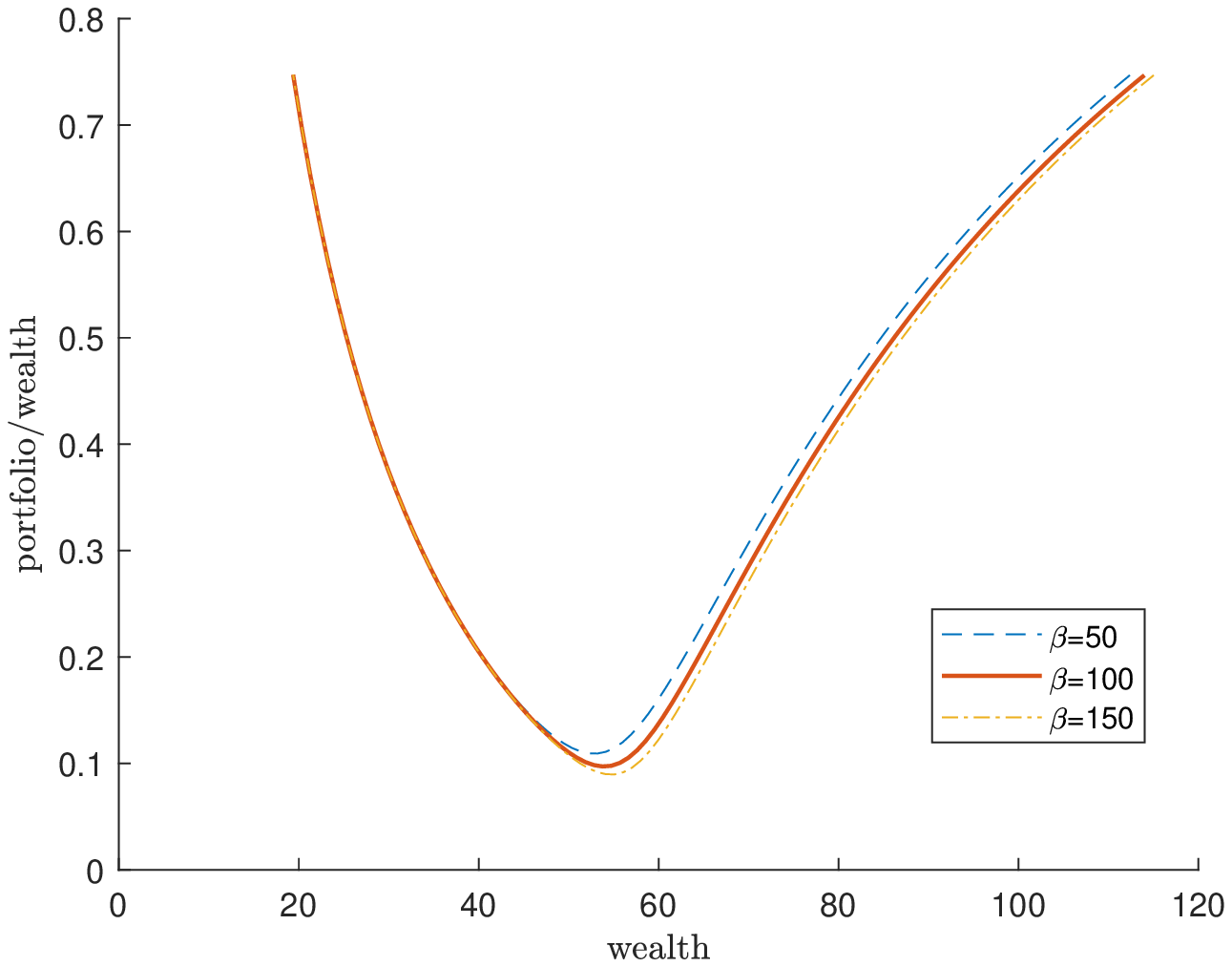}}
	\caption{Panels (a) and (c) plot RCRRA curves for shares for different $\alpha$ and $\beta$ values. Panels (b) and (d) plot the corresponding risky share, respectively.  Other parameter values are as follows: $\delta=0.02,r=0.015,\mu=0.085,\sigma=0.25,\; c=1$ and $\gamma=1.5$.	\label{RCCRA_riskyshare}}
\end{figure}

\subsection{RCRRA can be very large while actual risk aversion is small}

\begin{table}[h]
	\centering
	\begin{tabular}{c|cccc}
		\hline\hline
		($\alpha$, $\beta$)     & $\sigma = 0.25$ &$\sigma = 0.3$ & $\sigma=0.35$    & $\sigma = 0.4$\\
		\hline \hline
		(5, 10)     &  1.116 & 1.164  & 1.208 & 1.249\\
		\hline
		(5, 100)    &  1.782 & 1.974  & 2.164 & 2.358\\
		\hline
		(29, 100)   &  2.502 & 2.866  & 3.250 & 3.668\\
		\hline
		(49, 100)   &  7.459 & 9.607  & 12.40 & 16.07 \\
		\hline
	    (49, 1000)  &  13.18 & 18.48  & 26.20 & 37.46\\
		\hline\hline
	\end{tabular}
	\caption{The maximum values of RCRRA. The other parameters values are $\delta=0.02,r=0.015,\mu=0.085,\;\gamma=0.9,\;  X=50$, and $c=1$. \label{max-RCRRA-alpha-beta-sigma}}
\end{table}

In Figure \ref{fig-saving-portfolio-RCRRA}(c), the actual risk aversion is $\gamma = 3$ and the maximum value of RCRRA is about 5.32. This difference does not seem large. However, the difference between the actual level of risk aversion and the RCRRA value can be dramatically large. Figures \ref{RCCRA-simulation1} and \ref{RCCRA-simulation2} plot a sample path of RCRRA$(t)$. While the actual risk aversion level is $\gamma =0.9$ in these figures, RCRRA can take more than 38.  The optimal wealth path touches neither $c\underline{x}$ nor $c\bar{x}$ (See the left panel in each figure), which means that the consumption path stays constant in Figures \ref{RCCRA-simulation1} and \ref{RCCRA-simulation2}.

The maximum value of RCRRA increases in $\alpha$ and $\beta$, respectively\footnote{Note that $\beta$ is more important since $\alpha$ cannot exceed  $\frac{1}{\delta}$ because of Assumption \ref{assumption-alpha}.} (see Table \ref{max-RCRRA-alpha-beta-sigma}). A high $\alpha$ or $\beta$ implies a high cost when the agent changes the consumption level. Therefore, the intertemporal hedging deman become large (with negative sign) to crow out the myopic demand (see Eq. \eqref{portfolio-decompose} and explanation below the equation), which makes the RCRRA level higher. In addition, a high risk, i.e., a high volatility $\sigma$ increase the maximum level of RCRRA by the same token.

\begin{figure}[h]
	\centering
	\subfigure{\includegraphics[width=13cm, height=4cm]{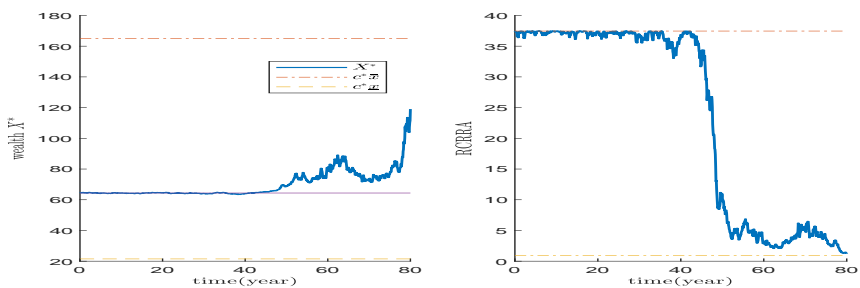}}
	\caption{RCRRA can be very large while the actual risk aversion is very low ($\gamma = 0.9$).
		Consumption stays constant in this sample path. Parameter values are as follows: $\delta=0.02,r=0.015,\mu=0.085,\sigma=0.4, T=80,  X=64.38,\; c=1,\alpha=49, \beta=1000.$ The maximum RCRRA value is 37.46.
		\label{RCCRA-simulation1}}
\end{figure}

\begin{figure}[h]
	\centering
	\subfigure{\includegraphics[width=13cm, height=4cm]{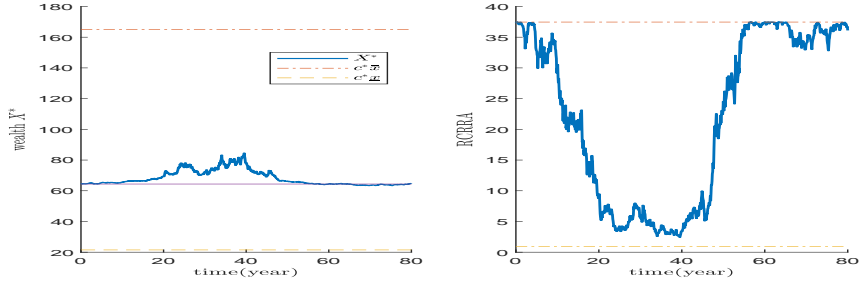}}
	\caption{RCRRA can be very large while the actual risk aversion is very low ($\gamma = 0.9$).
		Consumption stays constant in this sample path. Parameter values are as follows: $\delta=0.02,r=0.015,\mu=0.085,\sigma=0.4, T=80,  X=64.38,\; c=1,\alpha=49, \beta=1000.$ The maximum RCRRA value is 37.46.
		\label{RCCRA-simulation2}}
\end{figure}


\subsection{RCRRA tends to be large during the times when moderate shocks continuously arrive. RCRRA becomes small for the times with large shocks. \label{sec:RCRRA-large-when}}

We know that RCRRA can be very large while the given risk aversion is very small. Then, when does it happen? In habit models, it happens in the downturns, in particular, when the current consumption level is very close to the habit stock. The mechanism to induce the difference between the actual risk aversion and the RCRRA in our model is very different from that of habit models.

Notice in our model that RCRRA tends to be small during the times when there are consecutive large shocks and thus in these times there are high fluctuations in wealth. For example, those times are between $t=45$ and $t=60$ in Figure \ref{RCCRA-simulation1} and between $t=20$ and $t=50$ in Figure \ref{RCCRA-simulation2}. On the other hand, it tends to be high during the time when there are little or modest fluctuations in wealth due to moderately alternating good and bad shocks. For example, those times are before $t=45$ in Figure \ref{RCCRA-simulation1} and after $t=50$ in Figure \ref{RCCRA-simulation2}. There are other examples in the Appendix.

In other words, high cost for consumption adjustment amplifies risk aversion over small or moderate shocks (particularly during the times when the wealth level stays in the mid range of (s, S) band), making the household look very conservative. However, the household may look very aggressive (if her actual risk aversion is very low), showing RCRRA close to her actual risk aversion, during the time when there are large shocks.
Notably, if the household has a long sequence of moderate good shocks, then the risk aversion gradually decreases (equivalently the household  gradually increases risky asset holdings). This pattern looks like that the household gains more confidence due to the success in the stock market.

In this sense, our utility cost model can generate substantial risk aversion in times of moderate risk events and great risk-taking in times of large risk events. This result is quite consistent with the puzzle reported in the behavioral literature.

\subsection{Change in risky share relative to change in wealth}

\begin{figure}[h]
	\centering
	\subfigure[bull market]{\includegraphics[scale=0.25]{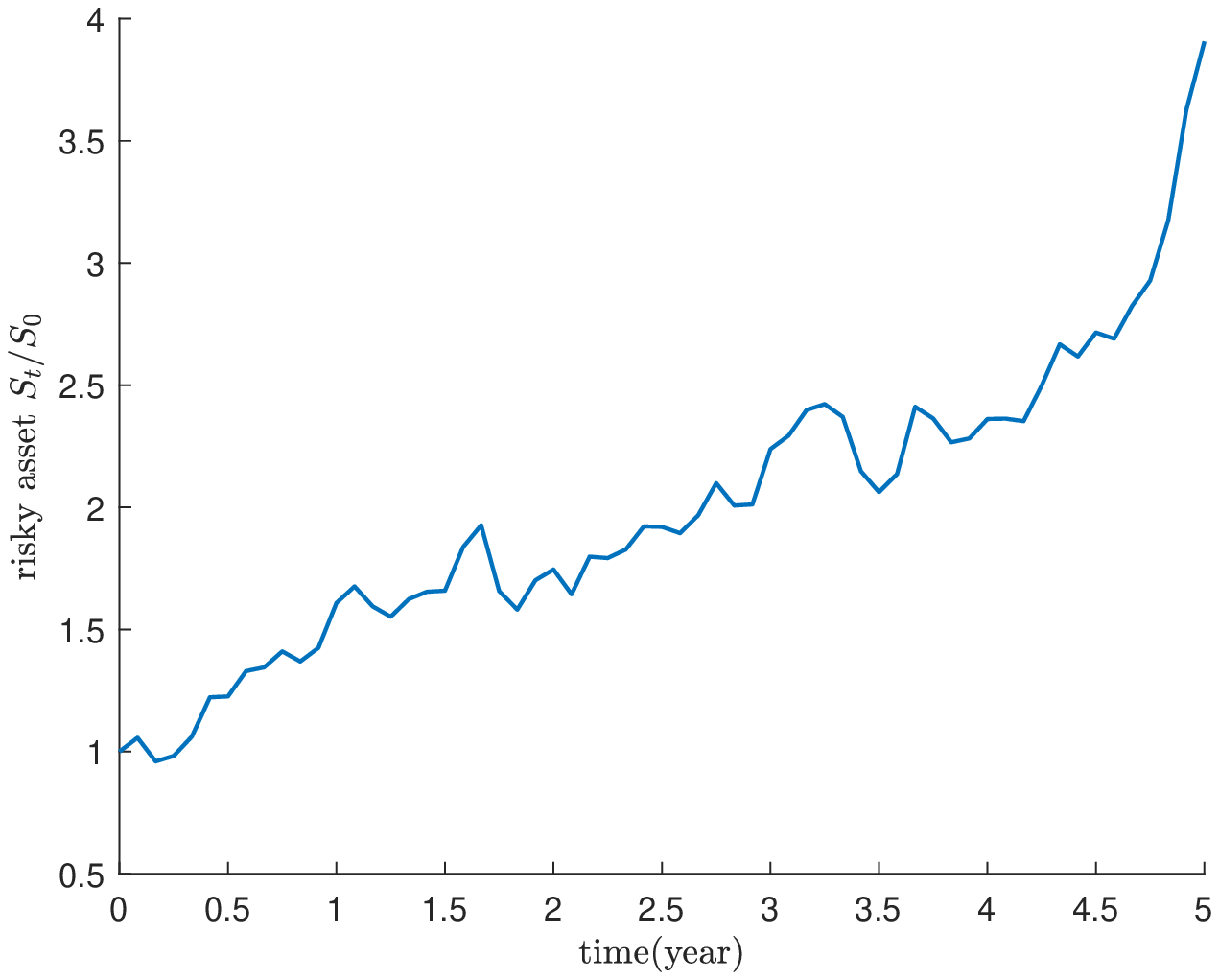}}
	\subfigure[intermediate ]{\includegraphics[scale=0.25]{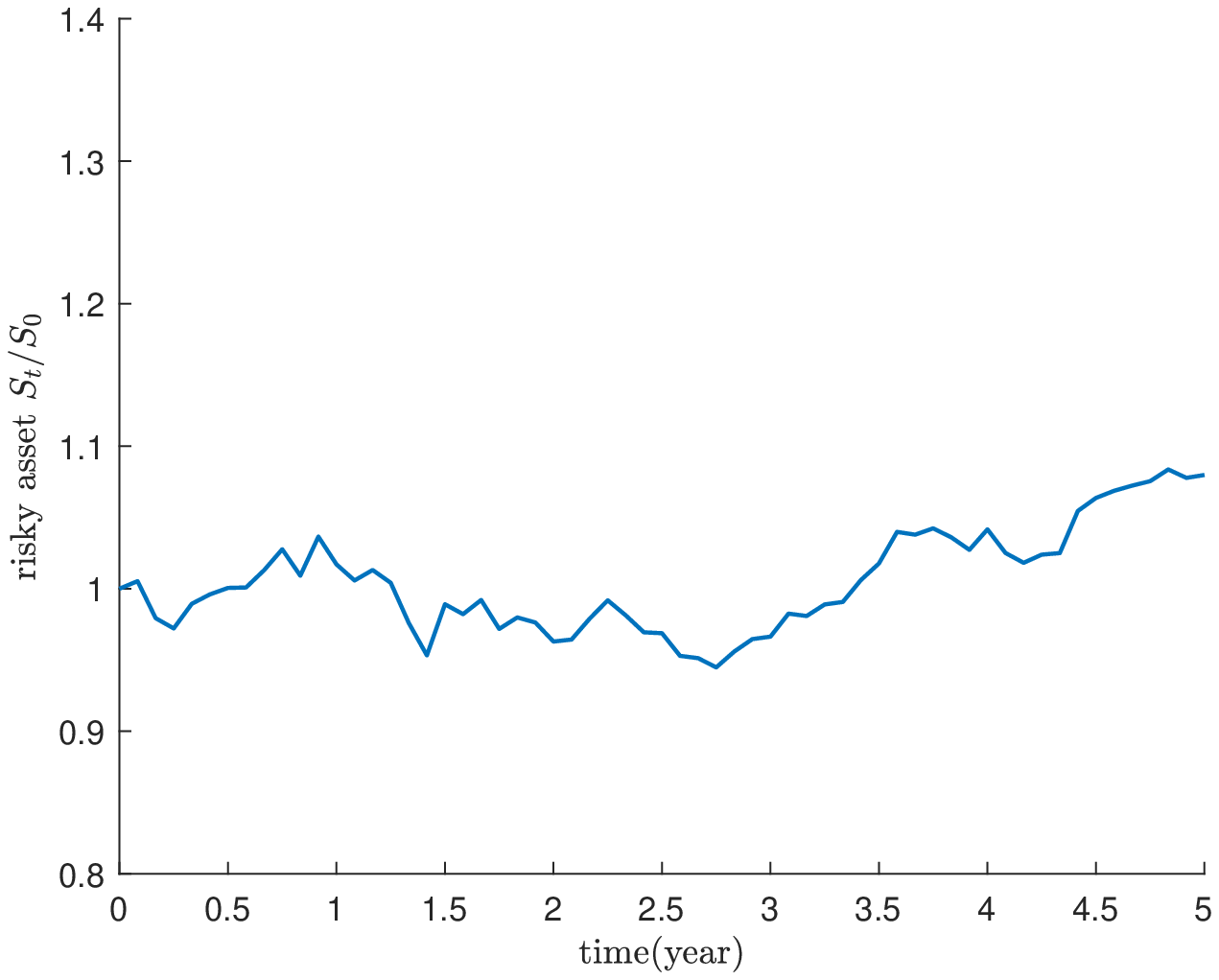}}
	\subfigure[bear market]{\includegraphics[scale=0.25]{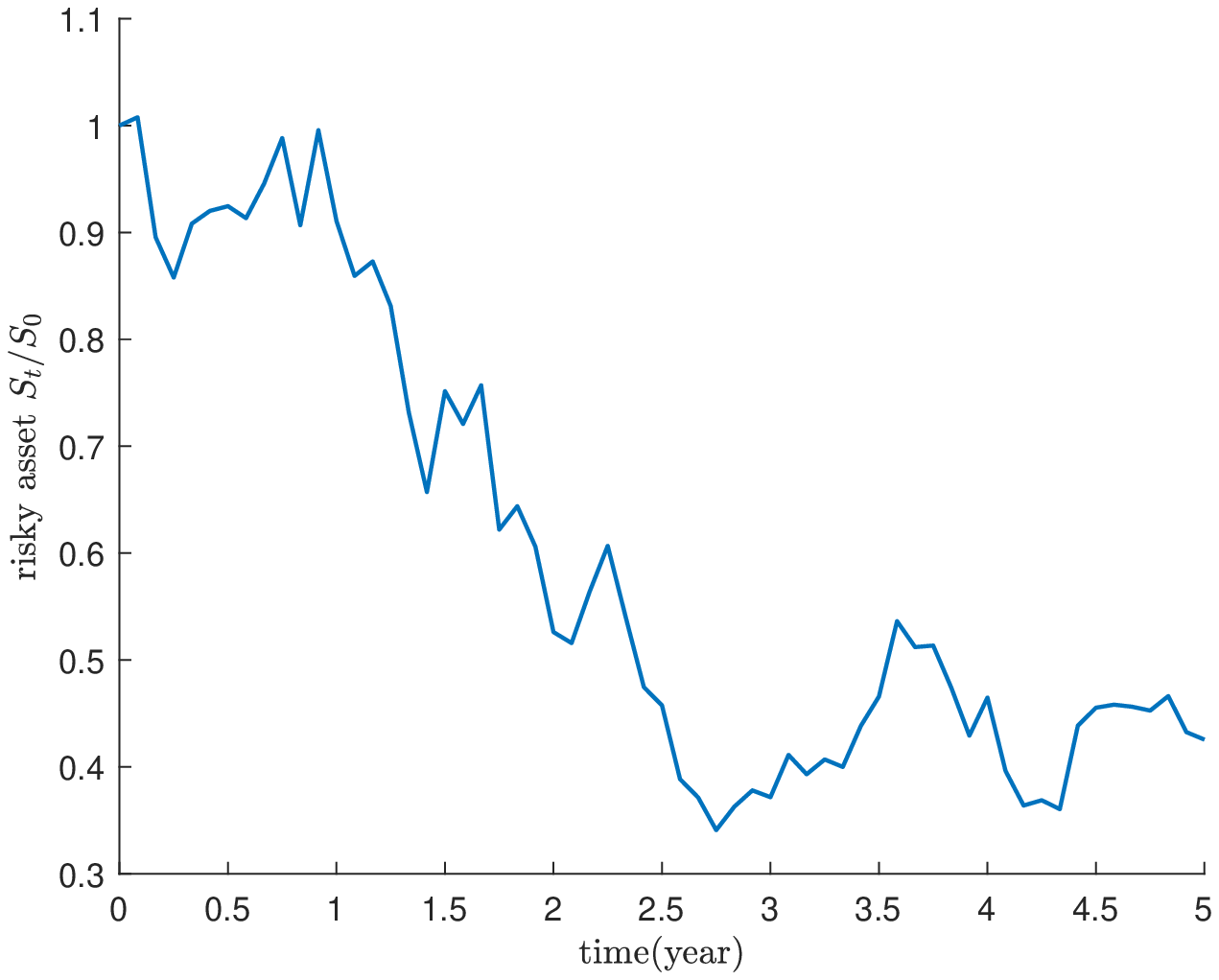}}
	\subfigure[highly fluctuating]{\includegraphics[scale=0.25]{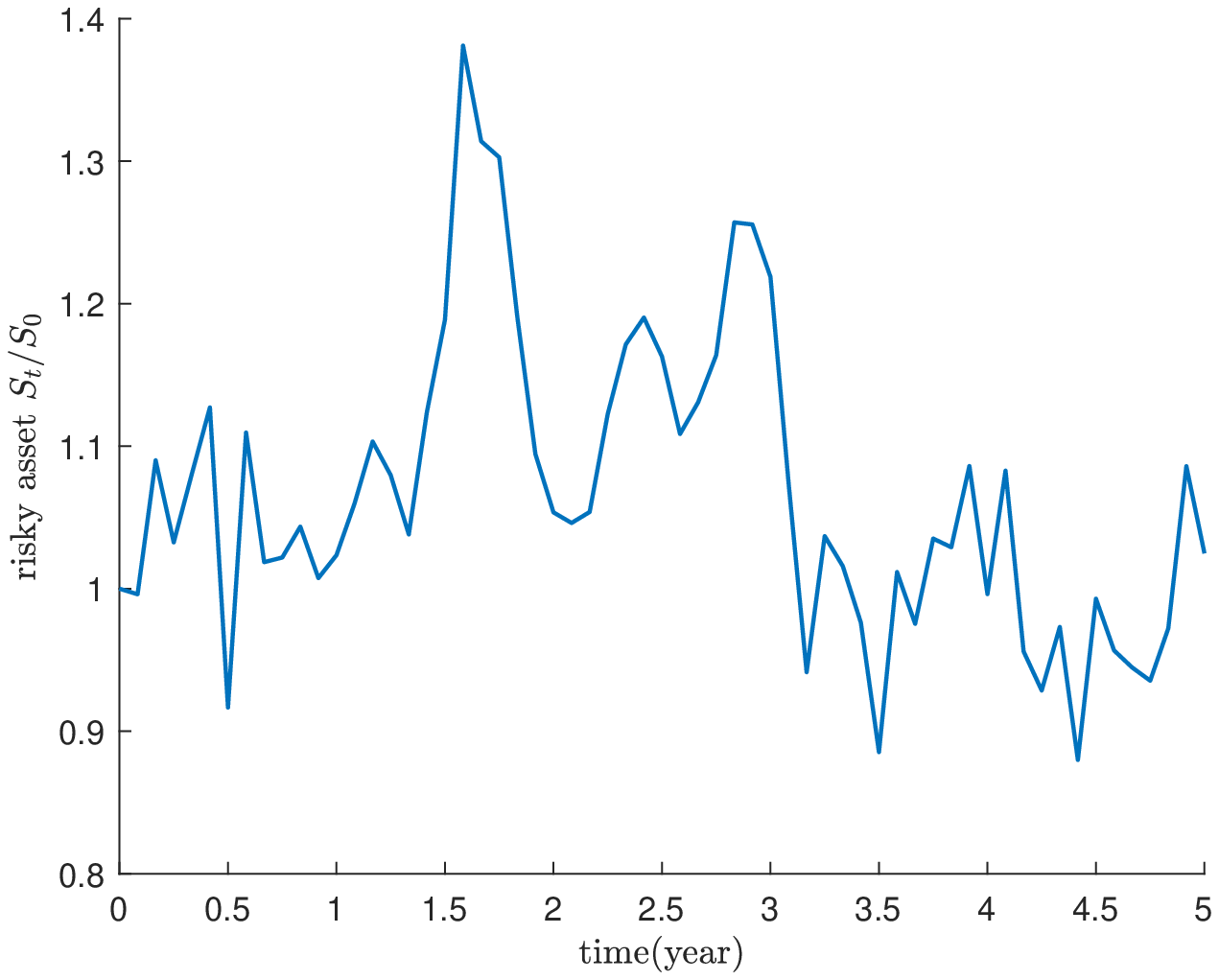}}
	\caption{Four different cases of stock price evolution: Parameter value are as follows : $\delta=0.02,r=0.015,\mu=0.085,\sigma=0.25,\gamma=1.5$ and $c=1$.
		\label{four_scenario}
	}
\end{figure}

Here we investigate how the risky share changes relative to an increase in wealth. Note that models with habit, commitment, or DRRA (decreasing relative risk aversion) often predict that the relationship is positive while  standard models with CRRA preference predicts no relationship. The empirical literature, however, seems rather inconclusive. For example, \cite{CCS2009} and \cite{CS2014} favor the habit or commitment or DRRA models, and \cite{CP2011} and \cite{BN2008} favor the classical CRRA model and show even slightly negative relationship. We shed on light on the debate in the literature. We argue that whether the change in wealth has positive, negative, or no impact on the risky share depends on what kind of a sample path (time-series of the stock market data) researchers use although the data generating process is fixed.

Figure \ref{four_scenario} shows the four types of sample paths we consider. Panel (a) is a typical long-term bull market, , Panel (b) is intermediate, Panel (c) is a typical bear market, and Panel (d) is the highly-fluctuating market over time. We will show that the impact of an wealth change on the risky share is different for each sample path. To do so, we generate the population with random $(\alpha, \beta)'s$ and simulate their wealth and risky share over time.

Basically our estimation analysis follows that of \cite{BN2008}. Consider the following equation
\begin{equation}\label{eq:object}
\Delta_k \log{\frac{\pi_t}{X_t}} = \rho \Delta_k \log{X_t},
\end{equation}
where $\Delta_k$ denotes a $k$-period(year) first-difference operator, $\Delta_k y_t \equiv y_t-y_{t-k}$.  Below we briefly explain how we regress Eq. \eqref{eq:object}.

\noindent {\bf (Step 1)} Generate initial consumption/wealth distributions of $N$ individuals.
\begin{itemize}
	\item We divide the interval $[0,T]$ into $12\times T$ subintervals with end points $t_j=1,2,...,12\times T$. (Here, we assume that $T$ is a positive integer)
	\item  We set equal initial consumption, $c_0=1$ for each individual and generate each individual's initial consumption $x_0$ randomly according to a uniform distribution over $(c\underline{x},c\bar{x})$.
	\item For each $i,\;i=1,2,...N$ and given pairs $(m_{\alpha},v_{\alpha})$, $(m_{\beta},v_{\beta})$, we generate log-normally distributed random variables $\alpha_i$ and $\beta_i$ whose the mean and variance are $(m_{\alpha},v_{\alpha})$ and $(m_{\beta},v_{\beta})$, respectively.\footnote{The log-normal distribution implies that there are households having fairly large values of $\alpha$'s although their density in the population is very small. We drop those household whose $\alpha$ values violates Assumption \ref{assumption-alpha}.}
	\item Generate a $12\times T$ random vector $\omega$ that follows a standard normal distribution. Using this vector $\omega$, we generate the process of the risky asset returns $\Delta S_t/S_t$, and the dual process $y_t^*$ in Proposition \ref{pro:consumption} for all the $N$ individuals. By Proposition \ref{pro:consumption}, we can simulate the optimal wealth and portfolio processes of $N$ individuals.
\end{itemize}
\noindent \textbf{(Step 2)} Compute changes in the ratio of risk asset holdings and changes in wealth.
\begin{itemize}
	\item Let $(X^1,P^1), (X^2,P^2),...,(X^N,P^N)$ be the simulated wealth/portfolio processes for $N$ individuals in our utility cost model obtained in \textbf{(Step 1)}. (Note that $X^i$ and $P^i$ are $(12\times T+1)$ random vectors for $i=1,2,...,N$).
	\item For given $T$ and $k$, there are $(T-k+1)$ numbers of $\Delta_k$.
	
	For $i=1,2,...,N$ and $j=1,2,...,(T-k+1)$
	\begin{footnotesize}
	\begin{eqnarray*}
	\begin{split}
	DR(i,j)&=\Delta_k \log{\frac{P_{j+k}^i}{X_{j+k}^i}}=\log{\frac{P^i(12\times(j+k-1)+1)}{X^i(12\times(j+k-1)+1)}}-\log{\frac{P^i(12\times(j-1)+1)}{X^i(12\times(j-1)+1)}},\\
	X(i,j)&=\Delta_k \log{{X_{j+k}^i}}=\log{{X^i(12\times(j+k-1)+1)}}-\log{{X^i(12\times(j-1)+1)}}.
	\end{split}
	\end{eqnarray*}
    \end{footnotesize}
	\item We regress Eq. \eqref{eq:object} with $OLS$ using the simulate results $DR$ and $X$.
\end{itemize}

\begin{table}[h]
	\centering
	\begin{footnotesize}
	\begin{tabular}{cc|cccc}
		\hline\hline
		$(m_{\alpha},v_{\alpha})$ &$ (m_{\beta},v_{\beta})$  &  bull markets &  intermediate  & bear markets   &  highly fluctuating   \\
		\hline \hline
		(5, $5^2$)  & (50, $20^2$)   &  $0.1369^{***}$ & 0.0208\;(0.22)  & $-0.0757^{***}$ 	  &$-0.0679^{***}$ \\
		\hline
		(10, $5^2$)  & (100, $20^2$) &  $0.0999^{***}$ & $0.0829^{***}$ & $-0.0746^{***}$  & $-0.0699^{***}$\\
		\hline
		(10, $5^2$) & (150, $50^2$) &  $0.1107^{***}$& $0.1261^{***}$ & -0.0478\;(0.01)   & -0.0271\;(0.26) \\\hline
		(15, $10^2$) & (150, $50^2$) &  $0.1118^{***}$& $0.1229^{***}$  &-0.0387\;(0.04)  & -0.0341\;(0.17) \\\hline
		(15, $10^2$) & (200, $50^2$) &  $0.0961^{***}$ & $0.1710^{***}$ & -0.0164\;(0.43)  & -0.0379\;(0.14) \\\hline
		&    & $N=1500$,   &$T=5$, $k=2$. & & \\
		\hline\hline
	\end{tabular}
	\end{footnotesize}	
	\caption{The regression coefficients for markets (a), (b), (c), and (d) in Figure \ref{four_scenario}: We generate the distribution of households with $\alpha$ and $\beta$ using the log-normal distributions with mean $m_{\alpha}$, $m_{\beta}$ and variance $v_{\alpha}$, $v_{\beta}$ respectively. The values in the parentheses are p-values when the p-value is greater than 1\%. $***$ means the p-value smaller than 1\%. The other parameters values are $\delta=0.02,r=0.015,\mu=0.085,\;\gamma=1.5$, and $c=1$. \label{BN_test}}
\end{table}

The  results are summarized in Table \ref{BN_test}. The regression results show the positive impact of the wealth increase on the risky share when the market is generally going up (Panel (a)) or intermediate with moderate up-and-downs. There is no wealth impact on the risky share or slightly negative (if any) when the market is generally going down (Panel (c)) or has a huge volatility (Panel (d)).
\begin{figure}[h]
	\centering
     \subfigure{\includegraphics[scale=0.6]{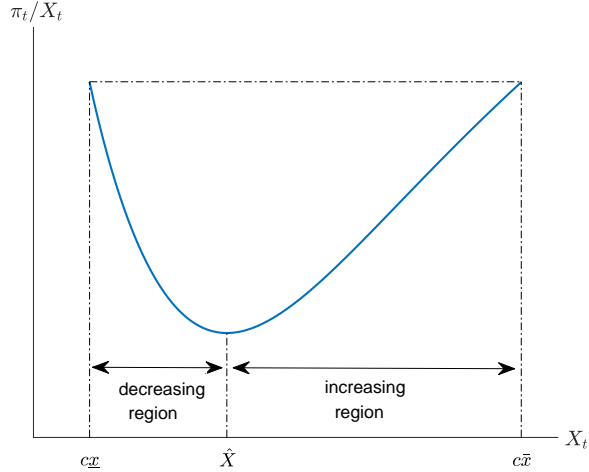}}
     \caption{The decreasing region and increasing region of the risky share in the (s, S) band. \label{RCRRA_region}}
\end{figure}
The intuition behind these results originates from the ratio of risky asset holding dynamics and wealth process.  The risky share is hump-shaped in each (s, S) band, which implies that there are two regions with the band: the one is the decreasing region in which the risky share decreases with wealth  and the other is the increasing region in which the risky share increases with wealth (see Figure \ref{RCRRA_region}). The decreasing region is in the left part of the (s, S) band and the increasing region is in the right part of the (s, S) band.  In general, the increasing region longer than the decreasing region since the risk premium is positive. In addition, recall Figure \ref{wealth-consumption-Ss} and its description below Theorem \ref{thm:wealth}.  In Figure \ref{wealth-consumption-Ss}, if there are consecutive good shocks, the wealth process tends to move in the following way: $X_0 \rightarrow  B \rightarrow B_u \rightarrow B_{uu}$. During the times of this journey, the wealth process will stay longer in the increasing region of each (s, S) band since there are more good shocks than bad shocks in size and amount. On the other hand, if there are consecutive bad shocks, the wealth process tends to move in the other way: $X_0 \rightarrow  A \rightarrow A_d \rightarrow A_{dd}$. During the times of this journey, the wealth process will stay longer in the decreasing region of each (s, S) band since there are more bad shocks than good shocks.

Having the above property in mind, if the market is doing well in the long-run as in Panel (a) of Figure  \ref{four_scenario}), over time the wealth process is more likely to stay in the increasing region of each (s, S) band. Thus, the risky share tends to increase with wealth during these times. On the other hand, if the market is doing poorly for a while as in Panel (c) of Figure \ref{four_scenario}), the wealth process is more likely to stay in the decreasing region of each (s, S) interval. Thus, the risky share decreases with wealth during these times.

The intermediate case (Panel (b) of Figure  \ref{wealth-consumption-Ss}) tends to be closer to the case of the bull market and the highly fluctuating period (Panel (d)) of  \ref{wealth-consumption-Ss}) tends to be closer to the case of the bear market. It is because, as mentioned above,  the right half part of the (s, S) band (increasing region) is generally longer than left half part of the (s, S) band (decreasing region).

\section{Consumption and Asset Pricing Moments \label{sec:imp-consumption}}

Here we provide two implications of optimal consumption. Our model can explain the excess sensitivity and excess smoothness puzzles. In addition, our model can generate the reasonable asset pricing moments such as auto-correlation coefficient of consumption, which is often pointed out as a weakness of habit models.

\subsection{Excess Sensitivity and Excess Smoothness}

The optimal consumption level is infrequently adjusted in our model (see, for example, Figures \ref{fig-sim-path1} and \ref{fig-saving-portfolio-RCRRA}). The excess smoothness appears since consumption in our model does not immediately respond to a permanent shock with a moderate level when the wealth level after the shock does not reach the (s, S) boundary. In other words, the household is less likely to change a consumption level with response to a small unexpected permanent shock (income), while they do increase or decrease the ratio of risky asset holdings for a positive (negative) shock.

The excess sensitivity also arises in our model by a similar reason. After a moderate good shock, the wealth level becomes higher than before. Although consumption does not immediately respond, it is more likely to increase later since the probability of increasing consumption in the next period becomes higher after the good shock. Therefore, our model can explain the excess sensitivity puzzle.

However, all the above arguments do not really work if the size of shock is large enough. For a large good shock, the wealth level immediately reaches the upper threshold, which makes consumption increase immediately. By the same token, the consumption level decreases immediately after a large bad shock. So, the consumption moves in the large shock events, following the permanent income hypothesis.

\subsection{Auto-correlation}

\begin{table}[h]
    \centering
    \begin{tabular}{c|ccccc}
        \hline\hline
        & $\mathbb{E}[\Delta c]$ & $\sigma(\Delta c)$ & EP     & Std of IMRS & AC1$(\Delta c)$ \\
        \hline \hline
        Data                                                                  & 0.0192                 & 0.0212             & -- & --      &  0.4600   \\
        \hline
        Our Model                                                                  & 0.0181                 & 0.0236             & 0.0052 & 0.0775      & 0.4900          \\ \hline
        \begin{tabular}[c]{@{}c@{}}Merton with temporal\\ aggregation\end{tabular} & 0.0200                 & 0.0858             & 0.0526 & 0.2996      & 0.1677          \\ \hline\hline
    \end{tabular}
    \caption{Mean, standard deviation, autocorrelation of the consumption growth rate, theoretical equity premium, and standard deviation of the IMRS. The parameter values as follows: $\delta=0.015,\;r=0.0086,\;\mu=0.0784,\;\sigma=0.2016,\;\gamma=3.5$, and $\alpha=5,\;\beta=10$. {The data in the first row is from \cite{BKY}, sampled on an annual frequency with the period from 1930 to 2008.} \label{table-ac1} }
\end{table}

Here we provide how the model can generate consumption data and try to match several asset pricing moments. By using the time series of aggregate consumption we compute the mean and standard deviation of consumption growth rates, the intertemporal marginal rate of substitution(IMRS) and the theoretical equity premium(EP). Based on our consumption model, we simulate optimal consumption processes for $N=100$ individuals following three steps:\\

\begin{table}[h]
	\centering
	\begin{tabular}{c|c|c|c|ccccc}
		\hline\hline
		$\delta$               & $\gamma$             & $\alpha$             & $\beta$            & $\mathbb{E}[\Delta c]$ & $\sigma(\Delta c)$ & EP     & Std of IMRS & AC1$(\Delta c)$ \\ \hline\hline
		\multirow{3}{*}{0.015} & \multirow{3}{*}{3.5} & \multirow{3}{*}{5} & 10                  & 0.0181                 & 0.0236             & 0.0052 & 0.0775      & 0.4900         \\
		&                      &                      & 15                  & 0.0182                 & 0.0232            & 0.0050 & 0.0758      & 0.4888         \\
		&                      &                      & 20                  & 0.0183                 & 0.0229             & 0.0049 & 0.0747      & 0.4881         \\ \hline\hline
		\multirow{3}{*}{0.015} & \multirow{3}{*}{3.5} & 0                    & \multirow{3}{*}{20} & 0.0182                 & 0.0233             & 0.0050 & 0.0762      & 0.4889         \\
		&                      & 5                  &                    & 0.0183                 & 0.0229             & 0.0049 & 0.0747      & 0.4881         \\
		&                      & 10                    &                    & 0.0184                 & 0.0226             & 0.0048 & 0.0736      & 0.4869         \\ \hline\hline
		\multirow{3}{*}{0.015} & \multirow{3}{*}{3.5} & 30                    & 100& 0.0191                 & 0.0217             & 0.0044 & 0.0699      & 0.4825         \\
		&                      & 50                  & 100                   & 0.0192                 & 0.0215             & 0.0044 & 0.0693      & 0.4808         \\
		&                      & 50                    &  1000                  & 0.0192                 & 0.0215             & 0.0044 & 0.0691      & 0.4802         \\ \hline\hline
		0.010                   & \multirow{3}{*}{3.5} & \multirow{3}{*}{5} & \multirow{3}{*}{10} & 0.0194                 & 0.0242             & 0.0054 & 0.0791      & 0.4888         \\
		0.015                  &                      &                      &                    & 0.0181                 & 0.0236             & 0.0052 & 0.0775      & 0.4900         \\
		0.050                   &                      &                      &                    & 0.0093                 & 0.0204             & 0.0039 & 0.0683      & 0.4953         \\ \hline\hline
		\multirow{4}{*}{0.015} & 0.9                  & \multirow{4}{*}{5} & \multirow{4}{*}{10} & 0.0738                 & 0.0963             & 0.0052 & 0.0775      & 0.4889         \\
		& 1.5                  &                      &                    & 0.0431                 & 0.0563             & 0.0052 & 0.0775      & 0.4898         \\
		&3.5                    &                      &                    & 0.0181                 & 0.0236             & 0.0052 & 0.0775      & 0.4900         \\
		& 10                   &                      &                    & 0.0063                 & 0.0082             & 0.0052 & 0.0775      & 0.4900         \\ \hline\hline
	\end{tabular}
	\caption{Mean, standard deviation, autocorrelation of the consumption growth rate, theoretical equity premium, and standard deviation of the IMRS. The parameter values as follows: $\;r=0.0086,\;\mu=0.0784$ and $\sigma=0.2016$.  \label{table-ac2} }
\end{table}

\noindent \textbf{(Step 1)} Generate initial consumption/wealth distributions of $N$ individuals.
\begin{itemize}
	\item We divide the interval $[0,79]$ into $2\times12\times79$ subintervals with end points $t_j=1,..,2\times12\times79$.
	\item Similar to \citet{MP}, we set equal initial consumption, $c_0=1$ for each individual and generate each individual's initial consumption $x_0$ randomly according to a uniform distribution over $(c\underline{x},c\bar{x})$.
	\item Generate a $2\times12\times79$ random vector $\omega$ that follows a standard normal distribution. Using this vector $\omega$, we generate the process of the risky asset returns $\Delta S_t/S_t$, and the dual process $y_t^*$ in Proposition \ref{pro:consumption} for all the $N$ individuals. By Proposition \ref{pro:consumption}, we can simulate the optimal consumption processes of $N$ individuals.	
\end{itemize}
\noindent \textbf{(Step 2)} Aggregation of Consumption
\begin{itemize}
	\item Let $C^1, C^2,...,C^N$ be the simulated consumption processes for $N$ individuals in our utility cost model obtained in \textbf{(Step 1)}.
	\item (Cross sectional aggregation) The cross sectionally aggregated consumption process $CA$ of $C^1,C^2,...,C^N$ is defined as follows:
	
	\noindent For $j=1,2,...,2\times12\times79$,
	$$
	CA(t_j)=\dfrac{1}{N}\sum_{i=1}^N c^{i}(t_j).
	$$
	\item (Temporal aggregation) We temporally aggregate the cross sectionally aggregated series $CA$ to create monthly $CA^*$ (that is, $\Delta t=1/12$) as follows:
	
	\noindent For $j=1,2,...,12\times79$,
	$$
	CA^*(i)=\sum_{j=1}^{2}CA(t_{2\times(i-1)+j}).
	$$
\end{itemize}
\noindent \textbf{(Step 3)} Compute the consumption growth rate, IMRS, and theoretical EP of aggregated consumption $CA^*$
\begin{itemize}
	\item We use the simulated returns on the risky asset $r(t_j)=1,2,...,12\times79$.
	
	\item We derive the following time-series
	
	\noindent $i=1,2,...,12\times79-1$,
	\begin{eqnarray*}
		\begin{split}
			&\textrm{(Consumption growth rate)}\;\;CG(i)=\dfrac{CA^*(i+1)-CA^*(i)}{CA^*(i)}.\\
			&\textrm{(IMRS)}\;\;I(i)=e^{-\delta \Delta t}\left(\dfrac{CA^*(i+1)}{CA^*(i)}\right)^{-\gamma}.\\
			&\textrm{(EP)}\;\;EP(i)=-\dfrac{cov\left(e^{-\delta \Delta t}(CA^*(i+1)/CA^*(i))^{-\gamma}, (r(i+1)/r(i))\right)}{\mathbb{E}[e^{-\delta \Delta t}(CA^*(i+1)/CA^*(i))^{-\gamma}]}.
		\end{split}
	\end{eqnarray*}
	Using these time-series, we obtain desired statistics(the mean and the standard derivation of consumption growth, the IMRS, the theoretical EP and the auto-correlation of aggregated consumption $CA^*$).
\end{itemize}
Since each time-series depends on the random vector $\omega$, repeat \textbf{(Step 1)--(Step 3)} 1000 times.

The baseline results are summarized in Table \ref{table-ac1}. The model fits the consumption data well with reasonable values for the market and preference parameters. In particular, the auto-correlation simulated by our model is well matched with the data, which is not possible with the standard Merton case. It is surprising that a fairly small values of $(\alpha, \beta)$ can generate the autocorrelation value close to its historical data. Table \ref{table-ac2} confirms this result.

\section{Concluding Remarks \label{sec:conclusion}}

We model the partial irreversibility of consumption decision motivated by \cite{Duesenberry}. In order to do so, we introduce the adjustment cost of consumption. We show that the consumption partial irreversibility model can generate a number of novel implications. Some of our results are similar to those derived from  habit formation or consumption commitment models. However, the mechanism to generate time-varying risk aversion or the excess sensitivity and excess smoothness consumption is very different from that from these literature. In this sense, we view that our model as complement to the existing literature. In addition to these results, we find that the consumption adjustment cost can reconcile the gap between the asset pricing literature and the literature on estimating risk aversion. Also, we shed light on the debate on how the wealth change has impact on the households risky share.

We believe that the (partial) irreversibility is a very realistic and important aspect of  consumption decision. While we model it by introducing the adjustment cost, we admit that there would be other ways to model it. We hope that our model contributes to building other works toward this direction. One of important future works will be building a general equilibrium model with the partial irreversibility of consumption decision that can investigate further asset pricing implications (e.g., \cite{CJK-G}). This type of extension will help to provide the microfoundation for the {\it relative income hypothesis} by \cite{Duesenberry}.

\clearpage
\begin{center}
    {\bf {\Large Appendix}}
\end{center}

\appendix
\linespread{1.1}
\begin{footnotesize}
\section{Derivation of the dual value function $J(y,c)$}\label{sec:Appendix:A}

In this section, we will derive a solution to Problem \ref{pr:dual_problem} by solving the HJB equation \eqref{eq:HJB_dual_value}.

The following theorem guarantees that the solution to the HJB
equation \eqref{eq:HJB_dual_value} is the solution to the dual
problem.
\begin{thm}[Verification Theorem]~\label{thm:verification}\\
    \noindent 1. Suppose that the HJB equation \eqref{eq:HJB_dual_value} has a twice continuously differentiable solution $J(y,c)\;:\:\mathcal{R}\rightarrow\;\mathbb{R}$ satisfying the following conditions:
    \begin{itemize}
        \item[(1)] For any admissible consumption strategy $(c^{+},c^{-})$, the process defined by
        $$
        \int_{0}^{t}e^{-\delta s}(-\theta) y_s J_{y}(y_s,c_s) dB_s,\;\;t\ge 0,
        $$
        is a martingale.
        \item[(2)] For any admissible consumption strategy $(c^{+},c^{-})$,
        $$
        \liminf_{t\to \infty}\mathbb{E}[e^{-\delta t}J(y_t,c_t)]\ge 0.
        $$
    \end{itemize}
    Then, for initial condition $(y_0,c_0)\in\mathcal{R}$ and any admissible consumption strategy $(c^{+},c^{-})$,
    $$
    J(y_0,c_0) \ge \mathbb{E}\left[\int_{0}^{\infty}e^{-\delta t}\left(h(y_t,c_t)dt-\alpha u'(c_t)dc_t^{+}-\beta u'(c_t)dc_t^-\right)\right].
    $$
    \noindent 2. Given any initial condition $(y,c)\in\mathcal{R}$, suppose that there exist an admissible consumption strategy $(c^{*,+},c^{*,-})$ such that, if $c^*$ is the associated consumption process, then
    $$
    (y_t,c^{*}_t)\in\left\{(y,c)\in\mathcal{R}: \mathcal{L}J(y,c)+h(y,c)=0\right\},
    $$
    Lebesgue-a.e., $\mathbb{P}$-a.s.,
    \begin{eqnarray}
    \begin{split}
    &\int_{0}^{t}e^{-\delta s}\left(J_c(y_s,c_s^*)-\alpha u'(c_s^*)\right)dc_{s}^{*,+}=0,\;\;\;\mbox{for all}\;t\ge 0,\;\mathbb{P}-a.s.,\\
    &\int_{0}^{t}e^{-\delta s}\left(-J_c(y_s,c_s^*)-\beta u'(c_s^*)\right)dc_{s}^{*,-}=0,\;\;\;\mbox{for all}\;t\ge 0,\;\mathbb{P}-a.s,
    \end{split}
    \end{eqnarray}
    and
    $$
    \lim_{t\to \infty} \mathbb{E}\left[e^{-\delta t}J(y_t,c^{*}_t)\right]=0\qquad\textrm{({Transversality condition})}.
    $$
    Then, $J(y,c)$ is the dual value function for Problem \ref{pr:dual_problem} and $(c^{*,+},c^{*,-})$ is the optimal consumption strategy.
\end{thm}
\begin{proof}

\noindent \textbf{(Proof of 1.)}

    For given consumption process $\{c_t\}_{t=0}^{\infty}$, define a process
    \begin{eqnarray}
    M_t^c =\int_{0}^{t}e^{-\delta s} \left((u(c_s)-y_sc_s)ds)-\alpha u'(c_s)dc_{t}^{+}-\beta u'(c_s)dc_t^{-}\right) + e^{-\delta t} J(y_t, c_t).
    \end{eqnarray}
    By the generalized It\'{o}'s lemma (See \citet{Harrison}),
    \begin{eqnarray}
    \begin{split}
    dM_t^c=&e^{-\delta t}(u(c_t)-y_tc_t)dt - \alpha u'(c_t)dc_{t}^{+}-\beta u'(c_t)dc_t^{-}+\left(e^{-\delta t}dJ(y_t,c_t)-e^{-\delta t}\delta J(y_t,c_t)dt\right)\\
    =&e^{-\delta t}\left(\dfrac{\t^2}{2}y_t^2J_{yy}(y_t,c_t)+(\delta-r)y_tJ_y(y_t,c_t)-\delta J(y_t,c_t)+u(c_t)-y_t c_t \right)dt\\
    +&e^{-\delta t}(J_c (y_t,c_t)-\alpha u'(c_t))d(c_{t}^{+})^c +e^{-\delta t}(-J_c (y_t,c_t)-\beta u'(c_t))d(c_{t}^{-})^c\\
    +&e^{-\delta t}(J(y_t,c_t)-J(y_t,c_{t-})-\alpha u'(c_t)\Delta c_t){\bf 1}_{\{\Delta c_t>0\}}\\
    +&e^{-\delta t}(J(y_t,c_t)-J(y_t,c_{t-})+\beta u'(c_t)\Delta c_t){\bf 1}_{\{\Delta c_t<0\}} -\t e^{-\delta t}y_t J_y(y_t ,c_t) dB_t
    \end{split}
    \end{eqnarray}
    where $(c^+)^c$ and $(c^-)^c$ are the continuous parts of $c^+$ and $c^-$, respectively.

    Hence, for any fixed $T>0$,
    \begin{eqnarray}
    \begin{split}\label{eq:generalized_ito}
    M_{T}^{c}-M_t^{c}=&\underbrace{\int_{t}^{T}e^{-\delta s}\left(\dfrac{\t^2}{2}y_s^2J_{yy}(y_s,c_s)+(\delta-r)y_sJ_y(y_s,c_s)-\delta J+u(c_s)-y_sc_s \right)ds}_{(A)}\\
    +&\underbrace{\int_{t}^{T}\left(J_c(y_s,c_s)-\alpha u'(c_s)\right)d(c_{s}^{+})^c +\int_{t}^{T}\left(-J_c(y_s,c_s)-\beta u'(c_s)\right)d(c_{s}^{-})^c}_{(B)}\\
    +&\underbrace{\sum_{t\le s \le T}e^{-\delta s}\left(J(y_s,c_s)-J(y_s,c_{s-})-\alpha u'(c_s)\Delta c_s\right){\bf 1}_{\{\Delta c_s>0\}}}_{(C)}\\
    +&\underbrace{\sum_{t\le s \le T}e^{-\delta s}\left(J(y_s,c_s)-J(y_s,c_{s-})+\beta u'(c_s)\Delta c_s\right){\bf 1}_{\{\Delta c_s < 0\}}}_{(D)}\\
    +&\underbrace{\int_{t}^{T}(-\t)e^{-\delta s}y_s J_y(y_s,c_s) dB_s}_{(E)}.
    \end{split}
    \end{eqnarray}
    Since
    $$
    \max\{\mathcal{L}J+u(c)-yc, J_c(y,c)-\alpha u'(c), -J_c(y,c)-\beta u'(c)\}=0,
    $$
    we deduce that
    $$
    (A)\le 0 \;\;\mbox{and}\;\;(B)\le 0.
    $$
    Moreover,
    \begin{eqnarray}
    \begin{split}
    (C)=&\sum_{t\le s \le T} e^{-\delta s}\int_{c_s-\Delta c_s}^{c_s}\left(J_c(y_s,c_s)-\alpha u'(c)\right)dc \cdot {\bf 1}_{\{\Delta c_s >0\}} \le 0,\\
    (D)=&\sum_{t\le s \le T} e^{-\delta s}\int_{0}^{|\Delta c_s|}\left(-J_c(y_s,c_s-|\Delta c_s|+c)-\beta u'(c)\right)dc \cdot {\bf 1}_{\{\Delta c_s <0\}} \le 0,
    \end{split}
    \end{eqnarray}
    and by assumption, $\mathbb{E}[(E)]=0$.

    Thus, we can conclude that
    $$
    \mathbb{E}_t[M_T^c - M_t^c]\le 0,
    $$
    and $\{M_t^{c}\}_{t\ge 0}$ is a super-martingale.

    This implies that $\mathbb{E}[M_T^c]\le J(y_0,c_0)$ and
    \begin{eqnarray}
    \begin{split}
    J(y_0, c_0) &\ge \mathbb{E}\left[\int_{0}^{T}e^{-\delta s} \left((u(c_s)-y_sc_s)ds)-\alpha u'(c_s)dc_{t}^{+}-\beta u'(c_s)dc_t^{-}\right)\right] + e^{-\delta T} J(y_T, c_T).
    \end{split}
    \end{eqnarray}
    By assumption
    $$
    \liminf_{T \to \infty}\mathbb{E}[e^{-\delta T}J(y_T,c_T)]\ge 0.
    $$
    and Fatou's lemma, we deduce that
    \begin{eqnarray}\label{eq:super_optimal}
    J(y_0,c_0)\ge \mathbb{E}\left[\int_{0}^{\infty}e^{-\delta s} \left((u(c_s)-y_sc_s)ds)-\alpha u'(c_s)dc_{t}^{+}-\beta u'(c_s)dc_t^{-}\right)\right].
    \end{eqnarray}
    The relation \eqref{eq:super_optimal} holds for any admissible consumption strategy $(c^{+},c^{-})$, we obtain
    $$
    J(y_0,c_0) \ge \sup_{(c^+,c^-)\in \Pi(c)}\mathbb{E}\left[\int_{0}^{\infty}e^{-\delta t}\left(h(y_t,c_t)dt-\alpha u'(c_t)dc_t^{+}-\beta u'(c_t)dc_t^-\right)\right].
    $$

\noindent \textbf{(Proof of 2.)}

    By assumption, we can show that in \eqref{eq:generalized_ito}
    $$
    \mathbb{E}[(A)]=\mathbb{E}[(B)]=\mathbb{E}[(C)]=\mathbb{E}[(D)]=\mathbb{E}[(E)]=0 \;\;\;\mbox{for the process}\;M_t^{c^*}.
    $$
    This implies that $\{M_t^{c^*}\}_{t\ge 0}$ is a martingale and
    \begin{eqnarray}
    \begin{split}
    J(y_0, c_0) &= \mathbb{E}\left[\int_{0}^{T}e^{-\delta s} \left((u(c^*)-y_s c^*)ds-\alpha u'(c^*)dc_{t}^{*,+}-\beta u'(c^*)dc_{t}^{*,-}\right)\right] + e^{-\delta T} J(y_T, c_T).
    \end{split}
    \end{eqnarray}
    The transversality condition leads to
    \begin{eqnarray}
    \begin{split}
    J(y_0, c_0) &= \mathbb{E}\left[\int_{0}^{\infty}e^{-\delta s} \left((u(c^*)-y_s c^*)ds-\alpha u'(c^*)dc_{t}^{*,+}-\beta u'(c^*)dc_{t}^{*,-}\right)\right].
    \end{split}
    \end{eqnarray}
    Thus,
    $$
    J(y_0, c_0) = \sup_{(c^+,c^-)\in \Pi(c)}\mathbb{E}\left[\int_{0}^{\infty}e^{-\delta s} \left((u(c)-y_s c)ds-\alpha u'(c)dc_{t}^{+}-\beta u'(c)dc_t^{-}\right)\right]
    $$
    and the consumption strategy $(c^{*,+},c^{*,-})$ attains the maximum. Hence $(c^{*,+},c^{*,-})$ is the optimal.

\end{proof}
Now, we will obtain the analytic characterization of the dual value function by using
the the variational inequality \eqref{eq:HJB_dual_value}.

As \citet{DY},  we consider the double obstacle problem arising from variational inequality \eqref{eq:HJB_dual_value} as follows:
\begin{eqnarray}\label{eq:double_obstacle1}
\begin{split}
\begin{cases}
\mathcal{L}w(y,c) + u'(c)-y \ge 0, \qquad &\mbox{for}\;\;w(y,c)=\alpha u'(c),\\
\mathcal{L}w(y,c) + u'(c)-y \le 0, \qquad &\mbox{for}\;\;w(y,c)=-\beta u'(c),\\
\mathcal{L}w(y,c) + u'(c)-y = 0, \qquad &\mbox{for}\;\;-\beta u'(c)< w(y,c)<\alpha u'(c),\\
\end{cases}
\end{split}
\end{eqnarray}

Consider the following substitution:
$$
w(y,c)=u'(c) H(z)\;\;\;\mbox{and}\;\;z=\dfrac{y}{u'(c)}.
$$
Then, the double obstacle problem \eqref{eq:double_obstacle1} can be changed by
\begin{eqnarray}\label{eq:double_obstacle2}
\begin{split}
\begin{cases}
\mathcal{L}H(z) + 1-z \ge 0, \qquad &\mbox{for}\;\;H(z)=\alpha,\\
\mathcal{L}H(z) + 1-z \le 0, \qquad &\mbox{for}\;\;H(z)=-\beta,\\
\mathcal{L}H(z) + 1-z = 0, \qquad &\mbox{for}\;\;-\beta <H(z)<\alpha,\\
\end{cases}
\end{split}
\end{eqnarray}
The following proposition provides the exact solution of the double obstacle problem \eqref{eq:double_obstacle2}.
\begin{pro}\label{pro:solution_VI_H}
    The variational inequality \eqref{eq:double_obstacle2} has a unique $\mathcal{C}^{1}$-solution, which is
    \begin{eqnarray}
    \begin{split}\label{eq:explicit_H}
    H(z)=
    \begin{cases}
    \;\;\alpha, \qquad &\mbox{for}\;\;z\le b_{\alpha},\\
    \;D_1 \left(\dfrac{z}{b_{\alpha}}\right)^{m_1}  + D_2 \left(\dfrac{z}{b_{\alpha}}\right)^{m_2} + \dfrac{1}{\delta}-\dfrac{z}{r}, \qquad &\mbox{for}\;\;b_{\alpha}< z < b_{\beta},\\
    -\beta,  \qquad &\mbox{for}\;\;z\ge b_{\beta},\\
    \end{cases}
    \end{split}
    \end{eqnarray}
    where
    $$
    D_1 = \dfrac{(\alpha-\frac{1}{\delta})m_2 + (m_2-1)\frac{b_{\alpha}}{r}}{m_2 - m_1},\;\;D_2 = \dfrac{(\alpha-\frac{1}{\delta})m_1 + (m_1-1)\frac{b_{\alpha}}{r}}{m_1 - m_2},
    $$
    and $m_1$,$m_2$ are positive and negative roots of following quadratic equation:
    $$
    \dfrac{\theta^2}{2}m^2 + (\delta-r-\dfrac{\theta^2}{2})m- \delta =0.
    $$
    and $b_{\alpha}$, $b_{\beta}$ are defined as
    $$
    b_{\alpha}=(1-\delta \alpha)\dfrac{m_1-1}{m_1}\dfrac{\frac{1}{\kappa}w^{m_1}-1}{w^{m_1-1}-1},\qquad b_{\beta}=(1+\delta \beta)\dfrac{m_1-1}{m_1}\dfrac{w^{m_1}-\kappa}{w^{m_1}-w}.
    $$
    with $\kappa=\dfrac{1-\delta\alpha}{1+\delta\beta}$.

    Here, $w$ is the unique root to the equation $f(w) =0$ in $(0,1)$, where
    \begin{equation}\label{eq:ff}
    f(w)=(m_1-1)m_2(1-w^{1-m_2})(w^{m_1}-\kappa) - m_1 (m_2-1)(w^{m_1} -w)(1-\kappa w^{-m_2}).
    \end{equation}

    Also,
    $$
    H'(b_\alpha)=H'(b_\beta)=0,\;\;H'(z)<0\;\;\mbox{for}\;\;z\in(b_\alpha,b_\beta)
    $$
    and  $H'(z)$ attains minimum at $b_m\in(b_\alpha,b_\beta)$ defined by
    $$
    b_m = b_{\alpha} \cdot \left(\dfrac{-D_2m_2(m_2-1)}{D_1m_1(m_1-1)}\right)^{\frac{1}{m_1-m_2}}.
    $$
\end{pro}
\begin{proof}

The uniqueness of the solution is guaranteed due to the maximum principle of the partial differential equation theory(see \citet{Lieberman}).

Now, we will prove the remain part of proposition in the following steps.

\noindent {\bf (Step 1)} We first consider the following free boundary problem:
\begin{eqnarray}\label{eq:free_boundary}
\begin{split}
\begin{cases}
\mathcal{L}H + 1 - z =0,\;\;\;&\;\;b_{\alpha}<z<b_{\beta},\\
H(b_{\alpha})=\alpha,\;\;\;&H'(b_{\alpha})=0,\\
H(b_{\beta})=-\beta,\;\;\;&H'(b_{\beta})=0.
\end{cases}
\end{split}
\end{eqnarray}
Then we can extend the solution $H$ onto $\mathbb{R}_{+}$ by
\begin{eqnarray}
H(z)=\alpha\;\;\mbox{if}\;\;z\in(0,b_{\alpha})\;\;\;\mbox{and}\;\;\;H(z)=-\beta\;\;\mbox{if}\;\;z\in(b_{\beta},\infty).
\end{eqnarray}
Next, we show that $H(z)$ is the solution to variational inequality \eqref{eq:double_obstacle2}. We can let the general solution for \eqref{eq:free_boundary} in the form of
$$
H(z)=D_1 \left(\dfrac{z}{b_{\alpha}}\right)^{m_1} + D_2 \left(\dfrac{z}{b_{\beta}}\right)^{m_2}+\dfrac{1}{\delta}-\dfrac{z}{r}.
$$
From the smooth-pasting condition $H(b_{\alpha})=\alpha$ and $H'(b_{\alpha})=0$,
\begin{eqnarray}\begin{split}
&H(b_{\alpha})=D_1 + D_2 +\dfrac{1}{\delta}-\dfrac{b_{\alpha}}{r}=\alpha,\\
&H'(b_{\alpha})=m_1 D_1 + m_2 D_2 -\dfrac{1}{r}=0.
\end{split}
\end{eqnarray}
Therefore, $D_1$ and $D_2$ are given by
\begin{eqnarray}\label{eq:D1D2_1}
D_1=\dfrac{(\alpha-\frac{1}{\delta})m_2 + (m_2-1)\frac{b_{\alpha}}{r}}{m_2-m_1},\;\;D_2=\dfrac{(\alpha-\frac{1}{\delta})m_1 + (m_1-1)\frac{b_{\alpha}}{r}}{m_1-m_2}.
\end{eqnarray}
Similarly,
\begin{eqnarray}\begin{split}
&H(b_{\beta})=D_1\left(\dfrac{b_{\beta}}{b_{\alpha}}\right)^{m_1} + D_2 \left(\dfrac{b_{\beta}}{b_{\alpha}}\right)^{m_2}+\dfrac{1}{\delta}-\dfrac{b_{\beta}}{r}=-\beta,\\
&H'(b_{\beta})=\dfrac{m_1 D_1}{b_{\alpha}}\left(\dfrac{b_{\beta}}{b_{\alpha}}\right)^{m_1-1} + \dfrac{m_2 D_2}{b_{\alpha}}\left(\dfrac{b_{\beta}}{b_{\alpha}}\right)^{m_2-1} -\dfrac{1}{r}=0,
\end{split}
\end{eqnarray}
and
\begin{eqnarray}\label{eq:D1D2_2}
D_1=\dfrac{-(\beta+\frac{1}{\delta})m_2 + (m_2-1)\frac{b_{\beta}}{r}}{m_2 - m_1}\left(\dfrac{b_{\alpha}}{b_{\beta}}\right)^{m_1},\;\;D_2=\dfrac{-(\beta+\frac{1}{\delta})m_1 + (m_1-1)\frac{b_{\beta}}{r}}{m_1 - m_2}\left(\dfrac{b_{\alpha}}{b_{\beta}}\right)^{m_2}.
\end{eqnarray}
From \eqref{eq:D1D2_1} and \eqref{eq:D1D2_2},
\begin{eqnarray}
\begin{split}\label{eq:D1D2_3}
(\alpha-\dfrac{1}{\delta})m_2 + (m_2-1)\dfrac{b_{\alpha}}{r}&=\left(-(\beta+\frac{1}{\delta})m_2+(m_2-1)\dfrac{b_{\beta}}{r}\right)\cdot\left(\dfrac{b_{\alpha}}{b_{\beta}}\right)^{m_1},\\
(\alpha-\dfrac{1}{\delta})m_1 + (m_1-1)\dfrac{b_{\alpha}}{r}&=\left(-(\beta+\frac{1}{\delta})m_1+(m_1-1)\dfrac{b_{\beta}}{r}\right)\cdot\left(\dfrac{b_{\alpha}}{b_{\beta}}\right)^{m_2}.\\
\end{split}
\end{eqnarray}
Let us define
$$
w=\dfrac{b_{\alpha}}{b_{\beta}}.
$$
From \eqref{eq:D1D2_3},
\begin{eqnarray}\label{eq:D1D2_4}
\dfrac{m_2}{m_1}\dfrac{(\alpha-\frac{1}{\delta})+(\beta+\frac{1}{\delta})w^{m_1}}{(\alpha-\frac{1}{\delta})+(\beta+\frac{1}{\delta})w^{m_2}}=\dfrac{m_2-1}{m_1-1}\dfrac{w^{m_1}-w}{w^{m_2}-w}.
\end{eqnarray}
From \eqref{eq:D1D2_4}, we define $f(w)$ as \eqref{eq:f}.\\

\noindent{\bf (Step 2)} $f(w)$ has a unique solution $w\in(0,1)$ and $w\in(0,\kappa)$.\\

For $w\in(\kappa,1)$,
\begin{eqnarray}
\begin{split}
f(w)=&w^{-m_2}\left((m_1-1)m_2(w^{m_2}-w)(w^{m_1}-\kappa)-m_1(m_2-1)(w^{m_1}-w)(w^{m_2}-\kappa)\right)\\
<&w^{-m_2}\left((m_1-1)m_2(w^{m_2}-w)(w^{m_1}-\kappa)-m_1(m_2-1)(w^{m_1}-\kappa)(w^{m_2}-\kappa)\right)\\
=&w^{-m_2}(w^{m_1}-\kappa)(w^{m_2}-\kappa)(m_1-m_2)\\
<&0.
\end{split}
\end{eqnarray}
On the other hand,
\begin{eqnarray}
\begin{split}
f(0)=&(m_1-1)m_2(-\kappa)>0,\\
f(\kappa)=&(m_1-1)m_2(1-\kappa^{(1-m_2)})(\kappa^{m_1}-\kappa)-m_1(m_2-1)(\kappa^{m_1}-\kappa)(1-\kappa^{1-m_2})\\
=&(m_1-m_2)(1-\kappa^{(1-m_2)})(\kappa^{m_1}-\kappa)<0.
\end{split}
\end{eqnarray}
By the Intermediate Value Theorem, there exists $w\in(0,\kappa)$ such that $f(w)=0$. To show ``uniqueness", it is suffice to show that
$$
f''(w) > 0,\;\;\;w\in(0,\kappa).
$$
Then,
\begin{eqnarray}
\begin{split}
f(w)=&w^{-m_2}\left((m_1-1)m_2(w^{m_2}-w)(w^{m_1}-\kappa)-m_1(m_2-1)(w^{m_1}-w)(w^{m_2}-\kappa)\right)\\=&(m_1-m_2)w^{m_1}+(m_1-m_2)\kappa w^{1-m_2}-(m_1-1)m_2 w^{m_1-m_2+1}-\kappa (m_1-1)m_2\\
+&m_1(m_2-1)w+ \kappa m_1(m_2-1)w^{m_1-m_2},
\end{split}
\end{eqnarray}
and
\begin{eqnarray}\label{eq:f''1}
\begin{split}
f''(w)=&(m_1-m_2)m_1(m_1-1)w^{m_1-2}+\kappa(m_1-m_2)(1-m_2)(-m_2)w^{-m_2-1}
\\-&(m_1-1)m_2(m_1-m_2+1)(m_1-m_2)w^{m_1-m_2-1}\\+&\kappa m_1(m_2-1)(m_1-m_2)(m_1-m_2-1)w^{m_1-m_2-2}\\
=&(m_1-m_2)\left(m_1(m_1-1)w^{m_1-2}+\kappa m_2(m_2-1)w^{-m_2-1}-(m_1-1)m_2(m_1-m_2+1)w^{m_1-m_2-1}\right.\\
+&\left. \kappa m_1(m_2-1)(m_1-m_2-1)w^{m_1-m_2-2}\right)\\
=&(m_1-m_2)\left(\underbrace{m_1(m_1-1)w^{m_1-2}-(m_1-1)m_2m_1w^{m_1-m_2-1}+\kappa m_1(m_2-1)(m_1-1)w^{m_1-m_2-2}}_{{\bf (A)}}\right.\\
+&\left. \underbrace{\kappa m_2(m_2-1)w^{-m_2-1}+(m_1-1)m_2(m_2-1)w^{m_1-m_2-1}-\kappa m_1(m_2-1)m_2w^{m_1-m_2-2}}_{{\bf (B)}}\right).
\end{split}
\end{eqnarray}
Let us temporarily denote
$$
f_1(w)=(1-m_2 w^{1-m_2}+(m_2-1)\kappa w^{-m_2}).
$$
Since
$$
f_1'(w)=-m_2(1-m_2)w^{-m_2}-(m_2-1)m_2\kappa w^{-m_2-1}=(-m_2)(1-m_2)w^{-m_2-1}(w-\kappa)<0,
$$
\begin{eqnarray}\label{eq:f''2}
\begin{split}
{\bf (A)}=&m_1(m_1-1)(w^{m_1-2}-m_2w^{m_1-m_2-1}+\kappa(m_2-1)w^{m_1-m_2-2})
\\=&m_1(m_1-1)w^{m_1-2}\underbrace{(1-m_2 w^{1-m_2}+(m_2-1)\kappa w^{-m_2})}_{=f_1(w)}\\
>&m_1(m_1-1)w^{m_2-2}g_1(\kappa)=m_1(m_1-1)w^{m_2-2}f_1(\kappa)=m_1(m_1-1)w^{m_2-2}(1-\kappa^{1-m_2})\\
>&0.
\end{split}
\end{eqnarray}
Let us temporarily denote
$$
f_2(w)=(\kappa + (m_1-1)w^{m_1}-\kappa m_1 w^{m_1-1}).
$$
Since
$$
f_2'(w)=m_1(m_1-1)w^{m_1-1} - \kappa m_1(m_1-1)w^{m_1-2}=m_1(m_1-1)w^{m_1-2}(w-\kappa)<0,
$$
\begin{eqnarray}\label{eq:f''3}
\begin{split}
{\bf (B)}=&m_2(m_2-1)(\kappa w^{-m_2-1}+(m_1-1)w^{m_1-m_2-1}-\kappa m_1 w^{m_1-m_2-2})\\
=&m_2(m_2-1)w^{-m_2-1}\underbrace{(\kappa +(m_1-1)w^{m_1}-\kappa m_1 w^{m_1-1})}_{=f_2(w)}\\
>&m_2(m_2-1)w^{-m_2-1}f_2(\kappa)=m_2(m_2-1)w^{-m_2-1}\kappa(1-\kappa^{m_1-1})>0.
\end{split}
\end{eqnarray}
By \eqref{eq:f''1}, \eqref{eq:f''2} and \eqref{eq:f''3}, we can conclude that
$$
f''(w)>0\;\;\;\mbox{on}\;\;(0,\kappa)
$$
and $f(w)$ has a unique solution $w$ on $(0,1)$ with $w\in(0,\kappa)$.\\

\noindent {\bf (Step 3)} The two free boundaries $b_{\alpha}$ and $b_{\beta}$ are uniquely determined. Moreover, $0<b_{\alpha}<(1-\delta \alpha)$ and $(1+\delta \beta)<b_{\beta}<\infty$.\\

From \eqref{eq:D1D2_3},
$$
b_{\alpha}=(1-\delta \alpha)\dfrac{m_1-1}{m_1}\dfrac{\frac{1}{\kappa}w^{m_1}-1}{w^{m_1-1}-1}\;\;\;\mbox{and}\;\;\;b_{\beta}=(1+\delta \beta)\dfrac{m_1-1}{m_1}\dfrac{w^m_1 -\kappa}{w^{m_1}-w}.
$$
Thus, $b_{\alpha}$ and $b_{\beta}$ are uniquely determined.

Let us temporarily denote
$$
f_3(w)=w^{m_1}-m_1 w +\kappa m_1 - \kappa.
$$
Since
$$
f_3'(w)=m_1(w^{m_1-1}-1)>0\;\;\;\mbox{and}\;\;\;f_3(0)>0,
$$
we deduce that
$$
f_3(w)>0\;\;\;\mbox{on}\;(0,\kappa).
$$
This leads to
$$
\dfrac{m_1-1}{m_1}\dfrac{w^m_1 -\kappa}{w^{m_1}-w}>1
$$
and $b_{\beta}>(1+\delta \beta)$.

Let us also temporarily denote
$$
f_4(w)=(m_1-1)w^{m_1}-\kappa m_1 w^{m_1-1} +\kappa.
$$
Since
$$
f_4'(w)=(m_1-1)m_1 w^{m_1-2}(w-\kappa)<0\;\;\;\mbox{and}\;\;f_4(\kappa)=\kappa-\kappa^{m_1}>0,
$$
we deduce that
$$
f_4(w)>0\;\;\;\mbox{on}\;(0,\kappa).
$$
This leads to
$$
\dfrac{m_1-1}{m_1}\dfrac{\frac{1}{\kappa}w^{m_1}-1}{w^{m_1-1}-1}<1
$$
and $b_{\alpha}<(1-\delta \alpha)$.\\

\noindent {\bf (Step 4)} In $z\in(b_{\alpha}, b_{\beta})$, $H'(z) < 0$  and attains minimum at $b_m$.\\

First, we will show that
$$
D_1 > 0 \;\;\;\mbox{and}\;\;D_2 <0.
$$
Since
$$
D_1 = \dfrac{(\alpha-\frac{1}{\delta})m_2 + (m_2-1)\frac{b_{\alpha}}{r}}{m_2 - m_1}, \;\;D_2 = \dfrac{(\alpha-\frac{1}{\delta})m_1 + (m_1-1)\frac{b_{\alpha}}{r}}{m_1 - m_2},
$$
\begin{eqnarray}
\begin{split}
D_1 > 0 \;\;&\Longleftrightarrow\;\;(\alpha-\dfrac{1}{\delta})m_2 + (m_2-1)\dfrac{b_{\alpha}}{r}<0.\\
&\Longleftrightarrow\;\;b_{\alpha} > \dfrac{-m_2}{m_2-1}\dfrac{r}{\delta}(1-\delta \alpha).\\
&\Longleftrightarrow\;\;(1-\delta \alpha)\dfrac{m_1-1}{m_1}\dfrac{\frac{1}{\kappa}w^{m_1}-1}{w^{m_1-1}-1} > \dfrac{-m_2}{m_2-1}\dfrac{r}{\delta}(1-\delta \alpha).\\
&\Longleftrightarrow\;\;\dfrac{\frac{1}{\kappa}w^{m_1}-1}{w^{m_1-1}-1} > 1.\\
&\Longleftrightarrow\;\;\kappa w^{m_1-1} > w^{m_1}.
\end{split}
\end{eqnarray}
Similarly, we can deduce that $D_2<0$.

We know that $H'(b_{\alpha})=H'(b_{\beta})=0$ and
$$
H''(z)=\dfrac{D_1 m_1(m_1-1)}{b_{\alpha}^2}\left(\dfrac{z}{b_{\alpha}}\right)^{m_1-2}+\dfrac{D_2 m_2}{b_{\alpha}^2}\left(\dfrac{z}{b_{\alpha}}\right)^{m_2-2}.
$$
Since $H''(b_m)=0$, it is enough to show that
$$
b_\alpha < b_m < b_\beta.
$$
By the definition of $b_m$,
$$
b_\alpha < b_m < b_\beta\;\;\;\Longleftrightarrow\;\;\;1<\left(\dfrac{-D_2 m_2(m_2-1)}{D_1m_1(m_1-1)}\right)^\frac{1}{m_1-m_2}<\dfrac{1}{x}.
$$
Since
$$
D_1 = \dfrac{(\alpha-\frac{1}{\delta})m_2 + (m_2-1)\frac{b_{\alpha}}{r}}{m_2 - m_1}, \;\;D_2 = \dfrac{(\alpha-\frac{1}{\delta})m_1 + (m_1-1)\frac{b_{\alpha}}{r}}{m_1 - m_2},
$$
we can easily check that
$$
1<\left(\dfrac{-D_2 m_2(m_2-1)}{D_1m_1(m_1-1)}\right)^\frac{1}{m_1-m_2}\;\;\;\Longleftrightarrow\;\;\;b_\alpha <(1-\delta\alpha).
$$
Also, we know that
\begin{eqnarray*}
    D_1=\dfrac{-(\beta+\frac{1}{\delta})m_2 + (m_2-1)\frac{b_{\beta}}{r}}{m_2 - m_1}\left(\dfrac{b_{\alpha}}{b_{\beta}}\right)^{m_1},\;\;D_2=\dfrac{-(\beta+\frac{1}{\delta})m_1 + (m_1-1)\frac{b_{\beta}}{r}}{m_1 - m_2}\left(\dfrac{b_{\alpha}}{b_{\beta}}\right)^{m_2}.
\end{eqnarray*}
This implies that
$$
\left(\dfrac{-D_2 m_2(m_2-1)}{D_1m_1(m_1-1)}\right)^\frac{1}{m_1-m_2}<\dfrac{1}{x}\;\;\;\Longleftrightarrow\;\;\;b_\beta>(1+\delta\beta).
$$
Thus, we deduce that
$$
b_\alpha < b_m < b_\beta.
$$
and
$$
H''(z)<0,\;\;\mbox{on}\;\;(b_{\alpha},b_m)\;\;\;\mbox{and}\;\;\;H''(z)>0,\;\;\mbox{on}\;\;(b_m,b_{\beta}).
$$
Hence, $H'(z)$ attains minimum at $z=b_m$ and $H'(z)<0$ on $(b_{\alpha},b_{\beta})$.\\

\noindent{\bf (Step 4)} $H(z)$ satisfies the variational inequality \eqref{eq:double_obstacle2}.\\

\begin{itemize}
    \item For $z\in(b_{\alpha},b_{\beta})$.
    it is clear that
    \begin{eqnarray}
    \mathcal{L}H + 1 -z = 0.
    \end{eqnarray}
    Since $H(b_{\alpha})=\alpha,\;H(b_{\beta})=-\beta$ and $H'(z)$ is strictly decreasing function on $(b_{\alpha},b_{\beta})$,
    $$
    -\beta < H(z) < \alpha \;\;\;\mbox{on}\;\;(b_{\alpha},b_{\beta}).
    $$
    \item  For $z\le b_{\alpha}$, $H(z)=\alpha$ and
    \begin{eqnarray*}
        \begin{split}
            \mathcal{L}H+1-z&=1-\delta \alpha -z \ge 0.
        \end{split}
    \end{eqnarray*}
    \item  For $z\ge b_{\beta}$, $H(z)=-\beta$ and
    \begin{eqnarray*}
        \begin{split}
            \mathcal{L}H+1-z&=1+ \delta \beta -z \le 0.
        \end{split}
    \end{eqnarray*}
\end{itemize}

From {\bf (Step 1)} $\sim$ {\bf (Step 4)}, we have proved the desired result.

\end{proof}
By Proposition \ref{pro:solution_VI_H}, $w(y,c)$ given by
\begin{eqnarray}\label{eq:explicit_w}
\begin{split}
w(y,c)=
\begin{cases}
\;\;\alpha u'(c), \; &\mbox{for}\;\;\dfrac{y}{u'(c)}\le b_{\alpha},\\
\;D_1 u'(c)\left(\dfrac{y}{u'(c)b_{\alpha}}\right)^{m_1}  + D_2 u'(c)\left(\dfrac{y}{u'(c)b_{\alpha}}\right)^{m_2} + \dfrac{u'(c)}{\delta}-\dfrac{y}{r}, \; &\mbox{for}\;\;b_{\alpha}< \dfrac{y}{u'(c)} < b_{\beta},\\
-\beta u'(c),  \; &\mbox{for}\;\;\dfrac{y}{u'(c)}\ge b_{\beta},\\
\end{cases}
\end{split}
\end{eqnarray}
is a solution of the double obstacle problem \eqref{eq:double_obstacle1}.

Using the $w(y,c)$ in the equation \eqref{eq:explicit_w}, we construct the dual value function $J(y,c)$ as follows:

\noindent(i) For $b_{\alpha} u'(c) < y < b_{\beta} u'(c)$,
\begin{eqnarray}
\begin{split}\label{eq:J_define_1}
J(y,c)=&\int_{0}^{c}u'(x)D_{1}\left(\dfrac{y}{u'(x)b_{\alpha}}\right)^{m_1}dx-\int_{c}^{\infty}u'(x)D_{2}\left(\dfrac{y}{u'(x)b_{\alpha}}\right)^{m_2}dx +\dfrac{u(c)}{\delta}-\dfrac{yc}{r}
\end{split}
\end{eqnarray}
\noindent(ii) For $b_{\alpha} u'(c) \ge y$,
\begin{eqnarray}\label{eq:J_define_2}
\begin{split}
J(y,c)=J\left(y,I(\dfrac{y}{b_{\alpha}})\right)+\alpha\left(u(c)-u(I(\dfrac{y}{b_{\alpha}}))\right).
\end{split}
\end{eqnarray}
\noindent(iii) For $b_{\beta} u'(c) \le y$,
\begin{eqnarray}\label{eq:J_define_3}
\begin{split}
J(y,c)=J\left(y,I(\dfrac{y}{b_{\beta}})\right)-\beta\left(u(c)-u(I(\dfrac{y}{b_{\beta}}))\right).
\end{split}
\end{eqnarray}
where the function $I(\cdot) : \mathbb{R}_{+} \to \mathbb{R}_{+}$ is defined by
$$
I(y)\equiv (u')^{-1}(y)=y^{-\frac{1}{\gamma}}.
$$
\begin{rem}
    It is easy to check that
    $$
    m_1 >1 \;\;\mbox{and}\;\;m_2<0.
    $$
    and for $u(c)=\frac{c^{1-\gamma}}{1-\gamma}$ with $\gamma(\neq 1)>0$,
    $$
    m_2 < -\dfrac{1-\gamma}{\gamma} < m_1.
    $$
    Thus, the two integrals in \eqref{eq:J_define_1} are well-defined and $J(y,c)$ is
    $$
    J(y,c)=D_1 \dfrac{yc}{(1-\gamma+\gamma m_1)b_{\alpha}}\left(\dfrac{y}{c^{-\gamma}b_{\alpha}}\right)^{m_1-1} +D_2 \dfrac{yc}{(1-\gamma+\gamma m_2)b_{\alpha}}\left(\dfrac{y}{c^{-\gamma}b_{\alpha}}\right)^{m_2-1} +\dfrac{1}{\delta}\dfrac{c^{1-\gamma}}{1-\gamma}-\dfrac{yc}{r}.
    $$
\end{rem}
\begin{rem} For $J(y,c)$ defined in \eqref{eq:J_define_1}, \eqref{eq:J_define_2} and \eqref{eq:J_define_3}, we can easily confirm that
    $$
    J_c(y,c)= w(y,c).
    $$
\end{rem}
\begin{pro}\label{pro:J_HJB}
    For the function $J(y,c)$ defined in \eqref{eq:J_define_1}, \eqref{eq:J_define_2} and \eqref{eq:J_define_3}, the following statements are true:
    \begin{itemize}
        \item[1.] $J(y,c)$ is a twice continuously differentiable and satisfies the HJB equaton \eqref{eq:HJB_dual_value}. Moreover, the regions ${\bf IR}$, ${\bf NR}$ and ${\bf DR}$ are represented by
        \begin{eqnarray*}
            \begin{split}
                {\bf IR}&=\{(y,c)\in\mathcal{R} \mid y \le u'(c)b_{\alpha}\},\\
                {\bf NR}&=\{(y,c)\in\mathcal{R}\mid u'(c)b_{\alpha} < y < u'(c)b_{\beta} \},\\
                {\bf DR}&=\{(y,c)\in\mathcal{R} \mid  u'(c)b_{\beta}\le y\},
            \end{split}
        \end{eqnarray*}
        respectively.
        \item[2.] For any admissible consumption strategy $(c^{+},c^{-})$,
        $$
        \int_{0}^{t}(-\theta)y_s J_y(y_s,c_s)ds,\;\;\;\forall t\ge 0
        $$
        is a martingale.
        \item[3.] For any admissible consumption strategy $(c^{+},c^{-})$,
        $$
        \lim_{t \to \infty} e^{-\delta t} \mathbb{E}\left[J(y_t,c_t)\right]= 0.
        $$

    \end{itemize}
\end{pro}
\begin{proof}~

\noindent \textbf{(Proof of 1.)}

    First, with reference to the construction of the dual value function $J(y,c)$, we will show that $J$ is a continuously differentiable if we prove that $J_y,  J_{yy}$, and $J_{cc}$ are continuous along the free boundaries $c=I(\frac{y}{b_{\alpha}})$ and $c=I(\frac{y}{b_{\beta}})$.

    Then, we can compute
    \begin{eqnarray}
    \begin{split}
    J_y(y,c)=&J_y(y,I(\dfrac{y}{b_{\alpha}}))+\left(J_c(y,I(\dfrac{y}{b_{\alpha}}))-\alpha u'(I(\dfrac{y}{b_{\alpha}}))\right)\dfrac{d}{d y}\left(I(\dfrac{y}{b_{\alpha}})\right)\;\;\;\mbox{for}\;\;y\le u'(c)b_{\alpha}\\
    =&J_y(y,I(\dfrac{y}{b_{\alpha}})),
    \end{split}
    \end{eqnarray}
    and
    \begin{eqnarray}
    \begin{split}
    J_y(y,c)=&J_y(y,I(\dfrac{y}{b_{\beta}}))-\left(-J_c(y,I(\dfrac{y}{b_{\beta}})-\beta u'(I(\dfrac{y}{b_{\alpha}}))\right)\dfrac{d}{d y}\left(I(\dfrac{y}{b_{\beta}})\right)\;\;\;\mbox{for}\;\;y\ge u'(c)b_{\beta}\\
    =&J_y(y,I(\dfrac{y}{b_{\beta}})).
    \end{split}
    \end{eqnarray}
    Thus, $J_y(y,c)$ is continuous along the free boundaries.

    Similarly, we can obtain
    \begin{eqnarray}
    \begin{split}
    J_{yy}(y,c)=&J_{yy}(y,I(\dfrac{y}{b_{\alpha}})),\;\;\;\mbox{for}\;\;y\le u'(c)b_{\alpha},\\
    J_{yy}(y,c)=&J_{yy}(y,I(\dfrac{y}{b_{\beta}})),\;\;\;\mbox{for}\;\;y\le u'(c)b_{\beta},\\
    \end{split}
    \end{eqnarray}
    and hence $J_{yy}$ is continuous along the free boundaries.

    We know that $J_c(y,c)=u'(c)H(y/u'(c))$ and $H(z)$ is $\mathcal{C}^{1}$-function. Thus, it is clear that $J_{cc}(y,c)$ is continuous function and we conclude that $J(y,c)$ is $\mathcal{C}^{2}$-function.

    Next, we will show that $J(y,c)$ satisfies the HJB-equation \eqref{eq:HJB_dual_value}.
    \begin{itemize}
        \item The region {\bf NR}:\\
        Since $J_c(y,c)=u'(c)H(y/u'(c))$,
        \begin{eqnarray*}
            \begin{split}
                {\bf NR}=&\{(y,c)\in \mathcal{R}\mid -\beta u'(c)<J_c(y,c)<\alpha u'(c)\}\\
                =&\{(y,c)\in \mathcal{R} \mid -\beta <H(\frac{y}{u'(c)})<\alpha \}\\
                =&\{(y,c)\in \mathcal{R} \mid b_{\alpha} <\frac{y}{u'(c)}<b_{\beta} \}.\\
            \end{split}
        \end{eqnarray*}
        Also, we can easily confirm that
        $$
        \mathcal{L}J+ u(c)-yc=0.
        $$
        \item The region {\bf IR}:\\
        We deduce that
        \begin{eqnarray*}
            \begin{split}
                {\bf IR}&=\{(y,c)\in\mathcal{R} \mid J_c(y,c)=\alpha u'(c)\}\\
                &=\{(y,c)\in\mathcal{R} \mid \frac{y}{u'(c)}\le b_{\alpha} \}.\\
            \end{split}
        \end{eqnarray*}
        Clearly,
        $$
        -J_c(y,c)-\beta u'(c)=-(\alpha+\beta)u'(c)<0.
        $$
        Since $J_y(y,c)=J_y(y,I(\frac{y}{b_{\alpha}}))$ and  $J_{yy}(y,c)=J_{yy}(y,I(\frac{y}{b_{\alpha}}))$ on {\bf IR},
        \begin{eqnarray}
        \begin{split}
        &\mathcal{L}J(y,c)+u(c)-yc\\
        =&\underbrace{\left(\mathcal{L}J(y,I(\frac{y}{b_{\alpha}}))+u(I(\frac{y}{b_{\alpha}}))-yI(\frac{y}{b_{\alpha}})\right)}_{=0}+\delta J(y,I(\frac{y}{b_{\alpha}}))-\delta J(y,c)+u(c)-yc\\ -&(u(I(\frac{y}{b_{\alpha}}))-yI(\frac{y}{b_{\alpha}}))\\
        =&\int_{c}^{I(\frac{y}{b_{\alpha}})}\left(\delta J_c(y,\eta)-(u'(\eta)-y)\right)d\eta\\
        =&\int_{c}^{I(\frac{y}{b_{\alpha}})}u'(\eta)\left(\frac{y}{u'(\eta)}-(1-\delta\alpha)\right)d\eta
        \le 0 \;\;\;\left(\because \dfrac{y}{u'(\eta)}<b_{\alpha} <1-\delta \alpha \;\;\;\mbox{on {\bf IR}}\right).
        \end{split}
        \end{eqnarray}
        \item The region {\bf DR}:\\
        Similarly,
        \begin{eqnarray*}
            \begin{split}
                {\bf DR}&=\{(y,c)\in\mathcal{R} \mid J_c(y,c)=\beta u'(c)\}\\
                &=\{(y,c)\in\mathcal{R} \mid \frac{y}{u'(c)}\ge b_{\beta} \}.\\
            \end{split}
        \end{eqnarray*}
        and
        $$
        J_c(y,c)-\alpha u'(c)=-(\alpha+\beta)u'(c)<0,
        $$
        Since $J_y(y,c)=J_y(y,I(\frac{y}{b_{\beta}}))$ and  $J_{yy}(y,c)=J_{yy}(y,I(\frac{y}{b_{\beta}}))$ on {\bf IR},
        \begin{eqnarray}
        \begin{split}
        &\mathcal{L}J(y,c)+u(c)-yc\\
        =&\underbrace{\left(\mathcal{L}J(y,I(\frac{y}{b_{\beta}}))+u(I(\frac{y}{b_{\beta}}))-yI(\frac{y}{b_{\beta}})\right)}_{=0}+\delta J(y,I(\frac{y}{b_{\alpha}}))\\-&\delta J(y,c)+u(c)-yc -(u(I(\frac{y}{b_{\alpha}}))-yI(\frac{y}{b_{\alpha}}))\\
        =&-\int_{I(\frac{y}{b_{\beta}})}^{c}\left(\delta J_c(y,\eta)-(u'(\eta)-y)\right)d\eta\\
        =&-\int_{I(\frac{y}{b_{\beta}})}^{c}u'(\eta)\left(\frac{y}{u'(\eta)}-(1+\delta\beta)\right)d\eta
        \le 0 \;\;\;\left(\because \dfrac{y}{u'(\eta)}>b_{\beta} >1+\delta \beta \;\;\;\mbox{on {\bf DR}}\right).
        \end{split}
        \end{eqnarray}
        Thus, $J(y,c)$ satisfies the HJB-equation
        $$
        \max\{\mathcal{L}J+u(c)-yc, J_c-\alpha u'(c), -J_c -\beta u'(c)\}=0.
        $$
    \end{itemize}

\noindent \textbf{(Proof of 2.)}

    Let
    $$
    N_t = \int_{0}^{t}e^{-\delta s}(-\t y_s)J_y(y_s,c_s)dB_s.
    $$
    To show the process $N_t$ is a martingale, it is suffice to prove that
    $$
    \mathbb{E}\left[\int_{0}^{t}\left(e^{-\delta s}(-\t y_s)J_y(y_s,c_s)\right)^2dt\right]<\infty,\qquad \mbox{for}\;\forall t\ge 0.
    $$
    (see Chapter 3 in \citet{OS})

    First, we consider the case when $(y_t,c_t)\in{\bf NR}$.
    Then,
    $$
    I(\frac{y_t}{b_{\alpha}})< c_t < I(\frac{y_t}{b_{\beta}})\;\;\mbox{or}\;\;b_{\alpha}<\left(\dfrac{y_t}{u'(c_t)}\right)<b_{\beta}.
    $$

    Since
    $$
    yJ_y(y,c)=\dfrac{D_1m_1 yc}{(1-\gamma+\gamma m_1)b_{\alpha}}\left(\dfrac{y}{u'(c)b_{\alpha}}\right)^{m_1-1}+\dfrac{D_2m_2 yc}{(1-\gamma+\gamma m_2)b_{\alpha}}\left(\dfrac{y}{u'(c)b_{\alpha}}\right)^{m_2-1}-\frac{yc}{r},
    $$
    there exist constants $K_{11}, K_{12}>0$ such that
    \begin{eqnarray}\label{eq:estimate_1}
    \begin{split}
    \left|y_tJ_y(y_t, c_t) \right|&\le K_{11} y_t c_t \le K_{12} (y_t)^{-\frac{1-\gamma}{\gamma}}.
    \end{split}
    \end{eqnarray}
    When $(y_t,c_t)\in{\bf IR}$, we know that
    \begin{eqnarray*}
        \begin{split}
            J_y(y_t, c_t) = &J_y(y_t, I(\frac{y_t}{b_{\alpha}})),\;\;\;\mbox{if}\;\;(y_t,c_t)\in{\bf IR}.
        \end{split}
    \end{eqnarray*}
    In this case,
    $$yJ_y(y,I(\frac{y}{b_{\alpha}}))=\dfrac{D_1m_1 yI(\frac{y}{b_{\alpha}})}{(1-\gamma+\gamma m_1)b_{\alpha}}+\dfrac{D_2m_2 yI(\frac{y}{b_{\alpha}})}{(1-\gamma+\gamma m_2)b_{\alpha}}-\frac{yI(\frac{y}{b_{\alpha}})}{r},
    $$
    Thus, there exist constants $K_{21}, K_{22}>0$ such that
    \begin{eqnarray}\label{eq:estimate_2}
    \begin{split}
    \left|y_tJ_y(y_t, c_t) \right|&\le K_{21} y_t I(\frac{y_t}{b_{\alpha}}) \le K_{22} (y_t)^{-\frac{1-\gamma}{\gamma}}.
    \end{split}
    \end{eqnarray}
    Similarly, when $(y_t,c_t)\in{\bf DR}$, there exist constants $K_{31},K_{32}>0$ such that
    \begin{eqnarray}
    \begin{split}\label{eq:estimate_3}
    \left|y_tJ_y(y_t, c_t) \right|&\le K_{31} y_t I(\frac{y_t}{b_{\beta}}) \le K_{32} (y_t)^{-\frac{1-\gamma}{\gamma}}.
    \end{split}
    \end{eqnarray}
    By \eqref{eq:estimate_1}, \eqref{eq:estimate_2} and \eqref{eq:estimate_3}, for any $(y_t,c_t)\in \mathcal{R}$,
    $$
    \left|y_tJ_y(y_t, c_t) \right| \bar K_4 (y_t)^{-\frac{1-\gamma}{\gamma}}.
    $$
    for some constant $K_4>0$.

    Hence
    \begin{eqnarray}
    \begin{split}
    \mathbb{E}\left[\int_{0}^{t}\left(e^{-\delta s}(-\t y_s)J_y(y_s,c_s)\right)^2dt\right] \le K_4\mathbb{E}\left[e^{-\delta t}(y_t^{-\frac{1-\gamma}{\gamma}})^2\right]
    =  K_4 \int_{0}^{t} e^{-(K-\frac{\t^2}{2}(\frac{\gamma-1}{\gamma})^2)s}ds <\infty.
    \end{split}
    \end{eqnarray}
    This implies that $K_t$ is a martingale for $t\ge 0$. \\

\noindent \textbf{(Proof of 3.)}

    If $(y_t, c_t)\in{\bf NR}$, then
    $$
    b_{\alpha} <\left(\dfrac{y_t}{u'(c_t)}\right)<b_{\beta}.
    $$
    Since,
    \begin{eqnarray}
    \begin{split}
    |J(y,c)|=&\left|D_1 \dfrac{yc}{(1-\gamma+\gamma m_1)b_{\alpha}}\left(\dfrac{y}{c^{-\gamma}b_{\alpha}}\right)^{m_1-1} +D_2 \dfrac{yc}{(1-\gamma+\gamma m_2)b_{\beta}}\left(\dfrac{y}{c^{-\gamma}b_{\beta}}\right)^{m_2-1} +\dfrac{1}{\delta}\dfrac{c^{1-\gamma}}{1-\gamma}-\dfrac{yc}{r}\right|
    \end{split}
    \end{eqnarray}
    there exists a constant $K_{51}>0$ such that
    \begin{equation}\label{eq:estimante_J1}
    |J(y_t,c_t)|\le K_{51}(y_t)^{-\frac{1-\gamma}{\gamma}}.
    \end{equation}
    When $(y_t,c_t)\in{\bf IR}$, we know that
    \begin{equation}
    J(y_t,c_t)=J\left(y_t,I(\dfrac{y_t}{b_{\alpha}})\right)+\alpha\left(u(c_t)-u(I(\dfrac{y_t}{b_{\alpha}}))\right).
    \end{equation}
    Thus, there exits a constant $K_{52}$ such that
    \begin{equation}\label{eq:estimante_J2}
    |J(y_t,c_t)|\le K_{52}(y_t)^{-\frac{1-\gamma}{\gamma}}+\alpha|u(c_t)|.
    \end{equation}
    Similarly, for $(y_t,c_t)\in {\bf DR}$, there exists a constant $K_{53}$ such that
    \begin{equation}\label{eq:estimante_J3}
    |J(y_t,c_t)|\le K_{53}(y_t)^{-\frac{1-\gamma}{\gamma}}+\beta|u(c_t)|.
    \end{equation}
    By \eqref{eq:estimante_J1}, \eqref{eq:estimante_J2} and \eqref{eq:estimante_J3}, for any $(y_t,c_t)\in\mathcal{R}$, there exist a constant $K_{54}>0$ such that
    \begin{eqnarray}\label{eq:estimante_J4}
    |J(y_t,c_t)| \le K_{54}\left((y_t)^{-\frac{1-\gamma}{\gamma}}+ |u(c_t)|\right).
    \end{eqnarray}
    By the admissibility of the consumption strategy,
    $$
    \mathbb{E}\left[\int_{0}^{\infty}e^{-\delta t}|u(c_t)|dt\right]<+\infty.
    $$
    Since $\{c_t\}_{t=0}^{\infty}$ is a finite variation process, its sample paths can have at most countable discontinuities. Hence, applying Fubini's theorem, we deduce that
    $$
    \int_{0}^{\infty}\mathbb{E}\left[e^{-\delta t}|u(c_t)|\right]dt=\mathbb{E}\left[\int_{0}^{\infty}e^{-\delta t}|u(c_t)|dt\right]<+\infty.
    $$
    and
    $$
    \lim_{t\to \infty}\mathbb{E}\left[e^{-\delta t}|u(c_t)|\right]=0.
    $$
    From \eqref{eq:estimante_J4},
    \begin{eqnarray}
    \begin{split}
    \lim_{t \to \infty}\mathbb{E}\left[e^{-\delta t}|J(y_t,c_t)|\right]\le&K_{54}\left(\lim_{t \to \infty}\mathbb{E}\left[e^{-\delta t}(y_t)^{-\frac{1-\gamma}{\gamma}}\right]+\lim_{t \to \infty}\mathbb{E}\left[e^{-\delta t}|u(c_t)|dt\right]\right)=0.
    \end{split}
    \end{eqnarray}
    Thus, we can conclude that for any admissible consumption strategy $(c^+,c^{-})$ and its associated consumption process $c$,
    $$
    \lim_{t \to \infty}\mathbb{E}\left[e^{-\delta t}J(y_t,c_t)\right]=0.
    $$

\end{proof}

\section{Proof of Theorem \ref{thm:duality}}\label{sec:Append:B}

 We will show that the duality relationship in the following steps:\\

\noindent{\bf (Step 1)} First, we will prove that the dual value function $J(y,c)$ is strictly convex in $y$:\\

By direct computation,
\begin{eqnarray}
\begin{split}
y\dfrac{\partial^2 J}{\partial y^2}=&c\left(\dfrac{D_1 m_1(m_1-1)}{(1-\gamma+\gamma m_1)b_{\alpha}}\left(\dfrac{y}{ {c}^{-\gamma} b_{\alpha}}\right)^{m_1-1}+\dfrac{D_2 m_2(m_2-1)}{(1-\gamma+\gamma m_2)b_{\alpha}}\left(\dfrac{y}{ {c}^{-\gamma} b_{\alpha}}\right)^{m_2-1}\right).
\end{split}
\end{eqnarray}
Since
$$
D_1>0,\;D_2<0,\;1-\gamma+\gamma m_1>0,\;\mbox{and}\;1-\gamma +\gamma m_2 <0,
$$
we deuce that
$$
\dfrac{\partial^2 J}{\partial y^2} > 0\;\;\mbox{for}\;\;y>0,
$$
and thus $J(y,c)$ is strictly convex in $y$.

Let us denote the Lagrangian {\bf L} defined in \eqref{eq:Lagrangian} by ${\bf L}(y,c)$ for the Lagrangian multiplier $y$ and consumption profile $c$.
\\

\noindent{\bf (Step 2)} We will show that there exist a unique solution $y^*>0$ such that $y^*$ and the optimal consumption $\{c_t^*\}_{t=0}^\infty$ maximize the Lagrangian ${\bf L}(y,c)$.

From Proposition \ref{pro:solution_dual}, we deduce that
\begin{footnotesize}
	\begin{eqnarray}
	\begin{split}
	\dfrac{\partial J}{\partial y}(y,c)=
	\begin{cases}
	&\dfrac{c}{r}-\left(\dfrac{D_1 m_1c}{(1-\gamma+\gamma m_1)b_{\alpha}}\left(\dfrac{y}{c^{-\gamma}b_{\alpha}}\right)^{m_1-1} + \dfrac{D_2 m_2c}{(1-\gamma+\gamma m_2)b_{\alpha}}\left(\dfrac{y}{c^{-\gamma}b_{\alpha}}\right)^{m_2-1}\right),\\ \\ &\qquad\qquad\qquad\qquad\;\;\;\;\;\mbox{for}\;\;(y,c)\in {\bf NR},\\ \\
	\vspace{2mm}
	&\dfrac{\partial J}{\partial y}\left(y,I(\dfrac{y}{b_{\alpha}})\right),\qquad\;\mbox{for}\;\;(y,c)\in{\bf IR},\\ \\
	\vspace{2mm}
	&\dfrac{\partial J}{\partial y}\left(y,I(\dfrac{y}{b_{\beta}})\right),\qquad\;\mbox{for}\;\;(y,c)\in{\bf DR}.
	\end{cases}
	\end{split}
	\end{eqnarray}
\end{footnotesize}
For a sufficiently small $y>0$,
$$
(y,c)\in {\bf IR},
$$
and for a sufficiently large $y>0$,
$$
(y,c)\in {\bf DR}.
$$
This implies that
\begin{eqnarray}
\begin{split}
&\lim_{y\to 0} \dfrac{\partial J}{\partial y}(y,c)=\lim_{y\to 0}\dfrac{\partial J}{\partial y}\left(y,I(\dfrac{y}{b_{\alpha}})\right) = +\infty, \\
&\lim_{y\to \infty} \dfrac{\partial J}{\partial y}(y,c)=\lim_{y\to \infty}\dfrac{\partial J}{\partial y}\left(y,I(\dfrac{y}{b_{\beta}})\right) = 0.
\end{split}
\end{eqnarray}
Since $J(y,c)$ is strictly convex in $y$, for given $X>0$, there exists a unique $y^*$ such that
\begin{eqnarray}
X  = -\dfrac{\partial J}{\partial y}(y^*,c).
\end{eqnarray}
Thus,  there exist optimal consumption strategy $({c}^{*,+}, {c}^{*,-})$ such that
\begin{eqnarray}
\begin{split}
J(y^*,c)=\mathbb{E}\left[\int_{0}^{\infty}e^{-\delta t}\left(h(y_t^*,{c}_t^*)dt-\alpha u'({c}_t^*)d {c}_t^{*,+}-\beta u'({c}_t^*)d{c}_t^{*,-}\right)\right],
\end{split}
\end{eqnarray}
where $y_t^*=y^* e^{\beta t}H_t$, $({c}^{*,+}, {c}^{*,-})\in \Pi(c)$ and ${c}_t^* = c+{c}_t^{*,+} -{c}_t^{*,-}$.

This means that since the Lagrangian is concave, $y^*$ and $c^*$ are maximizers of the Lagrangian.\\

\noindent{\bf (Step 3)} $c^*$ satisfies the budget constraint with equality. \\

Define $y^{+h}=y^*+h$  and $y^{-h}=y^*-h$ with $y^*\geq h>0$ (For convenience of notation we drop the time subscript $t$). Then,
$$
\mathcal{L}(c^*,y^{\pm h})\geq\mathcal{L}(c^*,y^*).
$$

{Since $c^*,y^*$ maximizers of the Lagrangian ${\bf L}$, we have
	$$
	\limsup_{h \downarrow 0}  \frac{\mathcal{L}(c^*,y^{+h})-\mathcal{L}(c^*,y^*)}{h}\le 0, \,\, \liminf_{h \uparrow 0}  \frac{\mathcal{L}(c^*,y^{-h})-\mathcal{L}(c^*,y^*)}{h}\le 0,
	$$
    and thus we deduce
	\begin{align*}\pm\left(X- \mathbb{E}\left[\int_0^\infty H_t c_t^*dt\right]\right)\leq 0.
	\end{align*}
}

This leads to
\begin{equation*}\label{budget1}
X=\mathbb{E}\left[\int_0^\infty H_t c_t^*dt \right]
\end{equation*}

This implies that  $c^*$ satisfies the budget constraint with equality.\\

\noindent {\bf (Step 4)} {$c^*$ is optimal consumption.}
\bigskip

Let $(c^+,c^-) \in \Pi(c)$ be a feasible consumption strategy, i.e., it is admissible and satisfies the budget constraint.

Since $c$ satisfies the budget constraint,
\begin{eqnarray}
\begin{split}
&\mathbb{E}\left[\int_{0}^{\infty}e^{-\delta t}\left(u(c_t)dt-\alpha u'(c_t)dc_t^{+}-\beta u'(c_t)dc_t^-\right)\right]\\
\leq&\mathbb{E}\left[\int_{0}^{\infty}e^{-\delta t}\left(u(c_t)dt-\alpha u'(c_t)dc_t^{+}-\beta u'(c_t)dc_t^-\right)\right]+y^*\left(X-\mathbb{E}\left[\int_{0}^{\infty}H_t c_t dt \right]\right)\\
\end{split}
\end{eqnarray}
where $y^*$ is defined in {\bf (Step 2)}.

Since  $y^*$ and $c^*$ maximize the Lagrangian ${\bf L}$ and $c^*$ satisfies the budget constraint with equality,
\begin{eqnarray}
\begin{split}
&\mathbb{E}\left[\int_{0}^{\infty}e^{-\delta t}\left(u(c_t)dt-\alpha u'(c_t)dc_t^{+}-\beta u'(c_t)dc_t^-\right)\right]\\
\leq&\mathbb{E}\left[\int_{0}^{\infty}e^{-\delta t}\left(u(c_t)dt-\alpha u'(c_t)dc_t^{+}-\beta u'(c_t)dc_t^-\right)\right]+y^*\left(X-\mathbb{E}\left[\int_{0}^{\infty}H_t c_t dt \right]\right)\\
\leq&\mathbb{E}\left[\int_{0}^{\infty}e^{-\delta t}\left(u(c_t^*)dt-\alpha u'(c_t^*)dc_t^{*,+}-\beta u'(c_t^*)dc_t^{*,-}\right)\right]+y^*\left(X-\mathbb{E}\left[\int_{0}^{\infty}H_t c_t^* dt \right]\right)\\
=&\mathbb{E}\left[\int_{0}^{\infty}e^{-\delta t}\left(u(c_t^*)dt-\alpha u'(c_t^*)dc_t^{*,+}-\beta u'(c_t^*)dc_t^{*,-}\right)\right].
\end{split}
\end{eqnarray}
Therefore,
$(c^*_t)_{t=0}^{\infty}$ is optimal.
\bigskip

\noindent{\bf Step 5)} Proof of duality-relationship in \eqref{eq:dualityrelationship}
\bigskip

Since $c^*$ is optimal consumption, for $y>0$, we deduce
\begin{align*}
V(X,c) &= \mathbb{E}\left[\int_{0}^{\infty}e^{-\delta t}\left(u(c_t^*)dt-\alpha u'(c_t^*)dc_t^{*,+}-\beta u'(c_t^*)dc_t^{*,-}\right)\right] \\
&  = \mathbb{E}\left[\int_{0}^{\infty}e^{-\delta t}\left(u(c_t^*)dt-\alpha u'(c_t^*)dc_t^{*,+}-\beta u'(c_t^*)dc_t^{*,-}\right)\right]+y\left( X-\mathbb{E}\left[\int_{0}^{\infty}H_t c^*_t dt \right]\right)\nonumber\\
&\leq  \sup_{(c^+,c^-)\in \Pi(c)} \mathbb{E} \left[\int_0^\infty e^{-\delta t} u(c_t)dt\right]+y\left( X-\mathbb{E}\left[\int_{0}^{\infty}H_t c_t dt \right]\right)\nonumber\\
&= J(y,c)+yX,
\end{align*}
where $(c^t)_{t=}^{\infty}$ is the optimal consumption process for Problem \ref{pr:dual_problem} for $y>0$. This implies that
$$
V(X,c)\leq \inf_{y>0} \Big(J(y,c)+yX\Big).
$$
However, we know that
\begin{align*}
V(X,c)&=  \mathbb{E}\left[\int_{0}^{\infty}e^{-\delta t}\left(u(c_t^*)dt-\alpha u'(c_t^*)dc_t^{*,+}-\beta u'(c_t^*)dc_t^{*,-}\right)\right] \\
&=J(y^*,c)+y^* X.
\end{align*}

Thus,
$$
V(X,c)= \min_{y>0} \Big(J(y,c)+yX\Big).
$$
This completes the proof.

\section{Proof of Proposition \ref{pro:consumption}}\label{sec:Append:C}

First, we show that the optimal consumption strategy $(c^{*,+}, c^{*,-})$ given in \eqref{eq:optimal_consumption} is admissible.
We can see that for optimal consumption strategy $(c^{*,+},c^{*,-})$ and its associated consumption $c^{*}$,
\begin{eqnarray}\begin{split}\label{eq:estimate_4}
b_{\alpha} \le \dfrac{y_t}{(c_t^{*})^{-\gamma}} \le b_{\beta}\;\;\;\mbox{for}\;\forall\;t\ge 0.
\end{split}\end{eqnarray}
Then, there exist constants $K_{61},K_{62}>0$ such that
\begin{eqnarray}\label{eq:estimate_5}
|u(c_t^*)|\le K_{61}(y_t)^{-\frac{1-\gamma}{\gamma}}\;\;\;\mbox{and}\;\;\;y_t c_t^* \le K_{62} y_t^{-\frac{1-\gamma}{\gamma}},
\end{eqnarray}
for all $t \ge 0$.

Thus,
$$
\mathbb{E}\left[\int_{0}^{\infty}e^{-\delta t}|u(c_t^*)|dt\right] \le K_{61}\mathbb{E}\left[\int_{0}^{\infty}e^{-\delta t}(y_t)^{-\frac{1-\gamma}{\gamma}}dt\right]<+\infty.
$$
Since $(c_t^{*,+},c_t^{*,-})$ is the optimal strategy, it is clear that
$$
\mathbb{E}\left[\int_{0}^\infty e^{-\delta t}\left(\alpha u'(c_t^*)dc_{t}^{*,+}+\beta u'(c_t^*)dc_{t}^{*,-}\right)\right]<+\infty.
$$
Above two inequalities imply $(c_t^{*,+},c_t^{*,-})$ is admissible consumption strategy.

Moreover, by \eqref{eq:estimate_4} and \eqref{eq:estimate_5}, there exists a constant $K_{64}>0$ such that
\begin{eqnarray}
\begin{split}
|J(y_t,c_t^{*})|=&\left|D_1 \dfrac{y_t c_t^{*}}{(1-\gamma+\gamma m_1)b_{\alpha}}\left(\dfrac{y_t}{(c_t^{*})^{-\gamma}b_{\alpha}}\right)^{m_1-1} +D_2 \dfrac{y_tc_t^{*}}{(1-\gamma+\gamma m_2)b_{\beta}}\left(\dfrac{y_t}{(c_t^{*})^{-\gamma}b_{\beta}}\right)^{m_2-1} \right.\\+&\left.\dfrac{1}{\delta}\dfrac{(c_t^{*})^{1-\gamma}}{1-\gamma}-\dfrac{y_tc_t^{*}}{r}\right|\\
\le& K_{64} (y_t)^{-\frac{1-\gamma}{\gamma}}
\end{split}
\end{eqnarray}
This implies
\begin{eqnarray}
\begin{split}
\lim_{t \to \infty}\mathbb{E}\left[e^{-\delta t}|J(y_t,c_t^{*}) |\right]\le& K_{64} \lim_{t \to \infty}\mathbb{E} \mathbb{E}\left[e^{-\delta t}(y_t)^{-\frac{1-\gamma}{\gamma}}\right]\\
=&K_{64} \lim_{t \to \infty} e^{-Kt} =0.
\end{split}
\end{eqnarray}

From the construction of the optimal consumption strategy, it is easy to check that the consumption strategy $(c^{*,+},c^{*,-})$ given in \eqref{eq:optimal_consumption} satisfies the following assumption in Theorem \ref{thm:verification}:
$$
(y_t,c^{*}_t)\in\left\{(y,c)\in\mathcal{R}: \mathcal{L}J(y,c)+h(y,c)=0\right\},
$$
Lebesgue-a.e., $\mathbb{P}$-a.s.,
\begin{eqnarray}
\begin{split}
&\int_{0}^{t}e^{-\delta s}\left(J_c(y_s,c_s^*)-\alpha u'(c_s^*)\right)dc_{s}^{*,+}=0,\;\;\;\mbox{for all}\;t\ge 0,\;\mathbb{P}-a.s.,\\
&\int_{0}^{t}e^{-\delta s}\left(-J_c(y_s,c_s^*)-\beta u'(c_s^*)\right)dc_{s}^{*,-}=0,\;\;\;\mbox{for all}\;t\ge 0,\;\mathbb{P}-a.s.
\end{split}
\end{eqnarray}

\section{Proof of Theorem \ref{thm:wealth}}\label{sec:Append:D}

From Theorem \ref{thm:duality}, we know that there exists a unique solution $y^*$ for the minimization problem \eqref{eq:dualityrelationship}. The first-order condition implies that
\begin{eqnarray}
\begin{split}
X=&-\dfrac{\partial J}{\partial y}(y^*,c)\\
=&\dfrac{c}{r}-\left(\dfrac{D_1 m_1c}{(1-\gamma+\gamma m_1)b_{\alpha}}\left(\dfrac{y}{c^{-\gamma}b_{\alpha}}\right)^{m_1-1} + \dfrac{D_2 m_2c}{(1-\gamma+\gamma m_2)b_{\alpha}}\left(\dfrac{y}{c^{-\gamma}b_{\alpha}}\right)^{m_2-1}\right).
\end{split}
\end{eqnarray}
Since Problem \ref{pr:dual_problem} is time-consistent, $y_s^*=y^*e^{\delta s}H_s$ is the minimizer for the duality relationship starting at $s\ge 0$. Thus, for optimal wealth $X_s^*$ at time $s$, we have
\begin{eqnarray}
\begin{split}
X_s^*= \dfrac{c^*_s}{r}-c_s^*\left(\dfrac{D_1
	m_1}{(1-\gamma+\gamma m_1)b_{\alpha}}\left(\dfrac{y_s^*}{
	{(c_s^{*})}^{-\gamma} b_{\alpha}}\right)^{m_1-1}+\dfrac{D_2
	m_2}{(1-\gamma+\gamma m_2)b_{\alpha}}\left(\dfrac{y_s^*}{
	{(c_s^{*})}^{-\gamma} b_{\alpha}}\right)^{m_2-1}\right).
\end{split}
\end{eqnarray}
During the time in which $y_t^*$ is inside the {\bf NR}, the optimal consumption $c^*$ is constant, i.e., $c_t^*=c_s^*$ and thus the optimal wealth $X_t^*$ for $t \ge s$  is given by
\begin{eqnarray}
\begin{split}
X_t^*= \dfrac{c^*_s}{r}-c_s^*\left(\dfrac{D_1
	m_1}{(1-\gamma+\gamma m_1)b_{\alpha}}\left(\dfrac{y_t^*}{
	{(c_s^{*})}^{-\gamma} b_{\alpha}}\right)^{m_1-1}+\dfrac{D_2
	m_2}{(1-\gamma+\gamma m_2)b_{\alpha}}\left(\dfrac{y_t^*}{
	{(c_s^{*})}^{-\gamma} b_{\alpha}}\right)^{m_2-1}\right).
\end{split}
\end{eqnarray}
By Proposition \ref{pro:consumption}, we know that the agent does not increase or decrease his/her consumption in the region
$$
b_{\alpha} < \dfrac{y_t^*}{(c_t^*)^{-\gamma}} < b_{\beta}.
$$

Let us define $\underline{x}$, $\bar{x}$ as follows:
$$
\underline{x}=\mathcal{X}(b_{\beta}),\;\bar{x}=\mathcal{X}(b_{\alpha}),
$$
where $\mathcal{X}(y)$ is
$$
\mathcal{X}(y)=\dfrac{1}{r}-\left(\dfrac{D_1
	m_1}{(1-\gamma+\gamma m_1)b_{\alpha}}\left(\dfrac{y}{
	b_{\alpha}}\right)^{m_1-1}+\dfrac{D_2
	m_2}{(1-\gamma+\gamma m_2)b_{\alpha}}\left(\dfrac{y}{
	b_{\alpha}}\right)^{m_2-1}\right).
$$

Since
$$
\dfrac{\partial \mathcal{X}}{\partial y}(y)=-\left(\dfrac{D_1
	m_1(m_1-1)}{(1-\gamma+\gamma m_1)b_{\alpha}^2}\left(\dfrac{y}{
	b_{\alpha}}\right)^{m_1-2}+\dfrac{D_2
	m_2(m_2-1)}{(1-\gamma+\gamma m_2)b_{\alpha}^2}\left(\dfrac{y}{
	b_{\alpha}}\right)^{m_2-2}\right)
$$
and $D_1>0,\;D_2<0,\;1-\gamma+\gamma m_1>0,\;\mbox{and}\;1-\gamma +\gamma m_2 <0$,  $\mathcal{X}(y)$ is strictly increasing function of $y$.

This means that the consumption stays for $t\ge s$ constant if and only if
$$
c_s^*\underline{x} < X_t^* < c_s^*\bar{x}.
$$
This completes the proof.

\section{Proof of Proposition \ref{pro:portfolio}}\label{sec:Append:E}

\noindent Proof of (a).

By applying the generalized It\'{o}'s lemma(see \citet{Harrison}) to the optimal wealth process $X_t^{*}$,
\begin{eqnarray}
\begin{split}\label{eq:general_ito}
dX_t^*=&-\dfrac{\partial^2 J}{\partial y^2}(y_t^*,c_t^*)dy_t^* -\dfrac{1}{2}\dfrac{\partial^3 J}{\partial y^3}(y_t^*,c_t^*)(dy_t^*)^2-\dfrac{\partial^2 J}{\partial y \partial c}(y_t^*,c_t^*)dc_t^{*,+}\\
&+\dfrac{\partial^2 J}{\partial y \partial c}(y_t^*,c_t^*)dc_t^{*,-}.
\end{split}
\end{eqnarray}
If $(y_t^*,c_t^*)\in {\bf NR}$, the agent does not adjust his/her consumption. This means that $dc_t^{*,+}=dc_t^{*,-}=0$ and thus
$$
\dfrac{\partial^2 J}{\partial y \partial c}(y_t^*,c_t^*)dc_t^{*,+}=\dfrac{\partial^2 J}{\partial y \partial c}(y_t^*,c_t^*)dc_t^{*,-}=0.
$$
If $(y_t^*, c_t^*) \in {\bf IR}$, the agent should increase his/her consumption. This implies that
$$
\dfrac{\partial J}{\partial c}(y_t^*, c_t^*)=\alpha u'(c_t^*),\;\;dc_t^{*,-}=0,
$$
and hence
$$
\dfrac{\partial^2 J}{\partial y \partial c}(y_t^*,c_t^*)dc_t^{*,+}=\dfrac{\partial^2 J}{\partial y \partial c}(y_t^*,c_t^*)dc_t^{*,-}=0.
$$
Similarly, we also obtain
$$
\dfrac{\partial^2 J}{\partial y \partial c}(y_t^*,c_t^*)dc_t^{*,+}=\dfrac{\partial^2 J}{\partial y \partial c}(y_t^*,c_t^*)dc_t^{*,-}=0.
$$
when $(y_t^*, c_t^*) \in {\bf NR}$.

Therefore, by comparing the equation \eqref{eq:general_ito} with the wealth dynamics \eqref{eq:wealth} and using the fact that $dy_t^*=(\delta-r)y_t^*dt -\theta y_t^* dB_t$, we deduce the optimal portfolio policy $\pi_t^*$ as follows:
\begin{eqnarray}
\pi_t^*=\dfrac{\theta}{\sigma}y_t^* \dfrac{\partial^2 J}{\partial y^2}(y_t^*,c_t^*),
\end{eqnarray}
and
\begin{equation} \label{opt-portfolio2}
\pi_t^*=\dfrac{\t }{\sigma}c_s^*\left(\dfrac{D_1
	m_1(m_1-1)}{(1-\gamma+\gamma m_1)b_{\alpha}}\left(\dfrac{y_t^*}{
	{(c_t^{*})}^{-\gamma} b_{\alpha}}\right)^{m_1-1}+\dfrac{D_2
	m_2(m_2-1)}{(1-\gamma+\gamma m_2)b_{\alpha}}\left(\dfrac{y_t^*}{
	{(c_t^{*})}^{-\gamma} b_{\alpha}}\right)^{m_2-1}\right).
\end{equation}
\noindent Proof of (b).

\section{Proof of Theorem \ref{thm:RCRRA}}\label{sec:Appen_F}

For $c_s^* \underline{x}<X_t^*<c_s^* \bar{x}$, for $t\ge s$, the consumption stays constant. Thus, for simplicity, we can assume $c_t^*=1$.

By Theorem \ref{thm:wealth} and Proposition \ref{pro:portfolio}, we deduce that
\begin{eqnarray}
\begin{split}
X_t^*(y_t^*)-\dfrac{\gamma\sigma}{\theta}\pi_t^*(y_t^*)=-H'(y_t^*),
\end{split}
\end{eqnarray}
where $H(\cdot)$ is defined in Proposition \ref{pro:solution_VI_H}.

Then,
\begin{eqnarray}
G(y_t^*)\triangleq \dfrac{\gamma\sigma}{\sigma}\dfrac{\pi_t^*(y_t^*)}{X_t^*(y_t^*)}=1+\dfrac{H'(y_t^*)}{X_t^*(y_t^*)}.
\end{eqnarray}

Since $H'(b_\alpha)=H'(b_\beta)=0$,
$$
G(b_\alpha)=G(b_\beta)=1.
$$

We will show that there exists a unique $\hat{b}\in(b_\alpha,b_\beta)$ such that
$G(\cdot)$ is a strictly decreasing function on $(b_\alpha, \hat{b})$ and strictly decreasing function on $(\hat{b},b_\beta)$.

\begin{eqnarray}
G'(y)=\dfrac{H''(y)X(y)-H'(y)X^{'}(y)}{(X(y))^2}.
\end{eqnarray}
(For convenience of notation we drop the time subscript $t$ and the optimal subscript $*$.)

Let us define the numerator of $G'(y)$ as $\bar{G}(y)$, i.e.,
\begin{eqnarray}
\begin{split}
\bar{G}(y)=H''(y)X(y)-H'(y)X'(y).
\end{split}
\end{eqnarray}
By the proof in Proposition \ref{pro:solution_VI_H}, we know that
$$
H''(b_\alpha)<0,\;H''(b_\beta)>0.
$$
Thus,
$$
\bar{G}(b_\alpha)<0\;\;\mbox{and}\;\;\bar{G}(b_\beta)>0.
$$
Since
\begin{eqnarray*}
\begin{split}
X(y)=&\dfrac{1}{r}-\dfrac{D_1
	m_1}{(1-\gamma+\gamma m_1)b_{\alpha}}\left(\dfrac{y}{
	 b_{\alpha}}\right)^{m_1-1}+\dfrac{D_2
	m_2}{(1-\gamma+\gamma m_2)b_{\alpha}}\left(\dfrac{y}{
	 b_{\alpha}}\right)^{m_2-1},\\
 H'(y)=&\dfrac{D_1m_1}{b_\alpha} \left(\dfrac{y}{b_{\alpha}}\right)^{m_1-1} + \dfrac{D_2m_2}{b_\alpha} \left(\dfrac{y}{b_{\beta}}\right)^{m_2-1}-\dfrac{1}{r},
\end{split}
\end{eqnarray*}
we have
\begin{eqnarray*}
\begin{split}
	X'(y)=&-\dfrac{D_1
		m_1(m_1-1)}{(1-\gamma+\gamma m_1)b_{\alpha}^2}\left(\dfrac{y}{
		b_{\alpha}}\right)^{m_1-2}+\dfrac{D_2
		m_2(m_2-1)}{(1-\gamma+\gamma m_2)b_{\alpha}^2}\left(\dfrac{y}{
		b_{\alpha}}\right)^{m_2-2},\\
	H'(y)=&\dfrac{D_1m_1(m_1-1)}{b_\alpha^2} \left(\dfrac{y}{b_{\alpha}}\right)^{m_1-2} + \dfrac{D_2m_2(m_2-1)}{b_\alpha^2} \left(\dfrac{y}{b_{\beta}}\right)^{m_2-2}.
\end{split}
\end{eqnarray*}
Hence,
\begin{footnotesize}
\begin{eqnarray*}
\begin{split}
\bar{G}(y)=&\dfrac{\gamma}{r}\left(\dfrac{D_1m_1(m_1-1)^2}{(1-\gamma+\gamma m_1)b_\alpha^2}\left(\dfrac{y}{
	b_{\alpha}}\right)^{m_1-2}+\dfrac{D_2m_2(m_2-1)^2}{(1-\gamma+\gamma m_2)b_\alpha^2}\left(\dfrac{y}{
	b_{\alpha}}\right)^{m_2-2}\right)\\&-\dfrac{\gamma D_1D_2 m_1m_2(m_1-m_2)^2}{(1-\gamma+\gamma m_1)(1-\gamma+\gamma m_2)b_\alpha^3}\left(\dfrac{y}{
	b_{\alpha}}\right)^{m_1+m_2-3}\\
=&\dfrac{\gamma y^{m_2-2}}{b_\alpha^2}\left(\dfrac{D_1m_1(m_1-1)^2}{r(1-\gamma+\gamma m_1)}\left(\dfrac{y}{
	b_{\alpha}}\right)^{m_1-m_2}-\dfrac{D_1D_2m_1m_2(m_1-m_2)^2}{(1-\gamma+\gamma m_1)(1-\gamma+\gamma m_2)b_\alpha}\left(\dfrac{y}{
	b_{\alpha}}\right)^{m_1-1}+\dfrac{D_2m_2(m_2-1)^2}{r(1-\gamma+\gamma m_2)}\right)\\
\triangleq&\dfrac{\gamma y^{m_2-2}}{b_\alpha^2}\mathcal{G}(y).
\end{split}
\end{eqnarray*}	
\end{footnotesize}
We know that
$$
m_1>1,\;m_2<0,\;D_1>0,\;D_2<0,\;1-\gamma+\gamma m_1>0\;\;\mbox{and}\;\;1-\gamma+\gamma m_2<0,
$$
thus, $\mathcal{G}(y)$ is a strictly increasing function of $y$.

Since $\bar{G}(b_\alpha)<0,\;\bar{G}(b_\beta)>0$, we deduce that
$$
\mathcal{G}(b_\alpha)<0,\;\mathcal{G}(b_\beta)>0.
$$
Thus, there exists a unique $\hat{b}\in(b_\alpha,b_\beta)$ such that
$$
\mathcal{G}(\hat{b})=0.
$$
This implies that
$$
G'(y) <0,\;\;\mbox{for}\;\;y\in(b_\alpha,\hat{b})\;\;\mbox{and}\;\;G'(y) >0,\;\;\mbox{for}\;\;y\in(\hat{b},b_\beta).
$$
To sum up, we conclude that RCRRA is a strictly increasing for $X\in(c\underline{x},\widehat{X})$ and a strictly decreasing for $X\in(\widehat{X}, c\bar{x})$. Moreover, RCRRA approaches $\gamma$ when $X$ approaches $c\underline{x}$ or $c\bar{x}$.
(Here, $\hat{X}=c\mathcal{X}(\hat{b})$ and $\hat{b}\in(b_\alpha,b_\beta)$ is a unique solution of the following algebraic equation:
\begin{footnotesize}
	\begin{eqnarray}
	\begin{split}
	\mathcal{G}(y)=&\dfrac{D_1m_1(m_1-1)^2}{r(1-\gamma+\gamma m_1)}\left(\dfrac{y}{
		b_{\alpha}}\right)^{m_1-m_2}-\dfrac{D_1D_2m_1m_2(m_1-m_2)^2}{(1-\gamma+\gamma m_1)(1-\gamma+\gamma m_2)b_\alpha}\left(\dfrac{y}{
		b_{\alpha}}\right)^{m_1-1}\\&+\dfrac{D_2m_2(m_2-1)^2}{r(1-\gamma+\gamma m_2)}.
	\end{split}
	\end{eqnarray}
\end{footnotesize}

\end{footnotesize}

\end{document}